\documentclass[oribibl]{llncs}

\usepackage{fullpage}
\usepackage{graphicx}
\usepackage[space]{grffile}
\usepackage{subfigure}
\usepackage[linesnumbered,ruled,vlined,english]{algorithm2e}
\usepackage{amsfonts}
\usepackage{verbatim}
\usepackage{amsmath}
\usepackage{amssymb}
\usepackage{stmaryrd}
\usepackage{epstopdf}
\usepackage{romannum}
\usepackage{amsmath}
\usepackage[normalem]{ulem}
\usepackage[numbers]{natbib}
\usepackage[usenames,dvipsnames]{color}
\usepackage{authblk}
\usepackage{natbib}
\usepackage{booktabs}
\usepackage{array}

\newcommand{\cP}{\mathcal P}
\newcommand{\cR}{\mathcal R}

\newcommand{\cT}{\mathcal T}
\newcommand{\cL}{\mathcal L}

\newcommand{\cF}{\mathcal F}

\newcommand{\cX}{\mathcal X}

\newcommand{\cN}{\mathcal N}

\newcommand{\cV}{\mathcal V}

\newcommand{\scAG}{\textsc{AG}}

\title{Computing Hybridization Networks for\\ Multiple Rooted Binary Phylogenetic Trees by\\ Maximum Acyclic Agreement Forests}

\author{Benjamin Albrecht
\thanks{
Benjamin Albrecht\\
Institut f\"ur Informatik, Ludwig-Maximilians-Universit\"at, Germany\\
Tel.: +49-89-2180-4069\\
E-mail: \email{albrecht@bio.ifi.lmu.de}}
}

\institute{Institut f\"ur Informatik, Ludwig-Maximilians-Universit\"at, Germany}
\date{\today}

\begin{document}
\maketitle
\vspace{3cm}

\pagestyle{plain}
\pagenumbering{arabic}

\textbf{Abstract} It is a known fact that, given two rooted binary phylogenetic trees, the concept of maximum acyclic agreement forests is sufficient to compute hybridization networks with minimum hybridization number. In this work, we demonstrate by first presenting an algorithm and then showing its correctness, that this concept is also sufficient in the case of multiple input trees. More precisely, we show that for computing minimum hybridization networks for multiple rooted binary phylogenetic trees on the same set of taxa it suffices to take only maximum acyclic agreement forests into account.

Moreover, this article contains a proof showing that the minimum hybridization number for a set of rooted binary phylogenetic trees on the same set of taxa can be also computed by solving subproblems referring to common clusters of the input trees.\\


\newpage

\section{Introduction}
\label{sec-intro}

The evolution of species is often described by a phylogenetic tree representing a set of speciation events. Due to reticulation events, however, a tree is often insufficient, because different genetic sequences can give rise to different phylogenetic trees. An important reticulation event, which could be prevalently discovered in plants but also in animals, is hybridization \cite{Mallet-2007}. In order to study evolution affected by hybridization, one can reconcile incongruent phylogenetic trees, which for instance have been constructed for certain species based on different genes, into a single hybridization network. Whereas phylogenetic trees only contain internal nodes of in-degree one referring to certain speciation events, hybridization networks can, additionally, contain nodes of larger in-degree representing putative hybridization events.

The problem of computing hybridization networks with minimum hybridization number is known to be \textit{NP-hard} \cite{Bordewich-2007b} but fixed-parameter tractable, even for the simplest case when only two binary phylogenetic input trees are given. In the general case, however, if the input consists of more than two trees, the problem still remains fixed-parameter tractable as recently shown by \citet{Iersel-2013}. This means, in particular, that the problem is exponential in some parameter related to the problem itself, which is the hybridization number in this case, but only polynomial to its input size, which is an important feature that facilitates the development of practical algorithms. Nevertheless, when developing an algorithm solving this computational hard problem, the challenge remains not only in guaranteeing its correctness but, especially, in providing a good practical running time. Thus, such an algorithm, on the one hand, has to be quite sophisticated and, on the other hand, has to be implemented in an efficient way, which can be achieved for instance by applying certain speedup techniques, by reducing the size of the input trees, or by running exhaustive parts of the algorithm in parallel.

Typically, a method computing hybridization networks for two rooted binary phylogenetic trees can be divided into the following two major steps. First, \emph{maximum acyclic agreement forests} are computed by cutting down the input trees in a specific way, and, second, the components of such an agreement forest are again reattached by introducing reticulation edges in a way that the resulting network displays both input trees. Broadly speaking, an \emph{agreement forests} can be seen as a set of common subparts occurring in both input trees. Moreover, in this context, the term \emph{maximum} simply denotes that there is no smaller set fulfilling the properties of an agreement forest and the \emph{acyclic} constraint denotes that it is possible, in a biological sense, to reattach its components back to a hybridization network. If this network contains a minimum number of reticulation edges, its hybridization number is minimal and, thus, this network is called a \emph{minimum hybridization network}.

In general, there exists not just one but a large number of minimum hybridization networks. To recognize putative hybridization events, biologists are interested in all of those networks, since the more frequently a hybridization event is contained in a set of possible evolutionary scenarios the more likely it is part of the true underlying evolutionary history of the considered species. Thus, given two input trees, there is a need for two types of algorithms; one for the computation of all maximum acyclic agreement forests and another one for the computation of all hybridization networks based on each of those agreement forests.

While there exist some software packages providing methods for computing hybridization networks for two rooted binary phylogenetic trees on the same set of taxa \cite{Chen-2010,Huson-2012}, in this work, we will present an algorithm computing a particular type of minimum hybridization networks, namely biologically relevant networks as defined later, for an arbitrary number of rooted binary phylogenetic trees on the same set of taxa. The workflow of this algorithm can be briefly summarized as follows. Starting with one input tree, all other input trees are embedded sequentially into a growing number of networks by adding further reticulation edges corresponding to certain components of a maximum acyclic agreement forest. In order to guarantee the computation of biologically relevant networks, it is important that each input tree is added to a so far computed network in all possible ways. This implies, in particular, that at the beginning, when adding the second tree of the ordering, say $T_2$, to the first one, say $T_1$, all biologically relevant networks embedding $T_1$ and $T_2$ have to be computed. Missing one of those networks could mean that a computational path leading to a biologically relevant network embedding the whole set of input trees is lost, and, as a consequence, the resulting output only consists of networks whose hybridization number is not minimal. A crucial observation of this work is that for this purpose it suffices to consider only maximum acyclic agreement forests.

Until now, the only software that is also able to compute minimum hybridization networks for multiple rooted binary phylogenetic trees is PIRNv2.0 \cite{Wu2010,Wu2013}. A recently conducted simulation, however, has revealed that an implementation of our algorithm provides the clearly better practical running time and, additionally, in general PIRNv2.0 does only output a small subset of all biologically relevant networks \cite{Albrecht2015}, which prohibits a significant biological interpretation of each network as discussed above.

This work is organized as follows. In a first step, the terminology that is used throughout this work is introduced. Next, in Section~\ref{sec-alg}, we give a detailed description of our algorithm \textsc{allHNetworks} whose correctness is shown in a subsequent section. Finally, we end the description of \textsc{allHNetworks} by briefly discussing its theoretical worst-case runtime and by giving some concluding remarks. In a second part, we describe some techniques improving the running time of our algorithm whereat one of those techniques is the well known cluster reduction. We finish this article by presenting a proof showing that the concept of the cluster reduction can be also applied to multiple rooted binary phylogenetic trees without having an impact on the computation of the minimum hybridization number.

\section{Preliminaries}
\label{sec-pre}

In this section, we give some preliminary definitions concerning phylogenetic trees, hybridization networks, and agreement forests following the work of \citet{Huson2011} and \citet{Scornavacca2012}, which will be first used for describing the algorithm \textsc{allHNetworks} and then for showing its correctness. We assume that the reader is familiar with general graph-theoretic concepts.\\

\textbf{Phylogenetic trees.} A \emph{rooted phylogenetic $\cX$-tree} $T$ is a tree whose edges are directed from the root to the leaves and whose nodes, except for the root, have a degree unequal to $2$. We call $T$ a \emph{binary tree} if its root has in-degree $0$ and out-degree $2$, each inner node in-degree $1$ and an out-degree $2$, and each leaf in-degree $1$ and out-degree $0$. The leaves of a rooted phylogenetic $\cX$-tree are labeled one-to-one by the taxa set $\cX$, which usually consists of certain species or genes and is denoted by $\cL(T)$. Considering a node $v$ of $T$, the label set $\cL(v)$ refers to each taxon that is contained in the subtree rooted at $v$. Given a set of trees $\cF$, the label set $\cL(\cF)$ denotes the union of each label set $\cL(F_i)$ of each tree $F_i$ in $\cF$.

Now, based on a taxa set $\cX' \subseteq \cX$, we can define a restricted subtree of a rooted phylogenetic $\cX$-tree, denoted by $T|_{\cX'}$. The restricted subtree $T|_{\cX'}$ is computed by, first, deleting each leaf repeatedly that is either unlabeled or whose taxon is not contained in $\cX'$, resulting in a subgraph denoted by $T(\cX')$, and, second, by suppressing each node of both in- and out-degree $1$. Moreover, given a tree $T$, by $\overline T$ we denote the tree that is obtained from $T$ by suppressing all nodes of both in- and out-degree $1$. The result of such a restriction is a rooted phylogenetic $\cX'$-tree.\\

\textbf{Phylogenetic networks.} A \emph{rooted phylogenetic network} $N$ on $\cX$ is a rooted connected digraph whose edges are directed from the root to the leaves as defined in the following. There is exactly one node of in-degree $0$, namely the \emph{root}, and no nodes of both in- and out-degree $1$. The set of nodes of out-degree $0$ is called the \emph{leaf set of $N$} and is labeled one-to-one by the \emph{taxa set} $\cX$, also denoted by $\cL(N)$. In contrast to a phylogenetic tree, such a network may contain undirected but not any directed cycles. Consequently, $N$ can contain nodes of in-degree larger than or equal to $2$, which are called \emph{reticulation nodes}. Moreover, each edge that is directed into such a reticulation node is called \emph{reticulation edge}.\\

{\bf Hybridization Networks.} A \emph{hybridization network} $N$ for a set $\cT$ of rooted binary phylogenetic $\cX'$-trees, with $\cX'\subseteq\cX$, is a rooted phylogenetic network on $\cX$ \emph{displaying $\cT$} (i.e., \emph{contains an embedding of each tree $T$ in $\cT$}). More precisely, this means that for each tree $T$ in $\cT$ there exists a set $E'\subseteq E(N)$ of reticulation edges \emph{referring} to $T$. More specifically, this means that $T$ can be derived from $N$ by conducting the following steps. 

\begin{enumerate}
\item[(1)] First, delete each reticulation edge from $N$ that is not contained in $E'$.
\item[(2)] Then, remove each node whose corresponding taxon is not contained in $\cX'$.
\item[(3)] Next, remove each unlabeled node of out-degree $0$ repeatedly.
\item[(4)] Finally, suppress each node of both in- and out-degree $1$.  
\end{enumerate}

From a biological point of view, this means that $N$ displays $T$ (i.e., contains an embedding of $T$) if each speciation event of $T$ is reflected by $N$. Moreover, each internal node of in-degree $1$ represents a speciation event and each internal node providing an in-degree of at least $2$ represents a reticulation event or, in terms of hybridization, a hybridization event. This means, in particular, that such a latter node represents an individual whose genome is a \emph{chimaera} of several parents. Thus, such a node $v$ of in-degree larger than or equal to $2$ is called \emph{hybridization node} (or reticulation node) and each edge directed into $v$ is called \emph{hybridization edge} (or reticulation edge). Moreover, each edge that is not a hybridization edge is called \emph{tree edge}.

Now, based on those hybridization nodes, the \emph{reticulation number} $r(N)$ of a hybridization network $N$ is defined by 
\begin{equation}
r(N)=\sum_{v\in V:\delta^-(v)>0}\left(\delta^-(v)-1\right)=|E|-|V|+1,
\label{04-eq-retNumBin}
\end{equation}
where $V$ denotes the node set and $E$ the edge set of $N$. Next, based on the definition of the reticulation number, for a set $\cT$ of phylogenetic $\cX$-trees the (minimum or exact) hybridization number $h(\cT)$ is defined by
\begin{equation}
h(\cT)=\text{min}\{r(N):\text{N is a hybridization network displaying }\cT\}.
\label{04-eq-hNumBin}
\end{equation}
Throughout this work, we call a \emph{hybridization network} $N$ for a set $\cT$ of rooted binary phylogenetic $\cX$-trees a \emph{minimum hybridization network}, if $r(N)=h(\cT)$.

Notice that the computation of the hybridization number for just two rooted binary phylogenetic $\cX$-trees is an \emph{NP-hard} problem \cite{Bordewich-2007b} which is, however, still fixed-parameter tractable \cite{Bordewich2007a}. More specifically, this means that the problem is exponential in some parameter related to the problem itself, namely the hybridization number, but only polynomial in the size of the input trees, which is an important feature facilitating the development of practical algorithms.

Lastly, given a hybridization network $N$ on $\cX$ and an edge set $E'$ referring to an embedded rooted phylogenetic $\cX'$-tree $T'$ of $N$ with $\cX'\subseteq\cX$, the \emph{restricted network} $N|_{E',\cX'}$ refers to the minimal connected subgraph $T$ only containing leaves labeled by $\cX'$ and edges that are either tree edges or contained in $E'$. Consequently, $N|_{E',\cX'}$ is a directed graph that corresponds to $T'|_{\cX'}$ but still contains nodes of both in- and out-degree~$1$, and, thus, each node in $N|_{E',\cX'}$ can be mapped back to exactly one specific node of the unrestricted network $N$ (cf.~Fig.\ref{fig-embTree}(c)). 

\begin{figure}[t]
\centering
\includegraphics[scale=0.9]{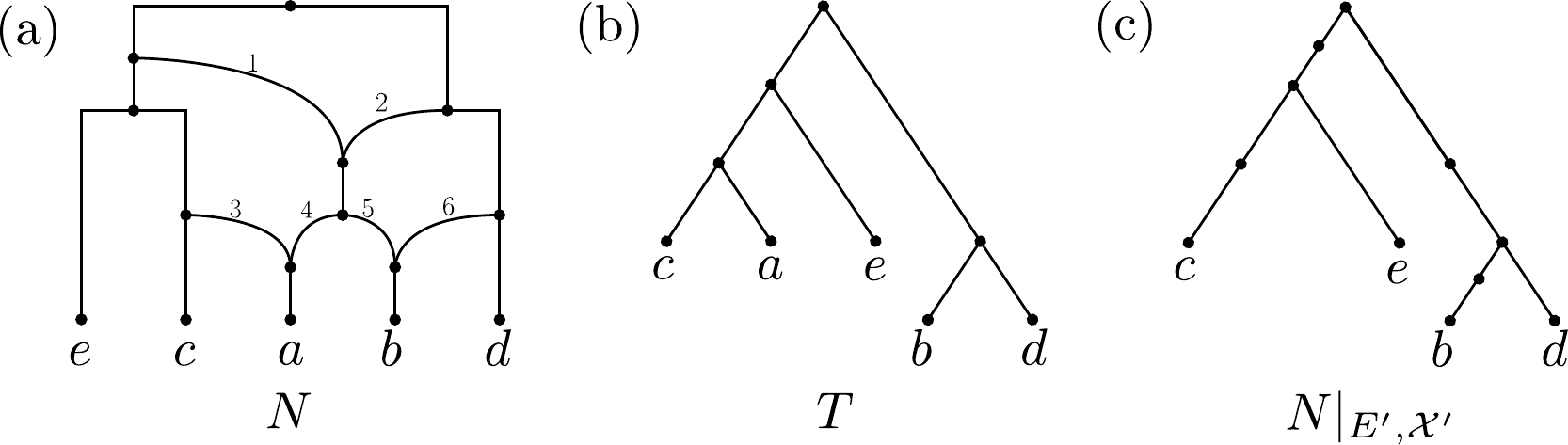}
\caption[Definitions regarding restricted networks]{\textbf{(a)} A hybridization network $N$ with taxa set $\cX=\{a,b,c,d,e\}$ whose reticulation edges are consecutively numbered. \textbf{(b)}~A~phylogenetic $\cX$-tree $T$ that is displayed by $N$. Based on $N$, both edge sets $E'=\{3,6,1\}$ and $E''=\{3,6,2\}$ refer to~$T$ and, thus, $\overline{N|_{E',\cX}}$ as well as $\overline{N|_{E'',\cX}}$ equals $T$. \textbf{(c)}~The restricted network $N|_{E',\cX'}$ with $\cX'=\{b,c,d,e\}$ still containing nodes of both in- and out-degree~$1$.} 
\label{fig-embTree}
\end{figure}

{\bf Forests.} Let $T$ be a rooted nonbinary phylogenetic $\cX$-tree $T$. Then, we call any set of rooted nonbinary phylogenetic trees $\cF=\{F_1,\dots,F_k\}$ with $\cL(\cF)=\cX$ a \emph{forest on $\cX$}, if we have for each pair of trees $F_i$ and $F_j$ that $\cL(F_i)\cap\cL(F_j)=\emptyset$. Moreover, if additionally for each component $F$ in $\cF$ the tree $T|_{\cL(F)}$ equals $F$, we say that $\cF$ is a \emph{forest for $T$}.\\

\textbf{Agreement forests.} For technical purpose, the definition of agreement forests is based on two rooted binary phylogenetic $\cX$-trees $T_1$ and $T_2$ whose roots are marked by a unique taxon $\rho\not\in\cX$ as follows. Let $r_i$ be the root of the tree $T_i$ with $i\in\{1,2\}$. Then, we first create a new node $v_i$ as well as a new leaf $\ell_i$ labeled by a new taxon $\rho\not\in\cX$ and then attach these nodes to $r_i$ by inserting the two edges $(v_i,r_i)$ and $(v_i,\ell_i)$. Notice that in this case $v_1$ and $v_2$ is the new root of $T_1$ and $T_2$, respectively. Moreover, since we consider $\rho$ as being a new taxon, the taxa set of both trees is $\cX\cup\{\rho\}$ (cf.~Fig.~\ref{fig-example-1}(a)).

Now, assuming we have given two trees $T_1$ and $T_2$ whose roots are marked by a unique taxon $\rho$, then, a \emph{binary agreement forest for $T_1$ and $T_2$} is a set of components $\cF=\{F_{\rho},F_1,\dots,F_{k-1}\}$ on $\cX \cup \{\rho\}$ satisfying the following properties.

\begin{itemize}
\item[(1)] Each component $F_i$ with taxa set $\cX_i$ equals $T_1|_{\cX_i}$ and $T_2|_{\cX_i}$. 
\item[(2)] There is exactly one component, denoted as $F_{\rho}$, with $\rho\in\cL(F_{\rho})$.
\item[(3)] Let $\cX_{\rho},\cX_1,\dots,\cX_{k-1}$ be the taxa sets of the components $F_{\rho},F_1,\dots,F_{k-1}$. All trees in $\{T_1(\cX_i)|i\in\{\rho,1,\dots,k-1\}\}$ and $\{T_2(\cX_i)|i\in\{\rho,1,\dots,k-1\}\}$ are node disjoint subtrees of $T_1$ and $T_2$, respectively (cf.~Fig.~\ref{fig-example-1}(b)).
\end{itemize}

Throughout this work, we call an agreement forest a \emph{maximum agreement forest}, if this agreement forest is of minimal size. This means, in particular, that there does not exist another set of components of smaller size satisfying the conditions of an agreement forest listed above.
 
Lastly, there is another important property an agreement forest can satisfy. We call an agreement forest $\cF$ for two rooted binary phylogenetic $\cX$-trees $T_1$ and $T_2$ \emph{acyclic}, if there is no directed cycle in the underlying \emph{ancestor-descendant graph} $AG(T_1,T_2,\cF)$, which is defined as follows. First, this graph contains one node corresponding to precisely one component of $\cF$. Moreover, two different nodes $F_i$ and $F_j$ of this graph are connected via a directed edge $(F_i,F_j)$, if,
\begin{itemize}
\item[(i)] regarding $T_1$, the root of $T_1(\cX_i)$ is an ancestor of the root of $T_1(\cX_j)$
\item[(ii)] or, regarding $T_2$, the root of $T_2(\cX_i)$ is an ancestor of the root of $T_2(\cX_j)$,
\end{itemize}
where $\cX_i\subseteq\cX$ and $\cX_j\subseteq\cX$ refers to the taxa set of the two components $F_i$ and $F_j$, respectively (cf.~Fig.~\ref{fig-example-1}(c)). Again, we call an acyclic agreement forest consisting of a minimum number of components a \emph{maximum acyclic agreement forest}. Notice that for a maximum acyclic agreement forest containing $k$ components there exists a hybridization network with hybridization number $k-1$ \cite{Baroni2005}. This means, in particular, if a maximum acyclic agreement forest for two rooted binary phylogenetic $\cX$-trees $T_1$ and $T_2$ contains only one component, $T_1$ equals $T_2$.\\

\textbf{Acyclic orderings.} Given an agreement forest for two rooted binary phylogenetic $\cX$-trees $T_1$ and $T_2$, then, if $\cF$ is acyclic and, thus, $\scAG(T_1,T_2,\cF)$ does not contain any directed cycles, one can compute an \emph{acyclic ordering} as already described in the work of Baroni \emph{et al.}~\cite{Baroni-2006}. First, select the node $v_{\rho}$ corresponding to $F_{\rho}$ of in-degree $0$ and remove $v_{\rho}$ together with all its incident edges. Next, again choose a node $v_1$ of in-degree $0$ and remove $v_1$. By continuing this way, until finally all nodes have been removed, one receives the ordering $\Pi_V=(v_{\rho},v_1,\dots,v_k)$ containing all nodes in $\scAG(T_1,T_2,\cF)$. In the following, we call the ordering $(F_{\rho}, F_1,\dots,F_k)$ of components corresponding to each node in $\Pi_V$ an \emph{acyclic ordering of $\cF$}. Notice that, as during each of those steps there can occur several nodes of in-degree $0$, especially if $\cF$ contains components consisting only of isolated nodes, such an acyclic ordering is in general not unique.\\

\begin{figure}
\centering
\includegraphics[scale=1]{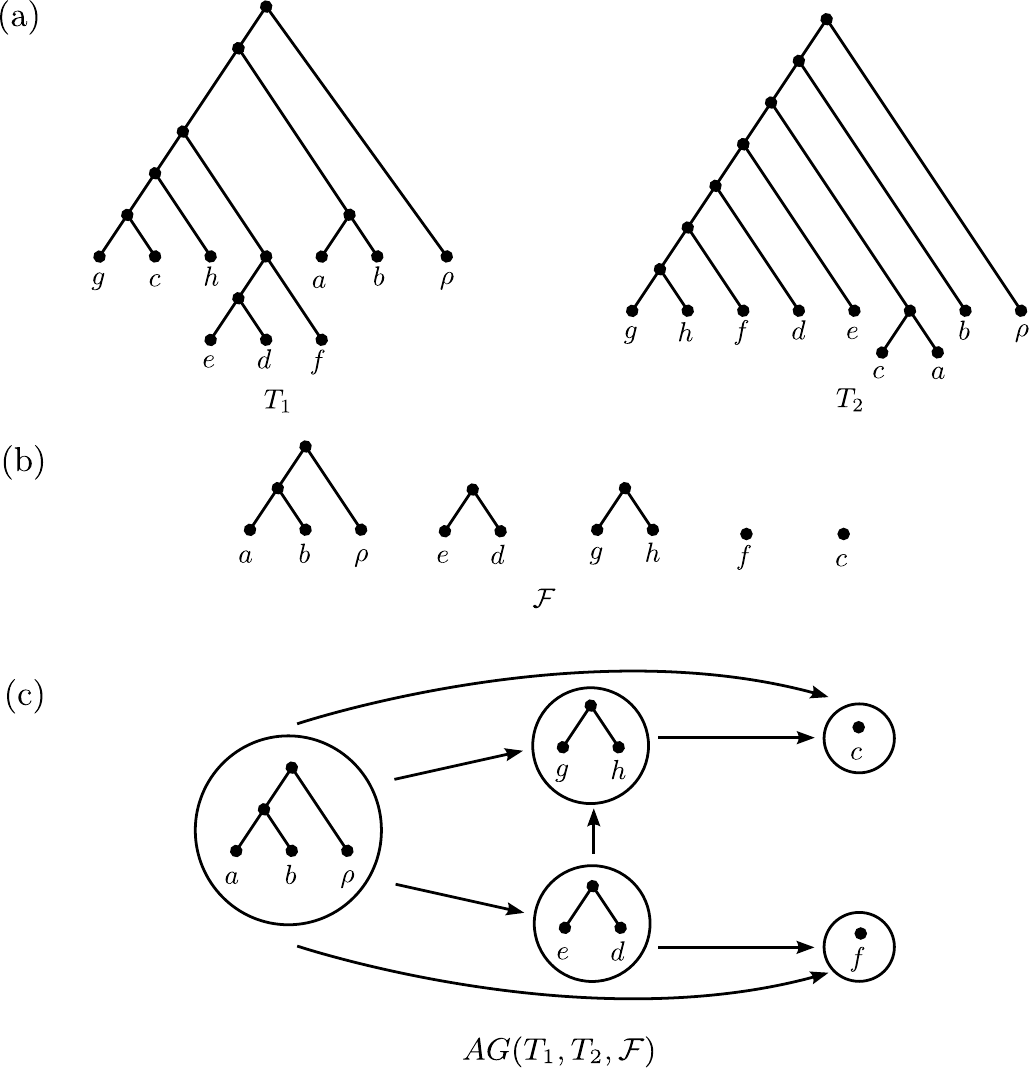}
\caption[Phylogenetic $\cX$-trees and agreement forests]{(a) Two rooted binary phylogenetic $\cX$-trees $T_1$ and $T_2$ with taxa set $\cX=\{a,b,c,d,e,f,g,h,\rho\}$. (b) An acyclic agreement forest $\cF$ for $T_1$ and $T_2$ in acyclic ordering. (c) The directed graph $AG(T_1,T_2,\cF)$ not containing any directed cycles and, thus, $\cF$ is acyclic. }
\label{fig-example-1}
\end{figure}

\textbf{Stacks of hybridization nodes.} Given a hybridization network displaying a set $\cT$ of rooted binary phylogenetic $\cX$-trees and containing a node $v$ of in-degree of at least $3$, one can generate further networks still displaying $\cT$ by dragging some of its reticulation edges upwards resulting in so-called \emph{stack of hybridization nodes}. More precisely, such a stack is a path $(v_1,\dots,v_n)$, with $n>1$, of hybridization nodes in which each node $v_i$ is connected through a reticulation edge to $v_{i+1}$ (cf.~Fig.~\ref{fig-stack}).\\ 

\begin{figure}[t]
\centering 
\includegraphics[scale=1]{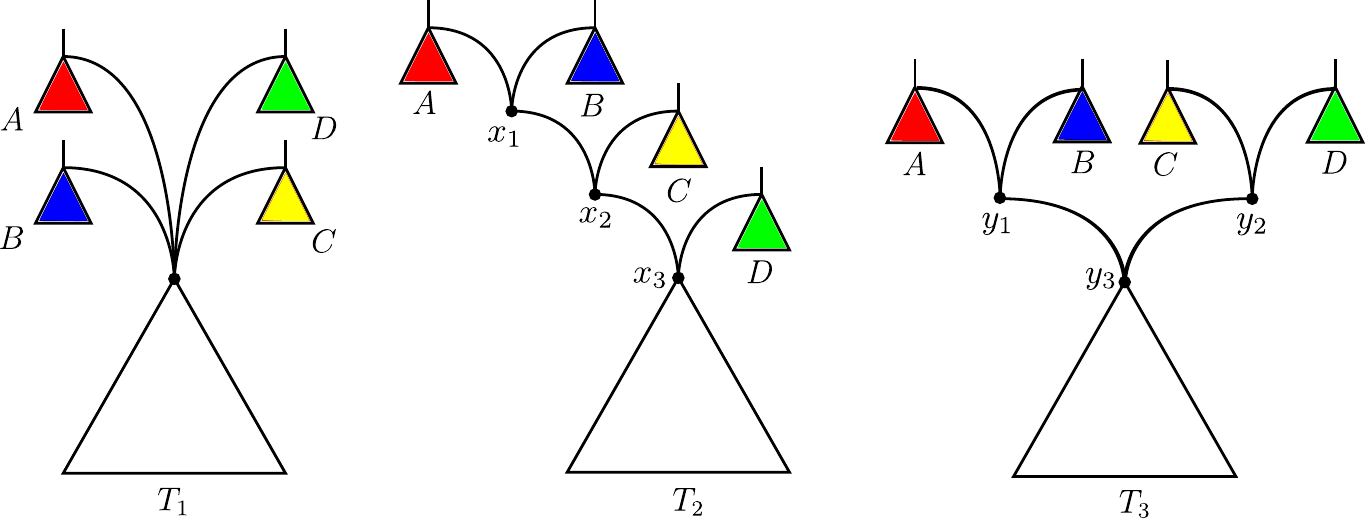}
\caption[An illustration of stacks of hybridization nodes]{An illustration of stacks of hybridization nodes. The hybridization node with in-degree $4$ of the left-hand tree $T_1$ can be resolved (amongst others) into distinctive stacks of hybridization nodes, e.g., $(x_1,x_2,x_3)$ and $(y_1,y_3)$, as demonstrated by $T_2$ and $T_3$, respectively. Notice that resolving a hybridization node into a stack of hybridization nodes does not produce new embedded trees compared with those of the unresolved network.}
\label{fig-stack}
\end{figure}

\textbf{Relevant networks.} Given a set $\cT$ of rooted phylogenetic $\cX$-trees and a phylogenetic network $N$ on $\cX$, then, we say \emph{$N$ is a relevant network for $\cT$}, if $N$ is a hybridization network displaying $\cT$ with minimum hybridization number and if $N$ does not contain any stacks of hybridization nodes. Notice that such a network leaves the interpretation of the ordering of the hybridization events adhering to a hybridization node of in-degree larger than or equal to $3$ open. 

Furthermore, we demand that each relevant network is a binary network not containing any nodes of out-degree larger than $2$. Notice that by allowing nonbinary nodes the set of relevant networks usually shrinks, since a nonbinary network can contain multiple binary networks. Moreover, in order to improve its readability, we further demand that all hybridization nodes of a relevant network have out-degree one. Notice that, in order to identify stacks of hybridization nodes, in such networks the out-edges of all hybridization nodes have to be suppressed.

Lastly, just for clarity, given two relevant networks $N_1$ and $N_2$ for a set $\cT$ of rooted phylogenetic $\cX$-trees, we say that $N_1$ equals $N_2$ if their graph topologies (disregarding the embedding of $\cT$) are isomorphic. 

\section{The Algorithm \textsc{allHNetworks}}
\label{sec-alg}

Given a set of rooted binary phylogenetic $\cX$-trees $\cT=\{T_1,\dots,T_n\}$ and a parameter $k\in \mathbb{N}$, our algorithm \textsc{allHNetworks} follows a \emph{branch-and-bound} approach conducting the following major steps. For each order of $\cT$, the trees are added sequentially to a set of networks $\cN$. In the beginning, $\cN$ consists only of one element, which is the first input tree of the ordering. By sequentially adding the other input trees to each so far computed network, the size of $\cN$ growths rapidly, since in general an input tree $T_i$ can be added to each network in $\cN$ in potential several ways. Each time the reticulation number of a so far extended network exceeds $k$, the processing of this network can be aborted. This is possible because by adding further input trees the reticulation number of the respective network is never decreased.

Given a set $\cT$ of rooted binary phylogenetic $\cX$-trees, based on two different objectives, our algorithm provides two different abort criteria: 
\begin{enumerate}
 \item \textbf{Objective:} Computation of the hybridization number of $\cT$.\\
\textbf{Abort criterion:} As soon as \emph{one} hybridization network with hybridization number $k$ is computed and each search after hybridization networks providing a hybridization number less than $k$ has failed.
 \item \textbf{Objective:} Computation of all relevant networks for $\cT$.\\
\textbf{Abort criterion:} As soon as \emph{all} hybridization networks with hybridization number $k$ are computed and each search after hybridization networks providing a hybridization number less than $k$ has failed.
\end{enumerate}

For the computation of a minimum hybridization network, parameter $k$ is set to an initial value and is increased by one if a network displaying $\cT$ with hybridization number smaller than or equal to $k$ could not be computed so far. At the beginning, $k$ can be either simply set to $0$ or to a lower bound, e.g., $$max\{R(T_i,T_j):i\ne j\}.$$ A more sophisticated method for the computation of such a lower bound is described in the work of Wu~\cite{Wu2010}. In practice, however, the lower bound does not significantly improve the runtime, since the required steps for those $k$'s that can be skipped at the beginning are usually of rather low computational complexity. 

\subsection{Inserting Trees into Networks}
\label{sec-addTree}

Given a hybridization network $N$, we say that a tree $T$ is displayed in $N$, if there exists a set of reticulation edges $E$ such that $\overline{N|_{E,\cX}}$ equals $T$ (cf.~Sec.~\ref{sec-pre}). This implies, if such a subset does not exist, we have to insert new reticulation edges for displaying $T$ in $N$. Given an edge set $E'$ referring to an embedded tree $T'$ that is already displayed in $N$, those edges can be derived from each component of an agreement forest for $T'$ and $T$. The here presented algorithm is based on the observation, that, in order to compute all relevant networks, it suffices to take only maximum acyclic agreement forests into account (cf.~Sec~\ref{sec-proof}). 

Hence, we can summarize the basic steps that are necessary for adding an input tree $T$ to a so far computed network $N$ as follows.
\begin{enumerate}
 \item Choose an edge set $E'$ referring to an embedded tree $T'$ of $N$ by selecting precisely one in-edge of each hybridization node.
 \item First compute a maximum acyclic agreement forest $\cF$ for the two trees $T'$ and $T$ and then choose an acyclic ordering $\Pi_{\cF}$ of $\cF$.
 \item Based on $\Pi_{\cF}$, for each component of $\cF$, except $F_{\rho}$, create a valid pair of source and target nodes (as defined later) such that, by connecting each node pair, $T$ is embedded in the resulting network. Notice that this step will be discussed separately in the upcoming section.
\end{enumerate}
It is easy to see, that the resulting network depends on the chosen edge set $E'$ referring to the embedded tree $T'$, which is the case because different embedded trees lead to different maximum acyclic agreement forests which consequently lead to different reticulation edges that are necessary for the embedding of $T$. Thus, to guarantee the computation of all relevant networks, all three steps have to be conducted for each edge set referring to an embedded tree in $N$. Note that, given a network containing $r$ hybridization nodes, this network can contain up to $2^r$ different embedded trees. Moreover, all maximum acyclic agreement forests of the chosen embedded tree $T'$ and the current input tree $T$ have to be taken into account, which can be done by applying the algorithm \textsc{allMAAFs} \cite{Scornavacca2012}. 

The insertion of components of a maximum acyclic agreement forest to a so far computed network is not a trivial step, since, usually, depending on other so far existing reticulation edges, there exist several potential ways of how $T$ can be inserted with the help of those components. Thus, this step will be discussed separately in the following section. 

\subsection{Inserting Components into Networks}
\label{sec-minNet}

Given an ordering of rooted phylogenetic $\cX$-trees, say $(T_1,T_2,\dots,T_n)$, and a network $N$ displaying each tree $T_j$ with $1\le j<i\le n$ together with an edge set $E'$ referring to some embedded tree $T'$ of $N$, we can add $T_i$ to $N$ by inserting further reticulation edges each corresponding to a specific component of a maximum acyclic agreement forest $\cF$ for $T'$ and $T_i$. Consequently, for each component a specific target and source node in $N$ has to be determined. Since different source and target nodes can lead to topologically different networks containing different sets of embedded trees, in order to obtain all relevant networks, we have take all valid combinations of source and target nodes for each component of $\cF$ into account. More precisely, we consider a pair $(s,t)$ of source and target nodes as being \emph{valid}, if $s$ cannot be reached from $t$. Furthermore, we have to consider each possible acyclic ordering of $\cF$.

Hence, we can summarize all important steps for inserting components of a maximum acyclic agreement forest $\cF$ into a network $N$ as follows.
\begin{enumerate}
 \item Choose an acyclic ordering $(F_{\rho},F_1,\dots,F_k)$ of $\cF$.
 \item Add each component $F_j$ of this ordering, except $F_{\rho}$, sequentially to $N$ by inserting a new reticulation edge connecting a certain source and target node.
\end{enumerate}
The output of these two steps is usually a large number of new networks, since, in general, there exist several pairs of source and target nodes enabling an embedding of $T_i$. Whereas all acyclic orderings of $\cF$ can be simply computed with the help of the directed graph $\scAG(T',T_i,\cF)$ (cf.~Sec.~\ref{sec-pre}), the second step inserting its components is quite more sophisticated. We will describe the way of adding a component $F_j$ of an acyclic ordering to a so far computed network $N$ by  first describing the computation of source and target nodes and then, based on these two nodes, by describing the way new reticulation edges are generated.\\

\textbf{\Romannum{1}~Computation of target and source nodes.} The set of source and target nodes corresponding to a component $F_j$ in $\cF$ is described in Step \Romannum{1}.\Romannum{1}--\Romannum{1}.\Romannum{3}. Therefor, let $$\cF'=\{F_{\rho},F_1,\dots,F_{j-1}\} \subset \cF=\{F_{\rho},F_1,\dots,F_{k}\}$$ be the set of components that has been added so far. Note that, since $N$ is initialized with $F_{\rho}$, at the beginning $\cL(\cF')$ equals $\cL(F_{\rho})$ and the first component that is added is $F_1$.\\

\textbf{\Romannum{1}.\Romannum{1}~Computation of target nodes.} The set $\cV_t$ of target nodes contains all nodes $v$ with $\overline{N|_{E',\cL(\cF')\cup\cL(F_j)}}(v)$ isomorphic to $T_i|_{\cL(F_j)}$. Due to the restriction of the network to $\cL(\cF')$, this set usually contains more than one node. Moreover, since we are only interested in relevant networks, we omit those target nodes that are source nodes of reticulation edges. This is a necessary step preventing the computation of networks containing stacks of hybridization nodes (cf.~Sec.~\ref{sec-pre}).\\

\textbf{\Romannum{1}.\Romannum{2}~Computation of source nodes of \textit{Type A}.} For each edge set $E_i$ referring to the embedded tree $T_i|_{\cL(\cF')}$ in $N$, the set $\cV_s^A$ of source nodes of \emph{Type A} contains all nodes $v$ with $\overline{N|_{E_i,\cL(\cF')}}(v)$ isomorphic to $T_i|_{\cL(\cF')}(v_\text{sib})$, where $v_\text{sib}$ denotes the sibling of the node $v'$ with $\cL(v')=\cL(F_j)$ in $T_i|_{\cL(\cF')\cup \cL(F_j)}$. Note that, due to the restriction of the network to $\cL(\cF')$, this set usually consists of more than one node. However, as we want to construct networks in which each hybridization node has out-degree one, we disregard those nodes having more than one in-edge.\\

\textbf{\Romannum{1}.\Romannum{3}~Computation of source nodes of \textit{Type B}.} The set $\cV_s^B$ of source nodes of \emph{Type B} is computed such that it contains each node $v$ of a subtree, whose root is a sibling of a node in $\cV_s^A$ not containing any leaves labeled by a taxon of $\cL(\cF')$. Moreover, its leaf set $\cL(v)$ has to consist only of those subsets representing the total taxa set $\cL(F)$ of a component $F$ in $\cF$, which means that $v$ must not be part of a subtree corresponding to a component that is added afterwards. However, as we want to construct networks in which each hybridization node has out-degree one, we disregard those nodes having more than one in-edge.\\

For a better understanding, the definitions of source and target nodes are illustrated in Figure~\ref{fig-stNodes}.\\

\begin{figure}[!t]
\centering
\includegraphics[width = 12cm]{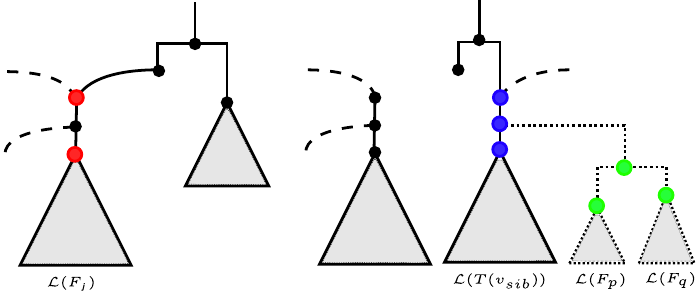}
\caption[An illustration of the definitions of source and target nodes regarding the algorithm \textsc{allHNetworks}]{An illustration of the definitions of target (left) and source nodes (right) for a component $F_j$ (with $j<p,q$) in which red nodes correspond to target nodes, blue nodes to source nodes of \textit{Type~A}, and green nodes to source nodes of \textit{Type~B}. Moreover, dashed edges and dotted edges are those edges that are disregarded when considering the restricted network in terms of the chosen embedded tree and the taxa set of the so far added components, respectively.}
\label{fig-stNodes}
\end{figure}

\emph{Remark.} Regarding two components $F_p$ and $F_q$ of an acyclic ordering of $\cF$ with $p<q$ and $\cF^*=\{F_{\rho},F_1,\dots,F_{q}\}$, it might be the case that both roots of $T_i|_{\cL(\cF^*)}(v_p)$ and $T_i|_{\cL(\cF^*)}(v_q)$ are siblings in $T_i|_{\cL(\cF^*)}$, where $v_p$ and $v_q$ denotes the lowest common ancestor of $\cL(F_p)$ and $\cL(F_q)$ in $T_i|_{\cL(\cF^*)}$. In this case, $F_p$ could be either added before $F_q$ or vice versa as both variants are acyclic orderings of $\cF$. If $F_p$ is inserted before $F_q$, a node whose leaf set corresponds to $\cL(F_p)$ in $N|_{T_i,\cL(\cF^*)}$ acts as source node when adding $F_q$ to $N$. Similarly, by adding $F_q$ before $F_p$ this happens the other way round which leads to a topologically different network. This implies that, in order to receive all relevant networks displaying $T_i$, we have to consider different acyclic orderings of a maximum acyclic agreement forest.\\

\textbf{\Romannum{2}~Adding new reticulation edges.} Now, given a valid pair $(s,t)$ of source and target nodes, a new reticulation edge is inserted as follows (cf.~Fig.~\ref{fig-stIns}).
\begin{enumerate}
 \item First, the in-edge $e$ of $s$ is split by inserting a new node $s'$, i.e., $e=(p,s)$ is first deleted and then two new edges $(p,s')$ ans $(s',s)$ are inserted. Second, if the parent of $t$ has in-degree one, the in-edge of $t$ is split two times in the same way by inserting two nodes $t'$ and $t''$. Let $t'$ be the parent of $t$ after splitting its in-edge. In this case, notice that $t'$ is necessary to receive only hybridization nodes of out-degree one and $t''$ is necessary to provide an attaching point for further reticulation edges as discussed below. Otherwise, if $t$ has an in-degree of at least two, $t'$ is set to $t$, which prevents the computation of networks containing stacks of hybridization nodes.
 \item Now, the two nodes, $s'$ and $t'$, are connected through a path $P$ consisting of two edges. As we do not allow nodes of in-degree larger than one as source nodes, this provides an attaching point for further reticulation edges within already inserted reticulation edges. Notice that, as direct consequence, in each completely processed network, in which all input trees have been inserted so far, one still has to suppress the source nodes of all reticulation edges as these nodes have both in- and out-degree one.
\end{enumerate}

\begin{figure}[t] 
\centering 
\includegraphics[scale=1.2]{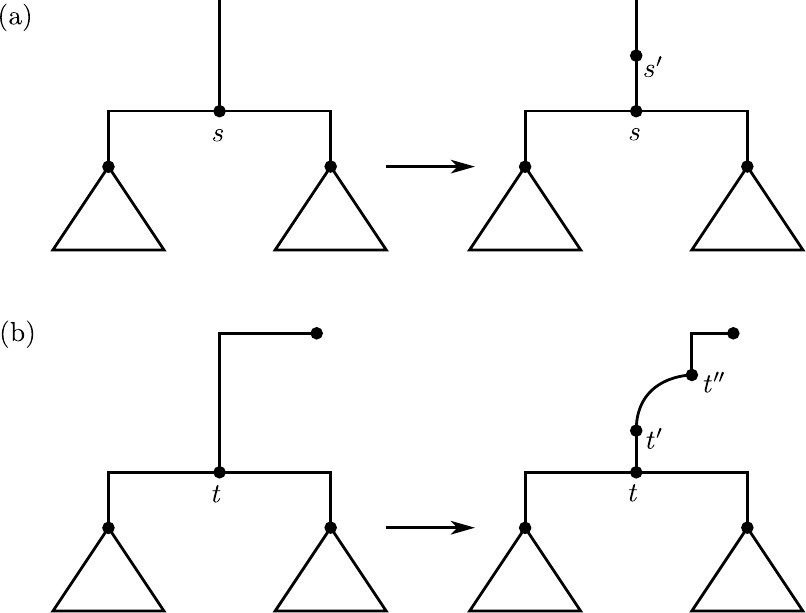}
\caption[Adding new reticulation edges into networks for binary trees]{Generating a source node \textbf{(a)} and a target node \textbf{(b)} for adding a new reticulation edge as described in Step~\Romannum{2}.}
\label{fig-stIns}
\end{figure}

In order to compute all relevant networks, one has to generate for each valid pair $(s,t)$ of source and target nodes a new network $\hat N$. This is necessary, since each of those networks contains different sets of embedded trees which can then be used for the insertion of further input trees and, thus, can initiate new computational paths leading to relevant networks.\\ 

For a better understanding, in Figure~\ref{fig-example} we illustrate the insertion of an input tree into a so far computed network.

\begin{figure}
\centering
\includegraphics[width = 12cm]{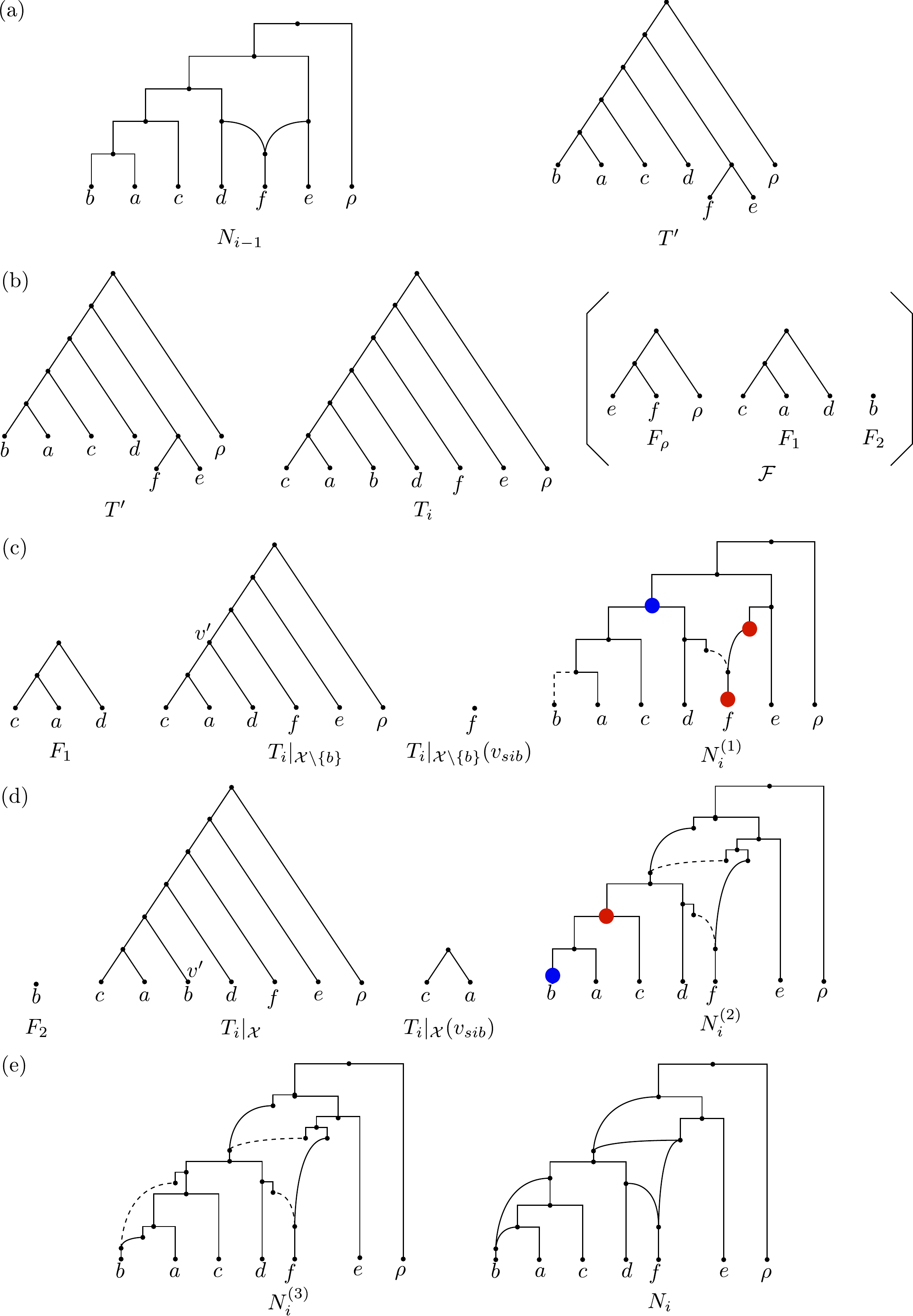}
\caption[An illustration of a tree-insertion regarding the algorithm \textsc{allHNetworks}]{An illustration of how an input tree $T_i$ is inserted into a network $N_{i-1}$. \textbf{(a)} The network $N_{i-1}$ together with an embedded tree $T'$. \textbf{(b)} The input tree $T_i$, which will be embedded into $N_{i-1}$ by inserting the maximum acyclic agreement forest $\cF$ of $T_i$ and $T'$ consisting of three components $F_{\rho}$, $F_1$, and $F_2$. \textbf{(c,d)} The important elements that have to be considered during the insertion of both components $F_1$ and $F_2$. Blue dots correspond to source nodes and red nodes to target nodes. Note that, regarding $N_i^{(1)}$, there is only one valid pair of source and target nodes. \textbf{(e)} The resulting network $N_i$, which is obtained from $N_i^{(3)}$ by suppressing each node of both in- and out-degree $1$.}
\label{fig-example}
\end{figure}

\subsection{Combinatorial complexity}

We finish the description of the algorithm by giving an idea of its combinatorial complexity. Given a set of rooted binary phylogenetic $\cX$-trees $\cT$, in order to guarantee the computation of all relevant networks for $\cT$, one has to consider the following combinations.

\begin{itemize}
\item[(1)] Take all possible orderings of $\cT$ into account.
\item[(2)] When adding a tree $T_i$ to a so far computed network $N$, each possible tree $T'$ that is displayed by $N$ has to be considered.
\item[(3)] When processing a so far computed network $N$ by adding a tree $T_i$ based on an a tree $T'$ that is displayed by $N$, take all acyclic orderings of each maximum acyclic agreement forest for $T_i$ and $T'$ into account.
\item[(4)] When adding a certain component of a maximum acyclic agreement forest for $T_i$ and $T'$ to a so far computed network, consider all valid pairs of source and target nodes.
\end{itemize}

Missing one those combinatorial elements could imply that a computational path leading to a relevant network is not visited. As a direct consequence, possibly either not all relevant networks are computed or the output consists only of those hybridization networks not providing a minimum hybridization number.

\clearpage
\subsection{Pseudocode of \textsc{allHNetworks}}

We end this section by giving a pseudocode summarizing all important steps of the algorithm \textsc{allHNetworks} described in the previous section. Some of those steps are denoted by a roman numeral that refers to the equally marked part of Section~\ref{sec-minNet}. 

\begin{algorithm}
\scriptsize
 \KwIn{Set $\cT$ of rooted phylogenetic $\cX$-trees}
 \KwOut{All topologically different hybridization networks $\cN$ with minimum hybridization number}
 \For{$k=1,\dots$}{
 $\cN=\emptyset$\;
 \ForEach{ordering $\pi$ of $\cT$}{
 $T_1 = \pi(1)$\;
 $\cN= \{T_1\}$\;
 \For{$i=2$ to $n$}{
	$T_i = \pi(i)$\;
	$\cN'=\emptyset$\;
	\ForEach{$N \in \cN$}{
	 \ForEach{$T'$ displayed in $N$}{
	 \ForEach{maximum acyclic agreement forest $\cF$ for $T'$ and $T_i$}{
	 \ForEach{acyclic ordering $(F_{\rho},F_1\dots,F_m)$ of $\cF$}{
	 	$\cN'' = \{N\}$\;
		\For{$j=1$ to $m$}{
		 $\cN''' = \emptyset$\;
		 \ForEach{$N'' \in \cN''$}{
		 	\Romannum{1}.\Romannum{1}~Compute all target nodes $\cV_t$ of $F_j$ in $N''$\;
			\Romannum{1}.\Romannum{2}~Compute all source nodes $\cV_a$ of \textit{Type A} of $F_j$ in $N''$\;
			\Romannum{1}.\Romannum{3}~Compute all source nodes $\cV_b$ of \textit{Type B} of $F_j$ in $N''$\;
			$\cV_s=\cV_a\cup \cV_b$\;
			\ForEach{$(s,t)\in\cV_s \times \cV_t : s\not\in N(t)$}{
			 $N''' = N''$\;
		 		\Romannum{2}~Insert reticulation edge $(s,t)$ in $N'''$\;
		 	 $\cN''' = \cN''' \cup \{N'''\}$\;
		 	}
		 }
		 $\cN'' = \cN'''$
		}
		\ForEach{$N'' \in \cN''$}{
		 	\If{$R(N'')<k$}{	 
		 		$\cN'=\cN'\cup \{N''\}$\;
		 	}
		 }
	 }
	 }
	 }
	}
	$\cN=\cN'$\;
 }
 }
 \If{$\cN\ne \emptyset$}{
 \Return{$\cN$}\;
 }
 }
\caption{\textsc{allHNetworks}$(\cT)$}
\end{algorithm}

\clearpage
\section{Use case}
\label{sec-useCase}

In the following, we give a demonstration of the algorithm \textsc{allHNetworks} by presenting a use case for three input trees with taxa set $\cX=\{\text{rho},1,2,\dots,10\}$. Each of the following Figures~\ref{fig-uc_1}--\ref{fig-uc_6} and Tables~\ref{tab-uc_1},~\ref{tab-uc} refers to a particular substep of the algorithm, which is discussed in the corresponding captions.\\

\begin{figure}
\centering
\begin{tabular}{cc}
\includegraphics[width = 6cm]{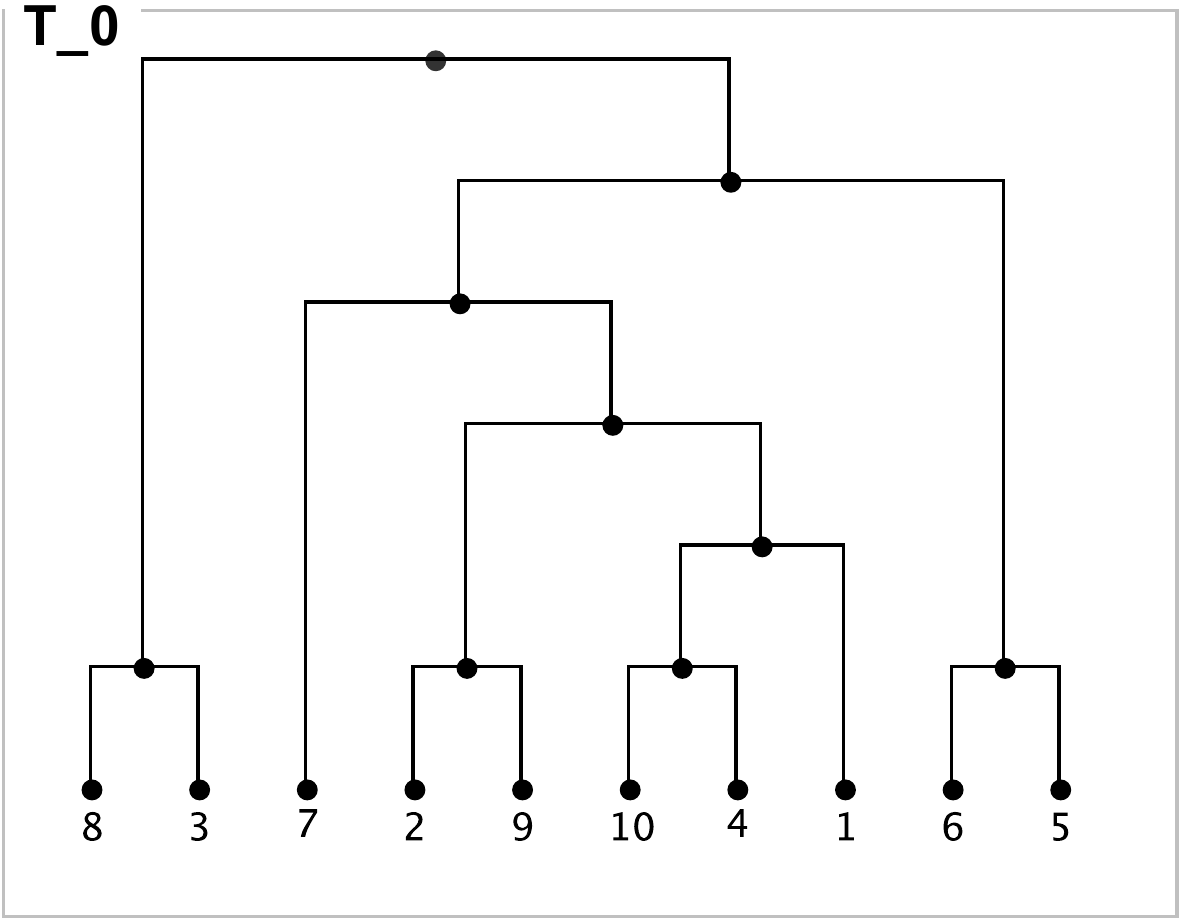}
&
\includegraphics[width = 6cm]{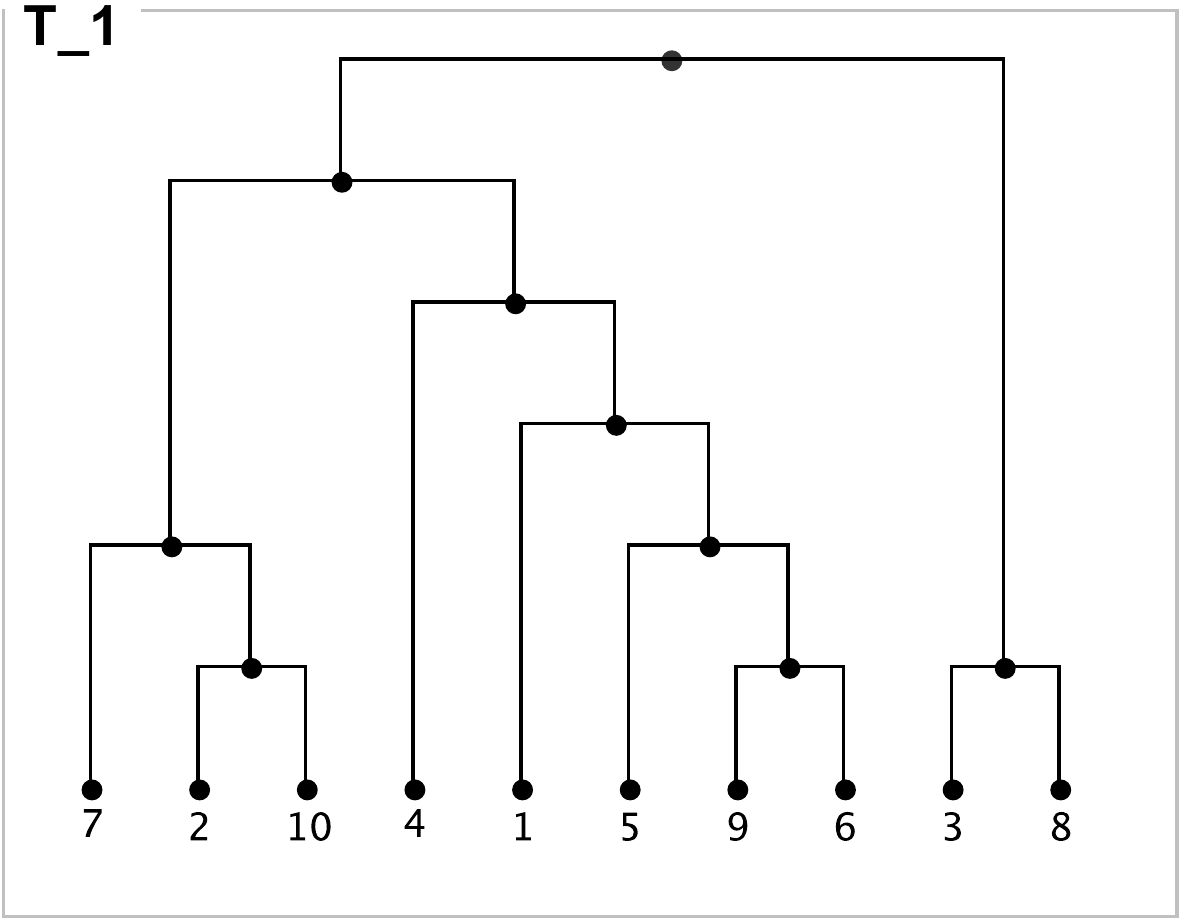}
\\
\includegraphics[width = 6cm]{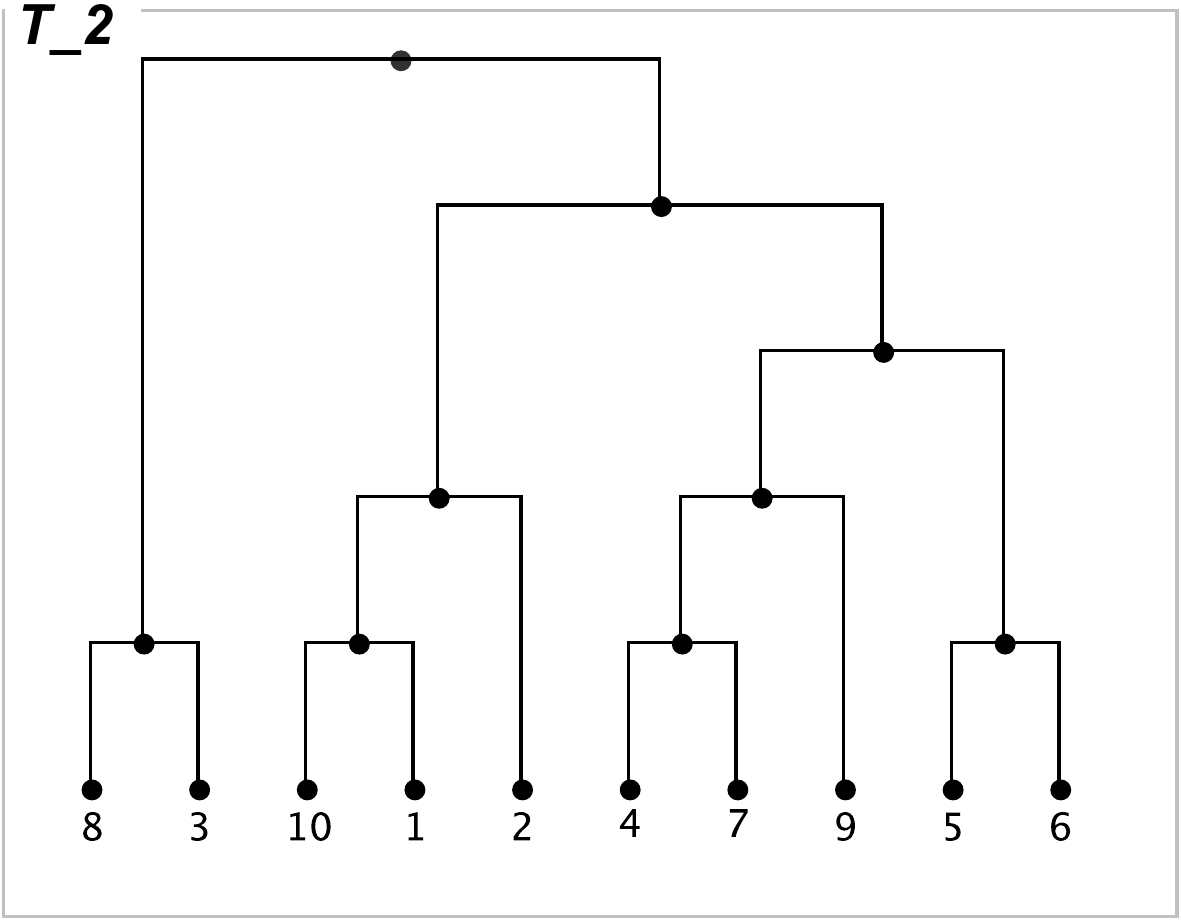}
&
\includegraphics[width = 6cm]{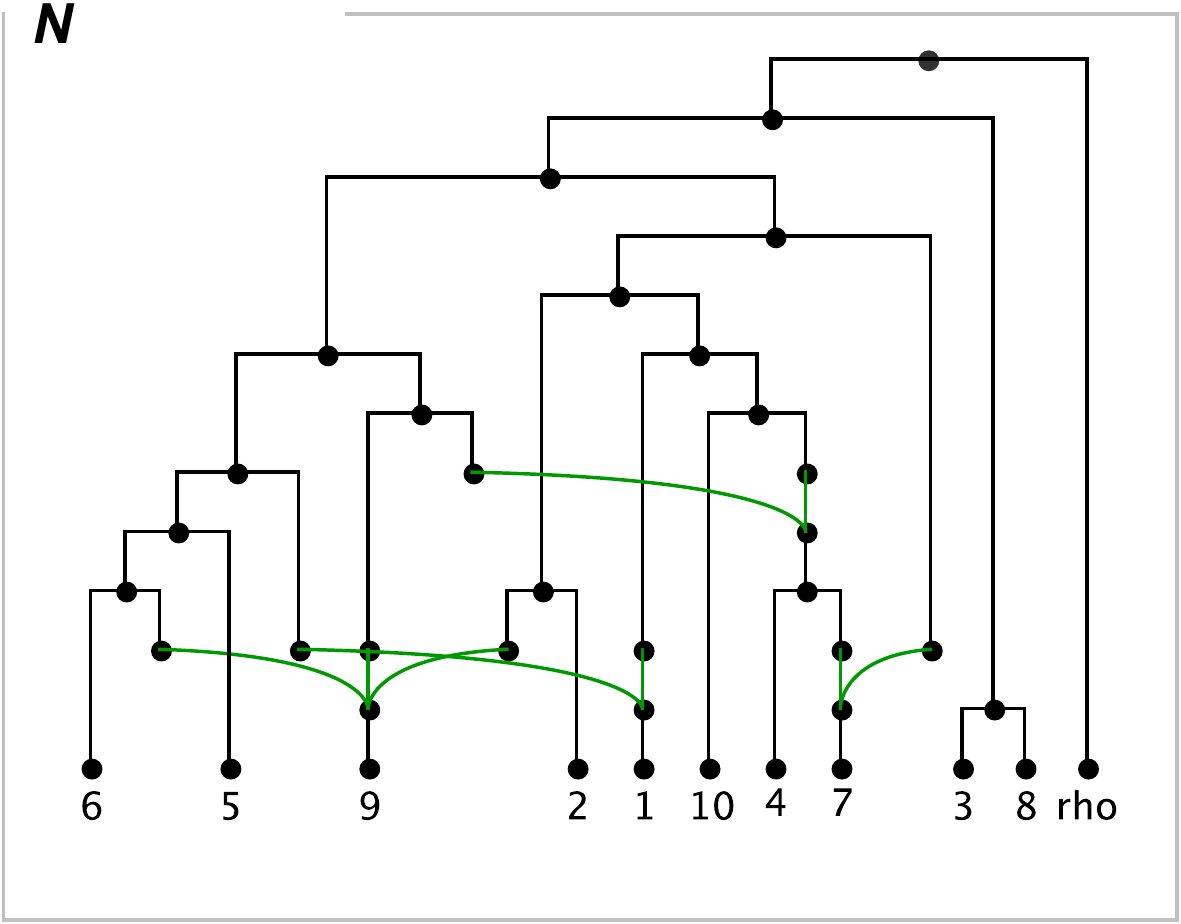}
\end{tabular}
\caption[Figure 1 of a use case referring to the algorithm \textsc{allHNetworks}]{The Figure shows the input set consisting of three rooted binary phylogenetic $\cX$-trees, namely $T_0$, $T_1$, and $T_2$, with $\cX=\{\text{rho},1,2,\dots,10\}$. The minimum hybridization network $N$ is one out five relevant networks for those input trees whose computation is now demonstrated step by step. The network is computed by applying the algorithm to the ordering $(T_1,T_2,T_0)$. Consequently, in a first step, $T_2$ has to be inserted into $T_1$.} 
\label{fig-uc_1}
\end{figure}

\begin{table}[!b]
\footnotesize
\caption[Parameters of a use case referring to the algorithm \textsc{allHNetworks}]{The computation of all pairs of source and target nodes based on $T_2$ given $T_1$ and the components depicted in Figure~\ref{fig-uc_2}. Notice that the notation refers to the one introduced in Section~\ref{sec-minNet}.} 
\begin{center} 
\begin{tabular}{>{\centering\arraybackslash}p{1cm}>{\centering\arraybackslash}p{1cm}>{\centering\arraybackslash}p{2cm}>{\centering\arraybackslash}p{1.5cm}>{\centering\arraybackslash}
p{1cm}>{\centering\arraybackslash}p{2cm}>{\centering\arraybackslash}p{1cm}>{\centering\arraybackslash}p{1cm}}
\toprule
$j$ & $F_j$ & $\cL(\cF')$ & $T_1|_{\cL(F_j)}$ & $\cV_t$ & $T_1|_{\cL(\cF')}(v_{\text{sib}})$ & $\cV_a$ & $\cV_b$\\
\midrule
1 & (7); & $\cX\setminus\{7,9,1\}$ & (7); & 7 & (4); & 4 & - \\
2 & (9); & $\cX\setminus\{9,1\}$ & (9); & 9 & (5,6); & 16 & - \\
3 & (1); & $\cX\setminus\{1\}$ & (1); & 1 & (10); & 10 & - \\
\bottomrule
\end{tabular} 
\end{center}
\label{tab-uc_1}
\end{table} 

\begin{figure}
\centering
\begin{tabular}{cc}
\includegraphics[width = 6cm]{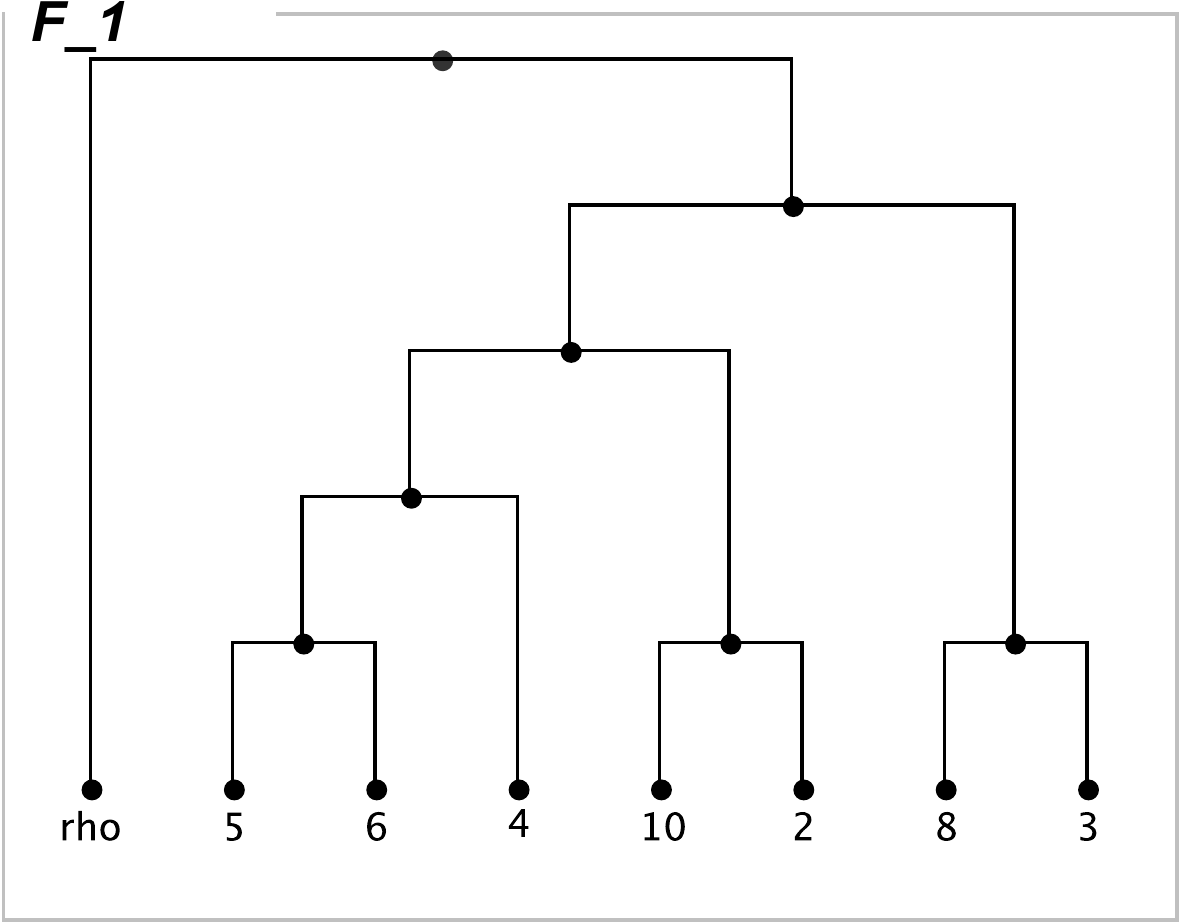}
&
\includegraphics[width = 6cm]{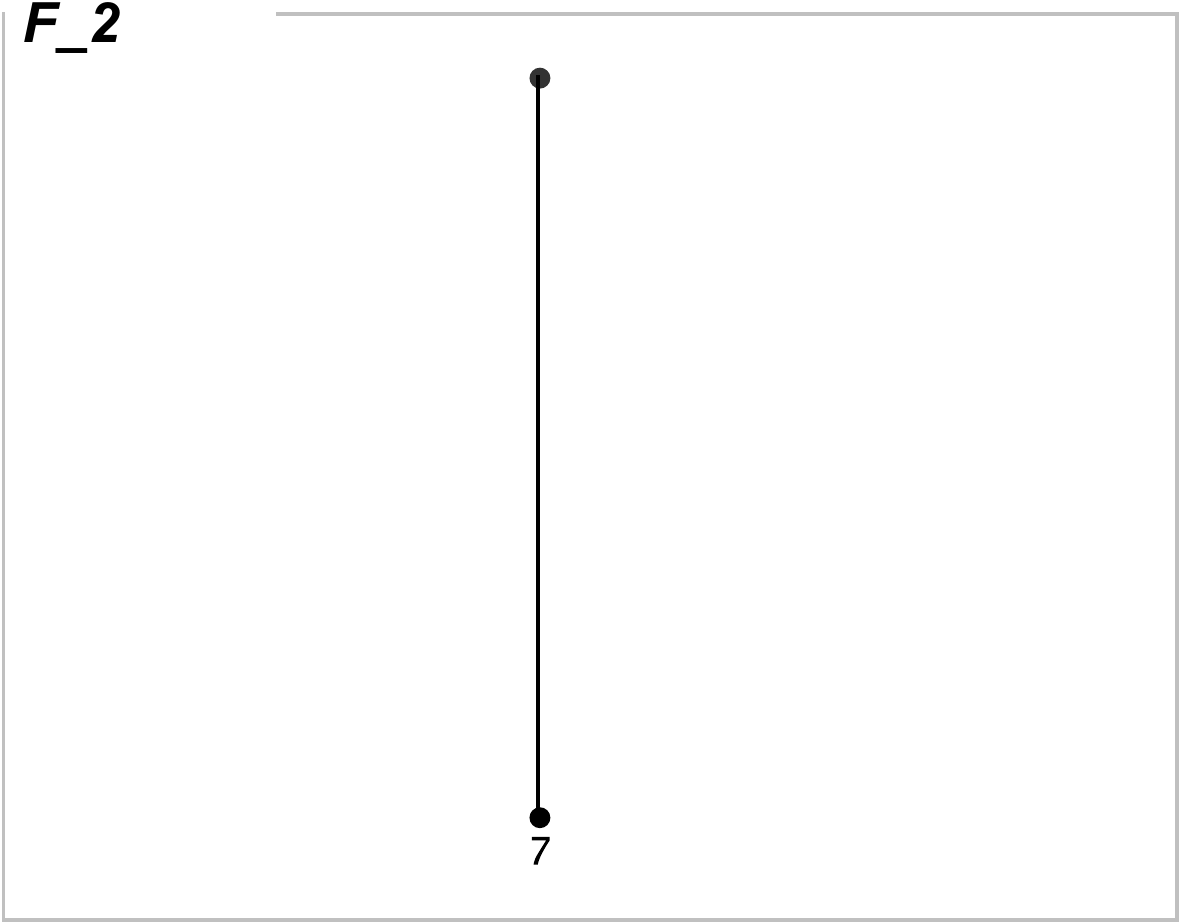}
\\
\includegraphics[width = 6cm]{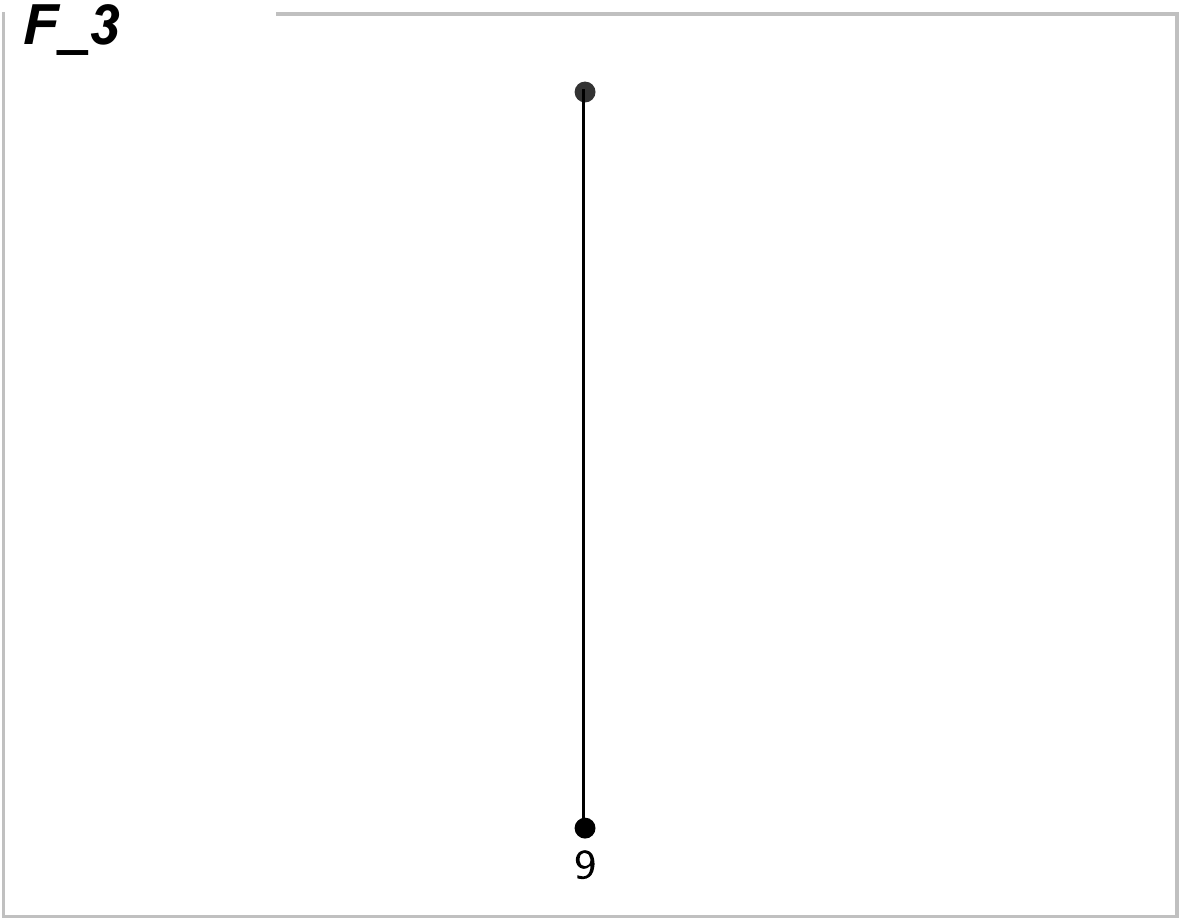}
&
\includegraphics[width = 6cm]{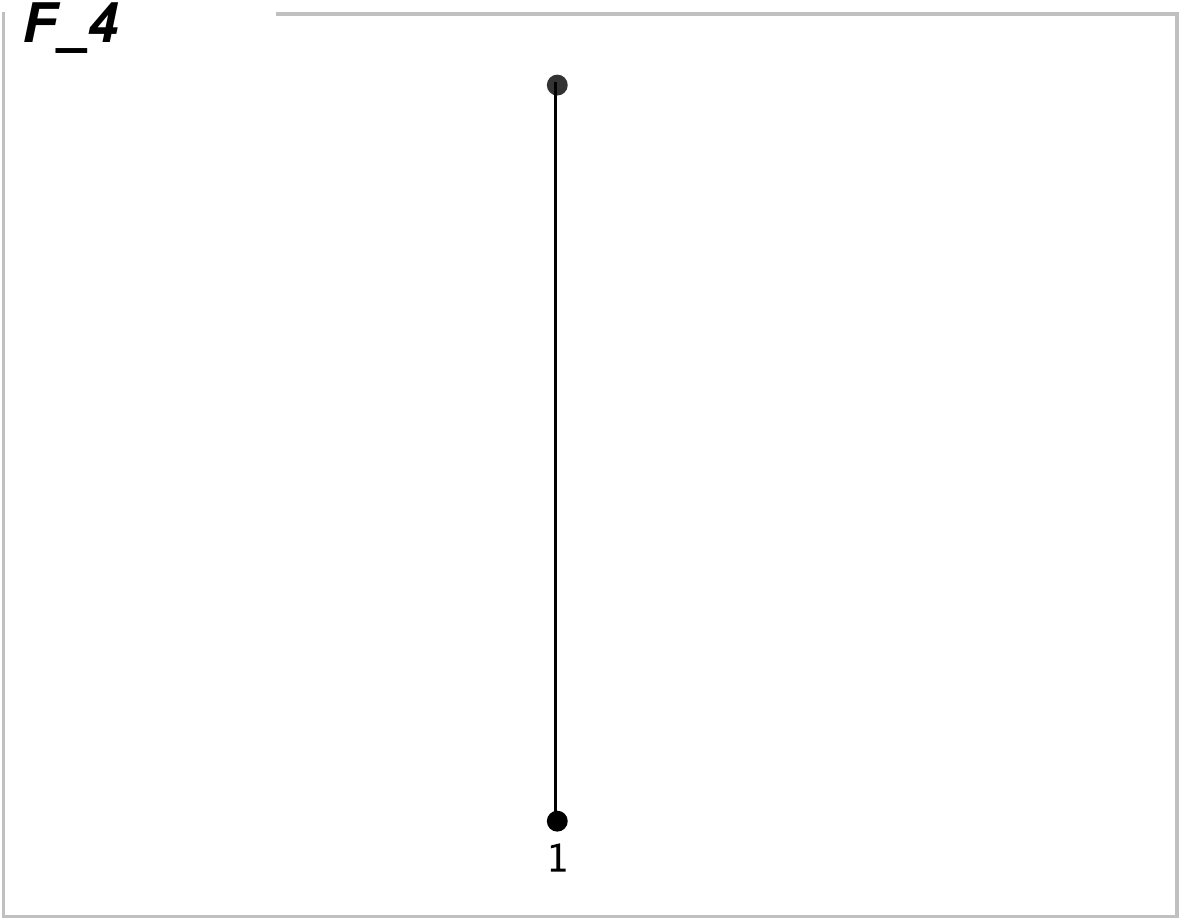}
\end{tabular}
\caption[Figure 2 of a use case referring to the algorithm \textsc{allHNetworks}]{At the beginning, the first network as well as its embedded tree both refer to $T_1$. Hence, in a first step, the maximum acyclic agreement forest $\{F_1,F_2,F_3,F_4\}$ for $T_1$ and $T_2$ is computed whose components are used in a subsequent step to receive $N_0$ displaying both trees.} 
\label{fig-uc_2}
\end{figure}

\begin{figure}
\centering
\begin{tabular}{cc}
\includegraphics[width = 6cm]{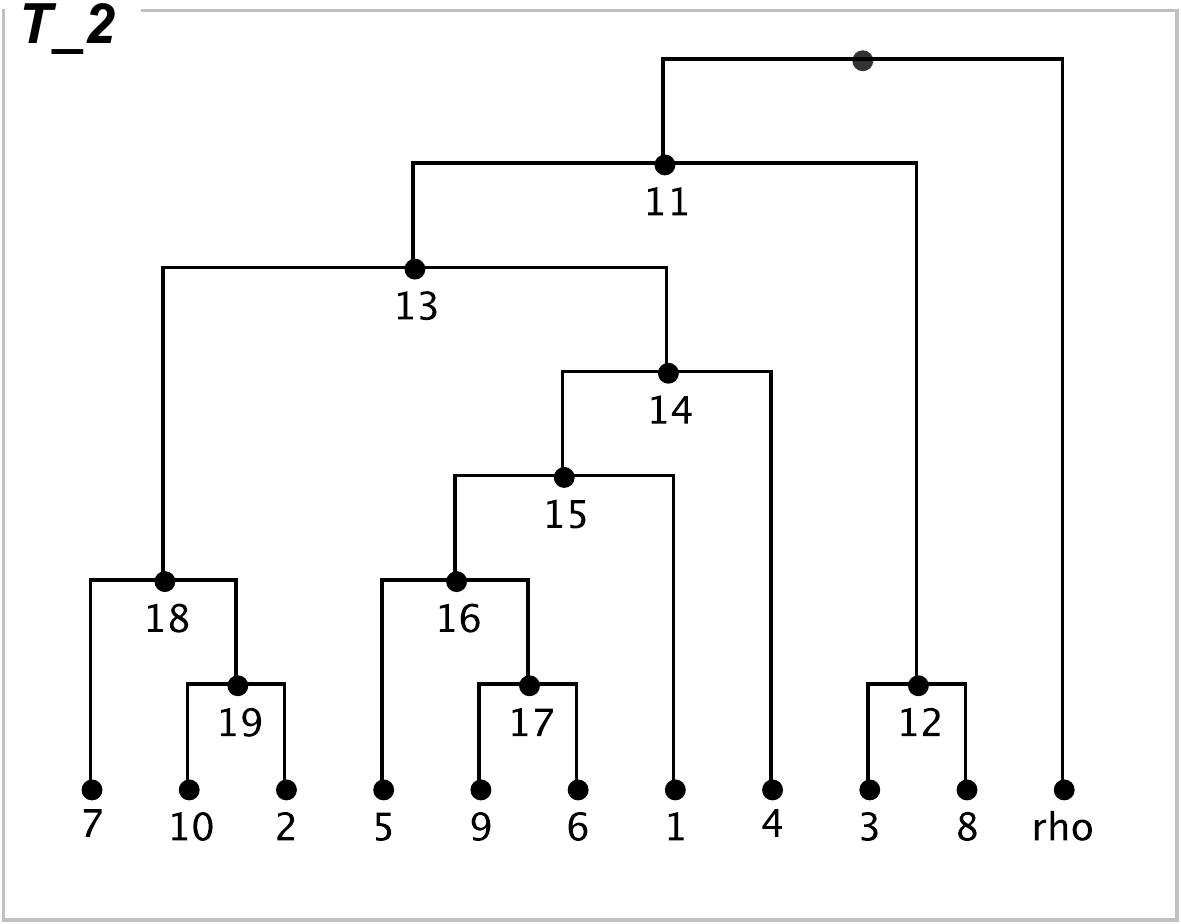}
&
\includegraphics[width = 6cm]{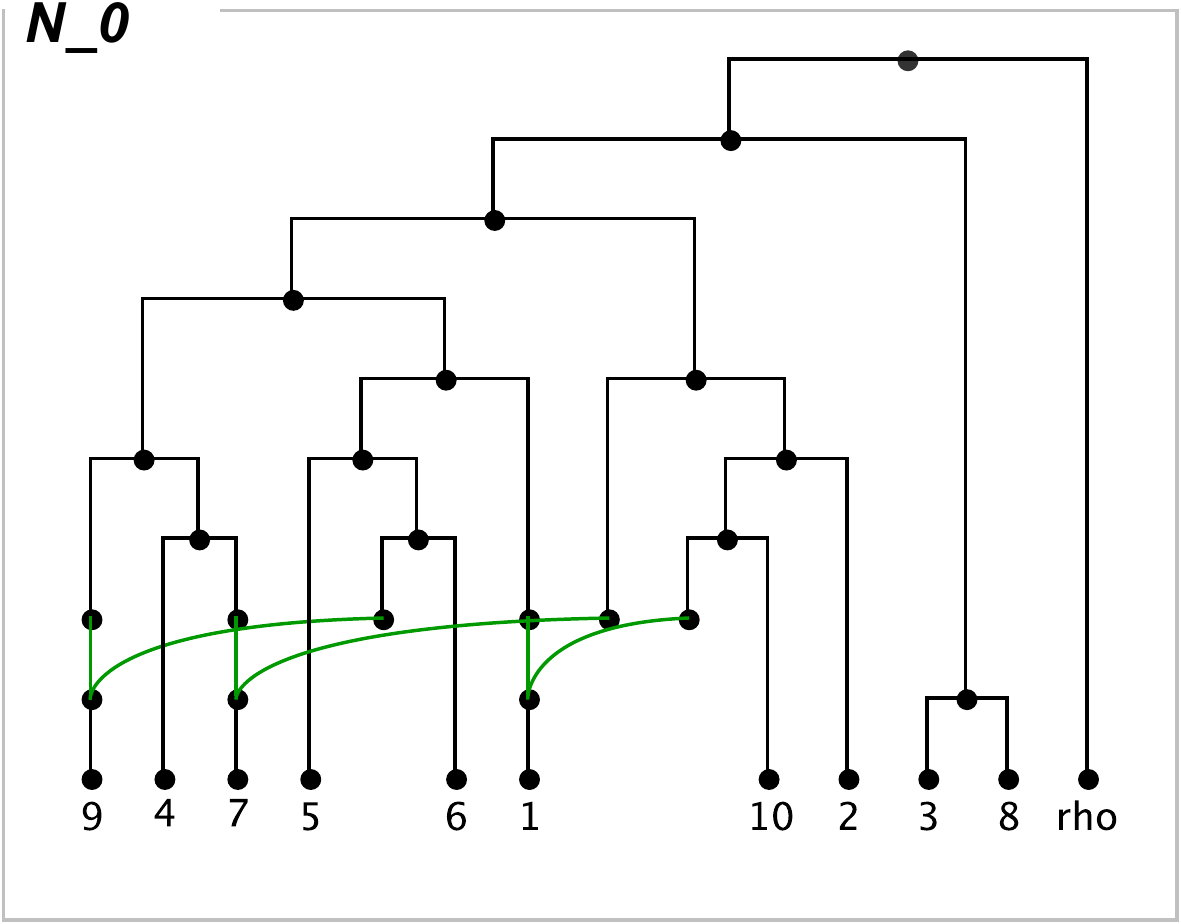}
\end{tabular}
\caption[Figure 3 of a use case referring to the algorithm \textsc{allHNetworks}]{Network $N_0$ is computed by adding the components $F_2$, $F_3$, and $F_4$ (cf.~Fig.~\ref{fig-uc_2}) sequentially in acyclic order to $T_2$. This is done by first computing pairs of target and source nodes (cf.~Step~\Romannum{1}.\Romannum{1}--\Romannum{3} of the algorithm \textsc{allHNetworks}) and then by inserting new reticulation edges for each of those pairs (cf.~Step~\Romannum{2} of the algorithm \textsc{allHNetworks}). Table~\ref{tab-uc_1} indicates the computation of these source and target nodes by referring to the notation used in Section~\ref{sec-minNet}.} 
\label{fig-uc_3}
\end{figure}

\begin{figure}
\centering
\begin{tabular}{cc}
\includegraphics[width = 6cm]{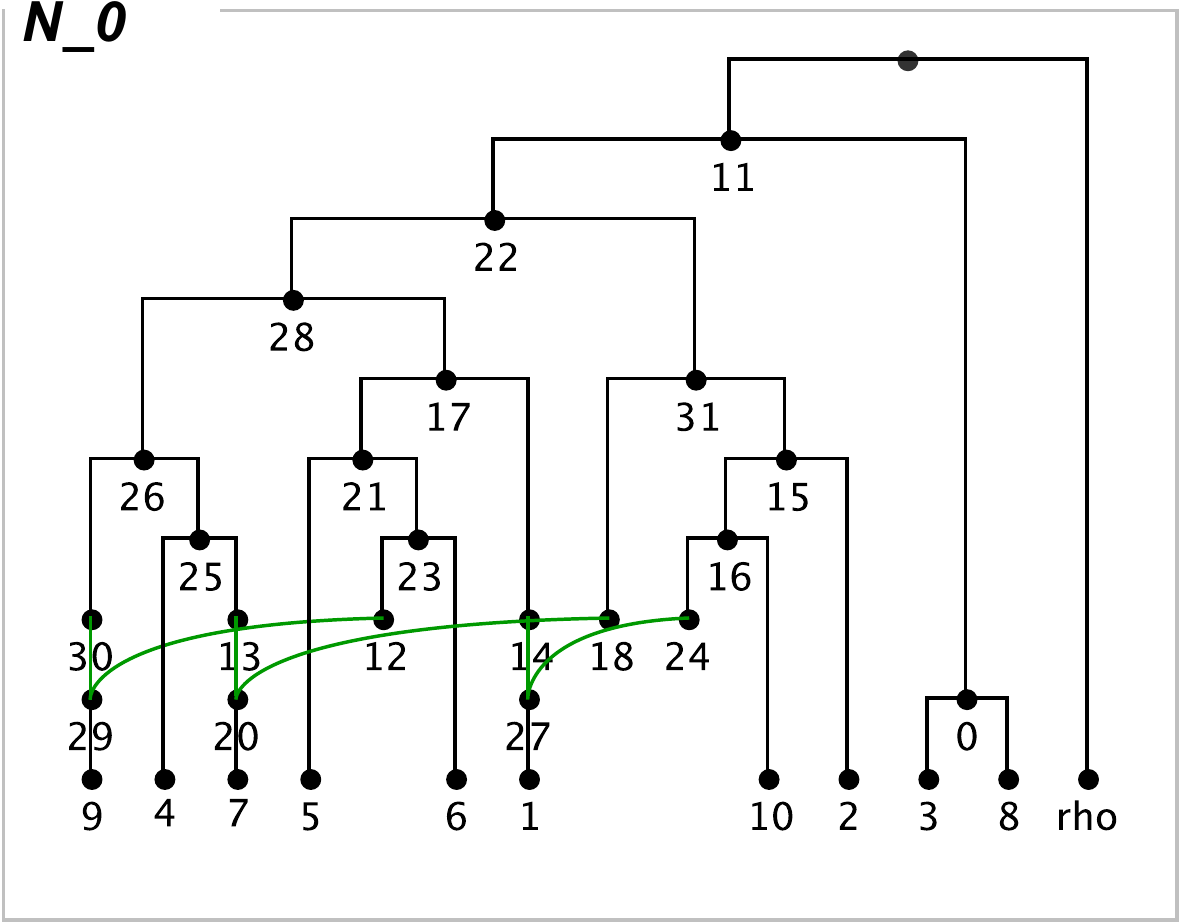}
&
\includegraphics[width = 6cm]{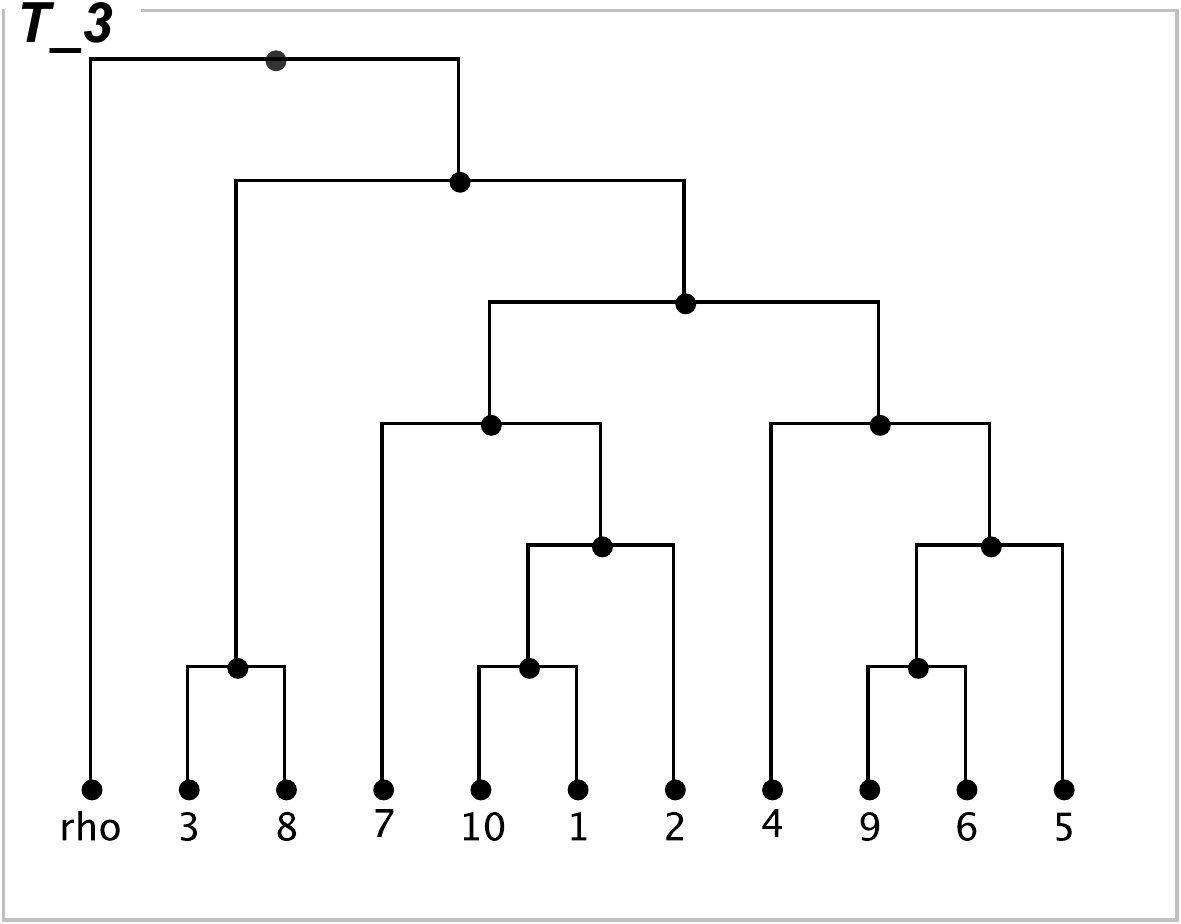}
\end{tabular}
\caption[Figure 3 of a use case referring to the algorithm \textsc{allHNetworks}]{Next, all embedded binary phylogenetic $\cX$-trees are extracted from $N_0$ each by selecting one in-edge of each hybridization node. One of those trees is $T_3$, which is received by selecting the edges $(12,29)$, $(18,20)$, and $(24,27)$.} 
\label{fig-uc_4}
\end{figure}

\begin{table}
\footnotesize
\caption[Parameters of a use case referring to the algorithm \textsc{allHNetworks}]{The computation of all valid pairs of source and target nodes based on $T_2$ given $N_0$, the extracted tree $T_3$ and the components depicted in Figure~\ref{fig-uc_5}. Note that the notation refers to the one introduced in Section~\ref{sec-minNet}.} 
\begin{center} 
\begin{tabular}{>{\centering\arraybackslash}p{1cm}>{\centering\arraybackslash}p{1cm}>{\centering\arraybackslash}p{2cm}>
{\centering\arraybackslash}p{1.5cm}>{\centering\arraybackslash}p{1cm}>{\centering\arraybackslash}p{2cm}>{\centering\arraybackslash}p{1cm}>{\centering\arraybackslash}p{1cm}}
\toprule
$i$ & $F_j$ & $\cL(\cF')$ & $T_0|_{\cL(F_j)}$ & $\cV_t$ & $T_0|_{\cL(\cF')}(v_{\text{sib}})$ & $\cV_a$ & $\cV_b$\\
\midrule
1 & (4); & $\cX\setminus\{4,9\}$ & (4); & 4, 25 & (10); & 10 & - \\
2 & (9); & $\cX\setminus\{9\}$ & (9); & 9 & (2); & 2 & - \\
\bottomrule
\end{tabular} 
\end{center}
\label{tab-uc}
\end{table} 

\begin{figure}
\centering
\begin{tabular}{cc}
\includegraphics[width = 6cm]{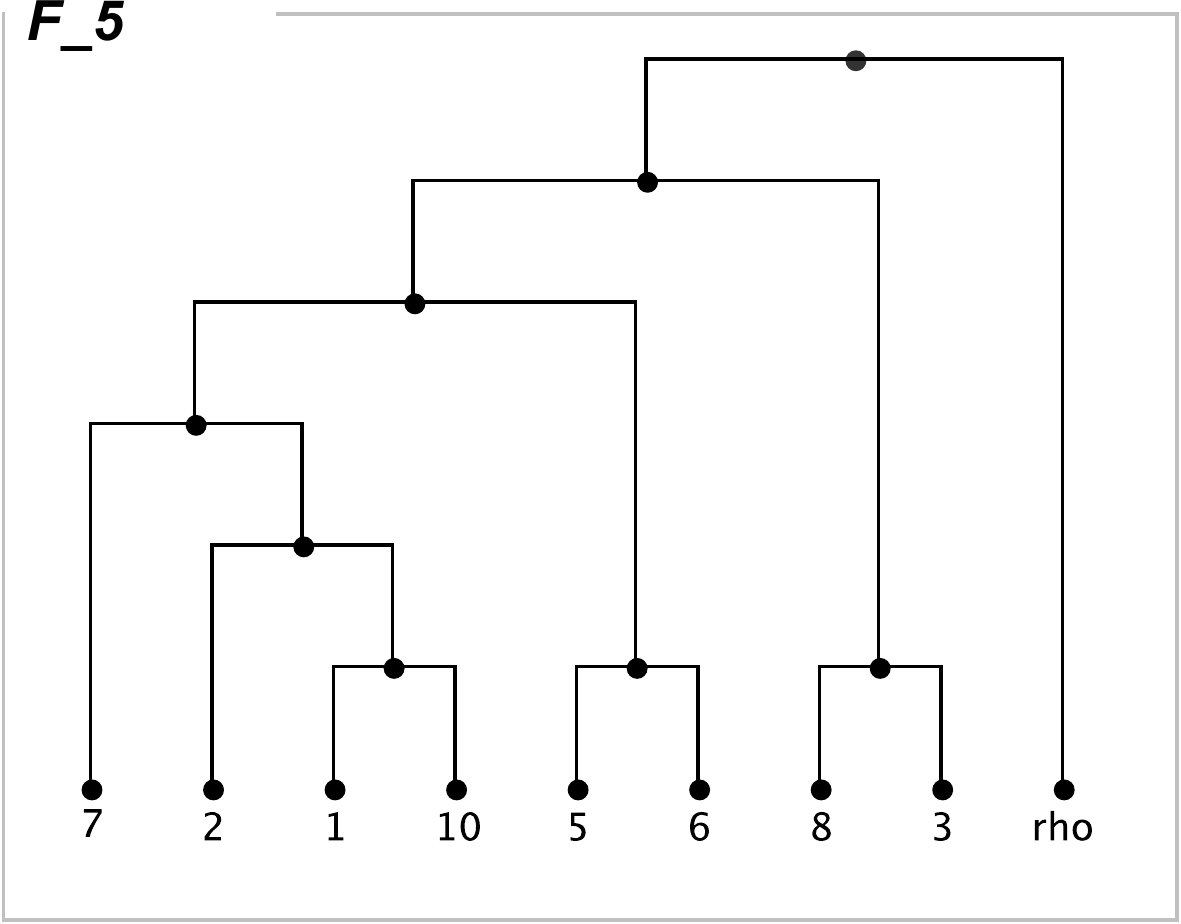}
&
\includegraphics[width = 6cm]{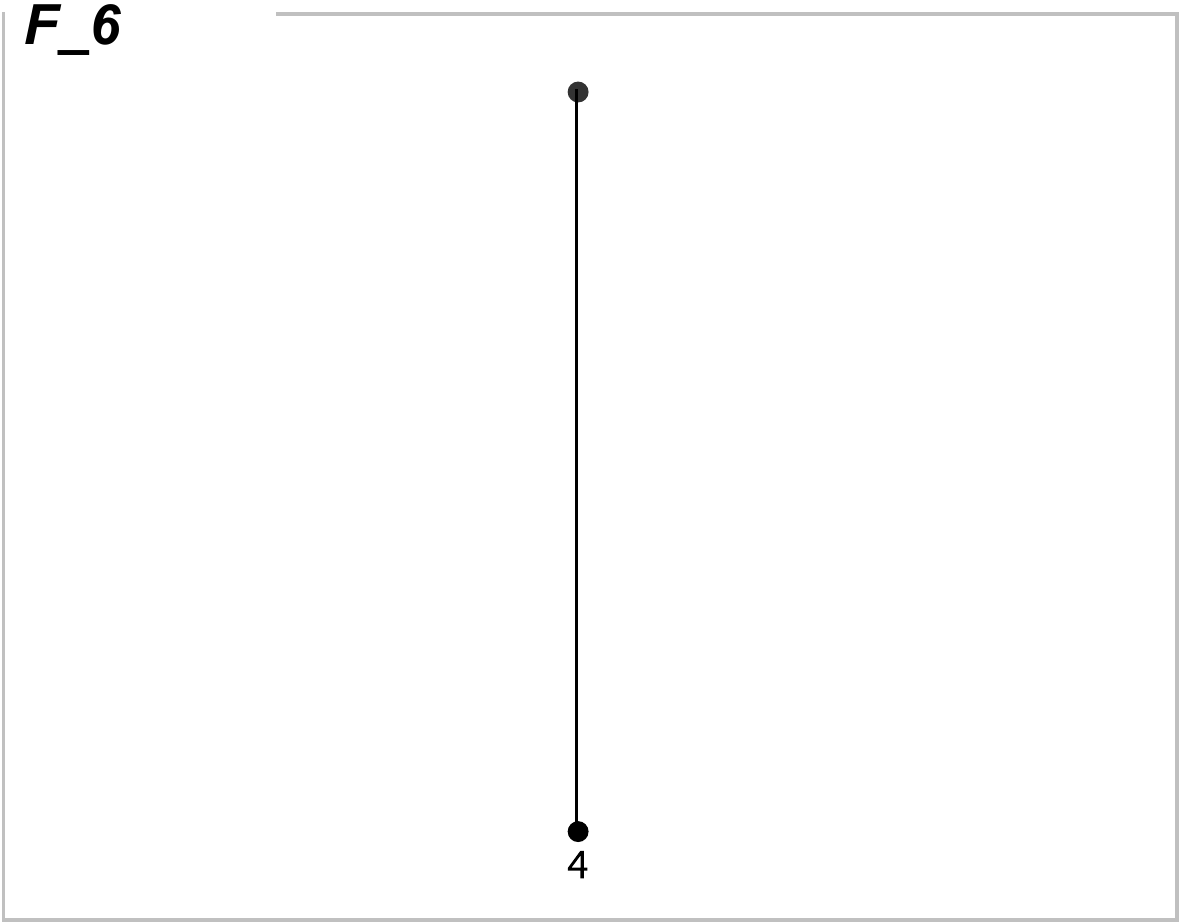}
\\
\multicolumn{2}{c}{\includegraphics[width = 6cm]{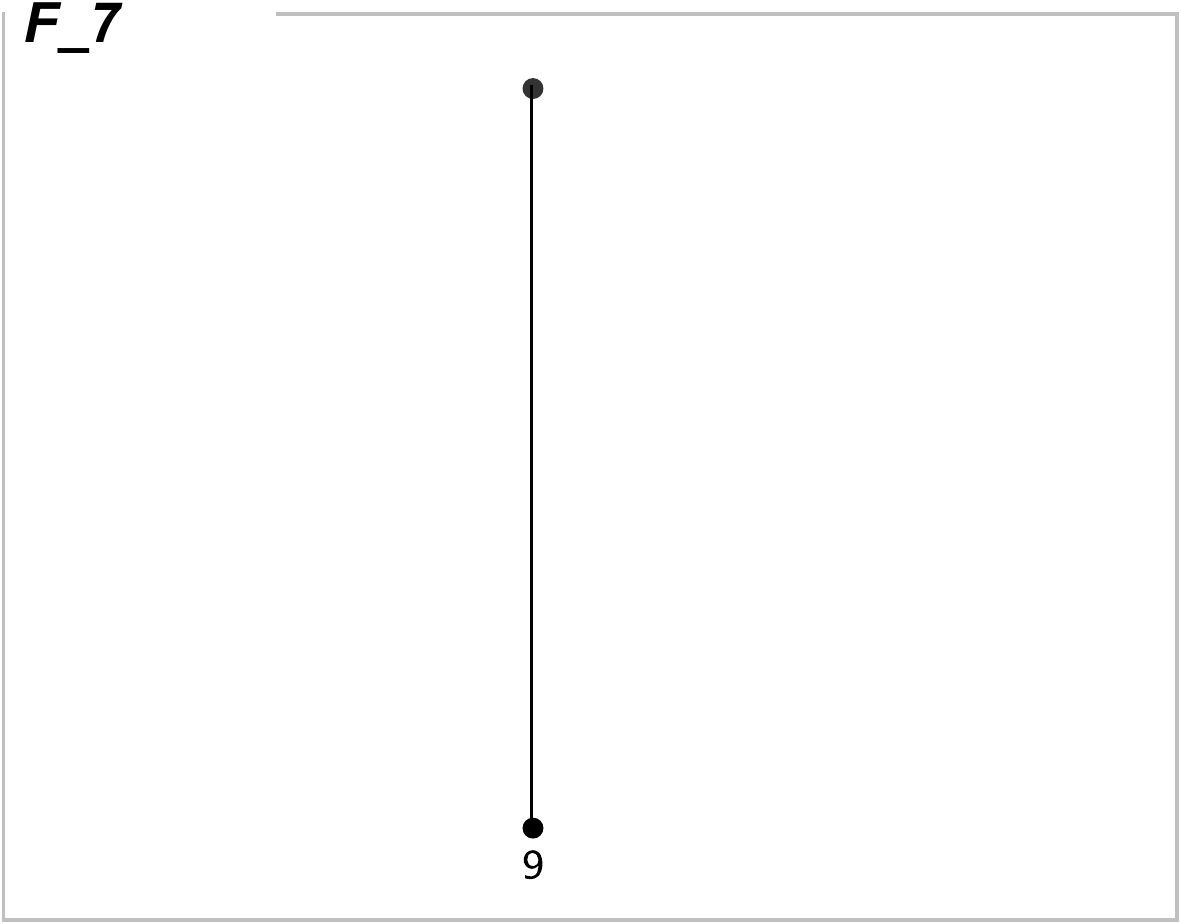}}
\\
\end{tabular}
\caption[Figure 4 of a use case referring to the algorithm \textsc{allHNetworks}]{Now, again a maximum acyclic agreement forest $\{F_5,F_6,F_7\}$ for the extracted tree $T_3$ and the input tree $T_0$ is computed whose components are used in a subsequent step to receive the final network $N$ displaying all input trees.} 
\label{fig-uc_5}
\end{figure}

\begin{figure}
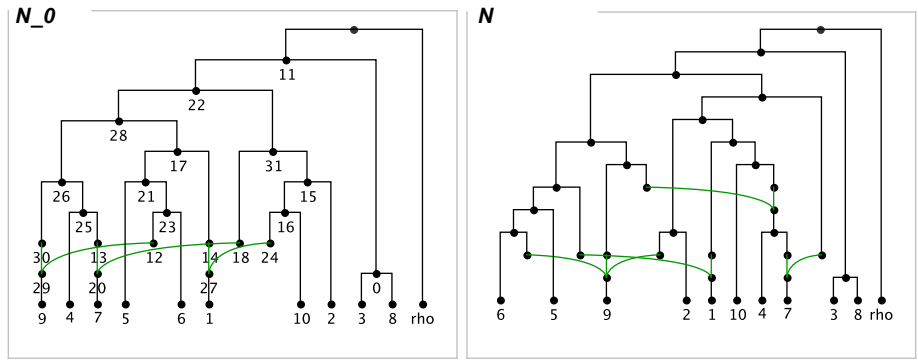

\centering
\begin{tabular}{cc}
\includegraphics[width = 6cm]{Network0with_Labels.pdf}
&
\includegraphics[width = 6cm]{Network2.pdf}
\end{tabular}
\caption[Figure 5 of a use case referring to the algorithm \textsc{allHNetworks}]{The relevant network $N$ is computed by adding the components $F_6$ and $F_7$ (cf.~Fig.~\ref{fig-uc_5}) sequentially in acyclic order to $N_0$. This is done by first computing target and source nodes (cf.~Step~\Romannum{1}.\Romannum{1}--\Romannum{3} of the algorithm \textsc{allHNetworks}) and then by inserting new reticulation edges for each pair of source and target nodes (cf.~Step~\Romannum{2} of the algorithm \textsc{allHNetworks}). Table~\ref{tab-uc} indicates the computation of all pairs of source and target nodes by referring to the notation used in Section~\ref{sec-minNet}.} 
\label{fig-uc_6}
\end{figure}

\clearpage
\section{Proof of correctness}
\label{sec-proof}

Here, we proof the main result of this work, namely that for a set of rooted binary phylogenetic $\cX$-trees the algorithm \textsc{allHNetworks} calculates all relevant networks as defined in Section~\ref{sec-pre}. 

\begin{theorem}
Given a set of binary rooted phylogenetic $\cX$-trees $\cT$, by calling $$\text{\textsc{allHNetworks}$(\cT)$}$$ all relevant networks for $\cT$ are calculated.
\label{th-one}
\end{theorem}
For clarity, here we consider two relevant networks $N_1$ and $N_2$ as being different if both graph topologies of $N_1$ and $N_2$ (disregarding the embedding of $\cT$) differ. 

\begin{proof}
The proof of Theorem~\ref{th-one} is based on the following three Lemmas \ref{lem-two}--\ref{lem-three}. Here, we first show that the concept of acyclic agreement forest suffices to generate all of the desired networks. Next, we argue that for inserting acyclic agreement forests the algorithm takes all necessary pairs of source and target nodes into account.  Finally, we proof that by taking all orderings of the input trees into account it suffices to focus only on acyclic agreement forests of minimum size, i.e., maximum acyclic agreement forests. Before entering the first lemma, however, we first have to introduce some further notations.

Let $N$ and $N'$ be two rooted phylogenetic networks on $\cX$. Then, we say that \emph{$N'$ is displayed by $N$}, shortly denoted by $N \supset N'$, if $N'$ can be obtained from $N$ by first deleting some of its reticulation edges and then by suppressing all nodes of both in- and out-degree~1.

Similarly, let $N$ be a hybridization network on $\cX$ displaying two rooted phylogenetic $\cX$-trees $T_1$ and $T_2$. Now, given an acyclic agreement forest $\cF$ for those two trees, we say that \emph{$N$ displays $\cF$}, shortly denoted by $N\supset \cF$, if we can obtain $\cF$ from $N$ as follows. Regarding $N$, let $E_1$ and $E_2$ be two sets of reticulation edges referring to $T_1$ and $T_2$, respectively. First in $N$ all reticulation edges are deleted that are not contained in $E_1\cap E_2$ and then all nodes of both in- and out-degree $1$ are suppressed. Notice that, by deleting those edges the network is disconnected into a set of disjoint trees each corresponding to exactly one of the components in $\cF$.

Let $\cT=\{T_1,T_2,\dots,T_n\}$ be a set of rooted binary phylogenetic $\cX$-trees and let $N$ be a hybridization network displaying $\cT$. Moreover, let $E_i$ be an edge set in $N$ referring to a tree $T_i\in T$. Then, for a tree $T_k\in\cT$ the edge set $\hat E^{(k)}$ refers to the edge set $$E_k\setminus E_1\cup E_2\dots\cup E_{k-1}\cup E_{k+1}\dots\cup E_n.$$ This means, in particular, that $\hat E^{(k)}$ consists of those reticulation edges that are only necessary for displaying $T_k$ and none of the other trees in $\cT$. 

Next, let $N$ be a phylogenetic network and let $E'$ be a subset of its reticulation edges. Then, by writing $N\ominus E'$ we refer to the network that is obtained from $N$ by first deleting each edge in $E'$ and then by suppressing each node of both in- and out-degree $1$.

\begin{lemma}
Let $\cT=\{T_1,T_2,\dots,T_n\}$ with $n>1$ be a set of rooted binary phylogenetic $\cX$-trees and let $N$ be a hybridization network displaying $\cT$. Moreover, let $E_i$ with $i\in[1:n]$ be an edge set referring to the respective tree $T_i\in\cT$ in $N$. Then, for each tree $T_k$ in $\cT$ the network $N\ominus\hat E^{(k)}$ contains an embedded tree $T'$ such that $N$ contains an acyclic agreement forest $\cF$ for $T'$ and $T_k$, i.e., $N\supset\cF$ holds.
\label{lem-two}
\end{lemma}

\begin{proof}
Let $E_k$ be an edge set in $N$ referring to $T_k$ and, based on $E_k$, let $\hat E^{(k)}$ be the edge set in $N$ as defined above. Moreover, let $\cF$ be a set of subtrees that is derived from $N$ as follows. First, the network $N'$ is computed by removing each edge $e$ with $e\not\in E_k$. Next, each edge $e$ in $N'$ with $e\in\hat E^{(k)}$ is removed and, finally, each node of both in- and out-degree $1$ is suppressed. As the tree that can be derived from $N'$ by suppressing its nodes of both in- and out-degree one corresponds to $T_k$, it is easy to see that $\cF$ consists of common subtrees of $T_k$. Furthermore, as $\cF$ is obtained from $N'$ by cutting some of its edges, this implies that $\cF$ is a set of node-disjoint subtrees in $T_k$.

Next, we will show how one can derive an edge set $E'$ referring to a phylogenetic $\cX$-tree $T'$ displayed in $N$ so that $\cF$ is an agreement forest for $T'$ and $T_k$. Therefor, we say a reticulation edge $e$ of $N$ is of \emph{Type A}, if $e\in E_k\setminus\hat E^{(k)}$, and of \emph{Type B}, if $e\not\in E_k$. Now, let $E'$ be a subset of reticulation edges that is obtained from $N$ by visiting all of its reticulation nodes as follows. If, for a reticulation node, there exists an in-edge $e$ of \emph{Type A}, this edge is selected, otherwise, an arbitrary in-edge of \emph{Type B} is selected. As each edge in $E'$ is also contained in $N\ominus\hat E^{(k)}$, it is easy to see that $T'$ is also displayed by $N\ominus\hat E^{(k)}$. 


Now, let $\hat E'$ be the set of reticulation edges that is removed from $N$ by restricting $N$ on $E'$ and let $E_{\cF}$ be the set of reticulation edges that has been removed from $N$ in order to obtain $\cF$. Then, the target of each reticulation edge in $\hat E'\setminus E_{\cF}$ is a reticulation node providing an in-edge of $E^{(k)}$, which  has been removed from $N'$ (and, thus, actually from $T_k$) in order to obtain $\cF$. As a direct consequence, each component in $\cF$ can be also obtained from $T'$ by cutting some of its edges, which directly implies that $\cF$ is a set of node-disjoint subtrees in $T'$.

As a direct consequence, $\cF$ is an agreement forest for both trees $T_k$ and $T'$. Moreover, since $N$ is a hybridization network and, consequently, does not contain any directed cycles, $\cF$ has to be an acyclic agreement forest for both trees. 
\end{proof}

This means that for inserting further rooted binary phylogenetic $\cX$-trees into so far computed networks it is sufficient to focus only on acyclic agreement forests. Notice, however, that the insertion of further reticulation edges based on such agreement forests can be conducted in several ways. Thus, in order to calculate all relevant networks, our algorithm has to guarantee that all of those possibilities are exploited, which is stated by the following lemma. 

\begin{lemma}
Let $\cT=\{T_1,T_2,\dots,T_i\}$ be a set of rooted binary phylogenetic $\cX$-trees, $N_{i-1}$ be a network displaying each tree in $\cT\setminus\{T_i\}$, $E'$ be an edge set referring to some embedded tree $T'$ of $N_{i-1}$, and $\cF$ be an acyclic agreement forest for $T'$ and $T_i$. Then, the algorithm \textsc{allHNetworks} inserts $\cF$ into $N_{i-1}$ so that each hybridization network $N_i$ displaying $\cT$ with $N_i\supset\cF$ and $N_i\supset N_{i-1}$ is calculated.
\label{lem-one}
\end{lemma}

\begin{proof} 
Given an acyclic ordering $(F_{\rho},F_1,\dots,F_k)$ of the maximum acyclic agreement forest $\cF$ for the two trees $T'$ and $T_i$, then, when inserting each component $F_j$ in ascending order, beginning with $F_1$, all possible target and source nodes in $N_{i-1}$ are taken into account. More precisely, let $\cX'=\cL(\cF')$ with $\cF'=\{F_{\rho},F_1,\dots,F_{j-1}\}$ and let $v_\text{sib}$ be the sibling of a node $v$ with $\cL(v)=\cL(F_j)$ in $T_i|_{\cL(\cF')\cup \cL(F_j)}$.
\begin{itemize}
\item Since for each target node $w \in \cV_t$ the two trees $\overline{N_{i-1}|_{E',\cL(F_j)}(w)}$ and $T_i|_{\cL(F_j)}$, with $E'$ referring to $T'$, are isomorphic, each node $w'$ not in $\cV_t$ automatically does not fulfill this property and, consequently, by using such a node $w'$ as target node the resulting network $N_i$ would not display $F_j$, and, thus, $N_i\supset\cF$ would not hold.
\item For each source node $u \in \cV_s^A \cup \cV_s^B$ either the two trees $\overline{N_{i-1}|_{E',\cL(\cF')}(u)}$ and $T_i|_{\cL(\cF')}(v_\text{sib})$ are isomorphic (if $u \in \cV_s^A$) or, after the insertion of all components in $\cF$, there exists a certain path leading to such a node whose edges can be used for displaying $T_i$ (if $u \in \cV_s^B$). Choosing a node $u'\not\in\cV_s^A\cup\cV_s^B$ as source node, the reticulation edge $e$ inserted for $u'$ and some node $w \in \cV_t$, could not be used for displaying $T_i$ in $N$, since $T_i|_{\cX'\cup \cL(F_j)}$ does not contain a node $v$ whose subtree $\overline{T_i|_{\cX'\cup \cL(F_j)}(v)}$ is isomorphic to $\overline{N_i|_{E_i,\cX'\cup \cL(F_j)}(u')}$, with $E_i$ referring to $T_i$. 
\end{itemize}

Thus, following an acyclic ordering of $\cF$, the algorithm \textsc{allHNetworks} considers all possible source and target nodes that can be used for inserting one of its components into the so far computed network~$N_{i-1}$. 

However, as already discussed, for $\cF$ there may exist different acyclic orderings and, depending on these acyclic orderings, the set $\cF'$ of so far added components varies. Consequently, for different acyclic orderings the tree $T_i|_{\cL(\cF')}(v_\text{sib})$ can differ, which may lead to different sets of source nodes. However, since for inserting an acyclic agreement forest $\cF$ the algorithm \textsc{allHNetworks} takes all of its acyclic orderings into account, all of these different sets of target nodes are automatically considered and, thus, Lemma~\ref{lem-one} is established.
\end{proof}

We have shown so far that, given an ordering of input trees $\Pi^*=(T_1,T_2,\dots,T_n)$, each input tree $T_i$ can be added sequentially to a so far computed network $N_{i-1}$ displaying all previous trees $\{T_1,T_2,\dots,T_{i-1}\}$ by inserting an acyclic agreement forest $\cF$ for some embedded tree $T'$ and $T_i$ in all possible ways such that there does not exist a network $N_i$ displaying $\{T_1,T_2,\dots,T_i\}$ with $N_i\supset\cF$ and $N_i\supset N_{i-1}$. Notice that, as for inserting $T_i$ all embedded trees are taken into account, if the algorithm would additionally consider all acyclic agreement forests of arbitrary size, Lemma \ref{lem-two} and \ref{lem-one} would be sufficient to establish Theorem~\ref{th-one}. 

However, in order to maximize efficiency, the algorithm \textsc{allHNetworks} only focuses on maximum acyclic agreement forests and, thus, we still have to show why we only have to consider acyclic agreement forests of minimum size. For instance, as depicted in Figure~\ref{fig-order}, it can happen that for a specific ordering of the input trees more reticulation edges have to be added when inserting leading input trees so that the resulting networks contain embedded trees that are necessary to obtain so-called \emph{hidden relevant networks} at the end. In the following, however, we will show that, if such a hidden relevant network for a specific ordering of input trees exists, this network has to be contained in a set of relevant networks calculated for another ordering of the input trees.

Now, before presenting the third lemma, we will first introduce a simple modification of the algorithm \textsc{allHNetworks}. Let $\cT$ be a set of rooted binary phylogenetic $\cX$-trees, then, \textsc{allHNetworks}$^*$ denotes a modification of the algorithm \textsc{allHNetworks} that considers for the insertion of an input tree $T_i\in\cT$ to so far computed networks all acyclic agreement forests of arbitrary size (instead of just those of minimum size).

\begin{lemma}
Let $\cT$ be a set of rooted binary phylogenetic $\cX$-trees. A relevant network $N$ for $\cT$ is calculated by calling \textsc{allHNetworks}$^*$($\cT$) if and only if it is calculated by calling \textsc{allHNetworks}($\cT$). 
\label{lem-three}
\end{lemma}

\begin{proof}

'$\Longleftarrow$': As each computational path of the algorithm \textsc{allHNetworks} is also conducted by the modified algorithm \textsc{allHNetworks}$^*$, each relevant network calculated by calling \textsc{allHNetworks}($\cT$) is obviously also calculated by calling \textsc{allHNetworks}$^*$($\cT$).\\

'$\Longrightarrow$': Here, we have to discuss why the algorithm \textsc{allHNetworks} has \textit{not} to consider non-maximum acyclic agreement forests leading to relevant networks. For this purpose, we will first show by induction on $n=|\cT|$ that, if for a specific ordering $\Pi$ of the input trees a relevant network $N^*$ can be only computed by applying a non-maximum acyclic agreement forest $\cF_i^*$, then, in this case, there exists a different ordering $\Pi^*$ computing $N^*$ by only taking components of maximum acyclic agreement forests into account.\\

\textbf{Base case.} The assumption, obviously, holds for $n=1$. For $n=2$ an agreement forest that is \textit{not} maximal cannot lead to relevant networks, since the insertion of a maximum acyclic agreement forest $\cF$ directly leads to a network whose reticulation number is smaller. This is, in particular, the case, since the algorithm inserts a reticulation edge for all components of an agreement forest, except $F_{\rho}$, and, thus, in this simple case, the hybridization number simply equals $|\cF|-1$. Note that, due to Lemma~\ref{lem-two}~and~\ref{lem-one}, in the case of two input trees, all relevant networks are calculated.\\

\textbf{Inductive step.} Now, let $\Pi^*=(T_1,\dots,T_i,\dots,T_{n})$, with $n>2$, be an ordering of input trees for which the algorithm \textsc{allHNetworks} calculates the set $\cN_{n-1}$ consisting  of all relevant networks for $\cT\setminus\{T_n\}$ and there exists a hidden relevant network $N^*$ for $\Pi^*$ that could only be computed by inserting reticulation edges for a \textit{non}-maximum acyclic agreement forest $\cF^*_i$ for an input tree $T_i$ ($1\le i<n$) and an embedded tree $T'_i$ of the network $N^*_{i-1}$ displaying $\{T_1,T_2,\dots,T_{i-1}\}$. Notice that this directly implies that in $N^*$ there exist $x>0$ reticulation edges only necessary for displaying both trees $T_i$ and $T_n$, where $x$ denotes the difference between $|\cF_i^*|$ and the size of a maximum acyclic agreement forest $\hat\cF^*_i$ for $T_i$ and $T_i'$, i.e., $x=|\cF^*_i|-|\hat\cF^*_i|$. In this case, however, as we will show in the following, $N^*$ can be also calculated by applying the algorithm to the ordering $\Pi=(T_1,\dots,T_{i-1},T_{i+1},\dots,T_{n},T_i)$, where $T_i$ is inserted right after $T_n$.

For this purpose, let $N_{n-2}$ be the relevant network displaying each tree except $T_i$ and $T_n$ in the same topological way as it is the case for $N^*$. More precisely, $N_{n-2}$ equals the network that is obtained from $N^*$ by first deleting a set of reticulation edges $E^*_i$, containing each edge that is not necessary for displaying an input tree in $\cT\setminus\{T_i\}$, then by deleting a set of reticulation edges $E^*_n$, containing each remaining edge that is not necessary for displaying an input tree in $\cT\setminus\{T_n\}$, and finally by suppressing all nodes of both in- and out-degree~$1$. Notice that we can calculate $N_{n-2}$ by applying the algorithm \textsc{allHNetworks} to $\cT\setminus\{T_i,T_n\}$  since, by induction hypothesis, the algorithm is able to calculate all relevant networks embedding $\cT\setminus\{T_n\}$.

Next, let $N_{n-1}$ be the relevant network displaying each tree except $T_i$ in the same topological way as it is the case for $N^*$. More precisely, $N_{n-1}$ equals the network that is obtained from $N^*$ by first deleting each reticulation edge that is not necessary for displaying an input tree in $\cT\setminus\{T_i\}$ and then by suppressing all nodes of both in- and out-degree~$1$. Notice that, based on $N_{n-2}$, due to both previous Lemmas~\ref{lem-two}~and~\ref{lem-one}, this network can be calculated by inserting the components of an specific acyclic agreement forest $\cF_n$ for $T_n$ and the embedded tree $T'_i$ of $N_{n-2}$ with $|\cF_n|=|E^*_n|+1$. Moreover, $T'_i$ is still contained in $N_{n-2}$, since for displaying this tree no reticulation edge is necessary that has been added during the insertion of $T_i$ and $T_n$ and, thus, would not exist in $N_{n-2}$.  

It still remains to show, however, why this acyclic agreement $\cF_n$ is of minimum size. For this purpose, we will establish a proof by contradiction showing that in this case we could construct a hybridization network $N'$ for $\cT$ providing a smaller reticulation number than $N^*$. In a first step, however, we have to recall each acyclic agreement forest that is used in $\Pi^*$ as well as in $\Pi$ in order to insert the two trees $T_i$ and $T_n$. The reader should keep in mind that by inserting a tree based on an acyclic agreement forest of size $k$, the algorithm \textsc{allHNetworks} inserts precisely $k-1$ reticulation edges.

\begin{itemize}
\item Regarding $\Pi^*$, first the tree $T_i$ is inserted by a non-maximum acyclic agreement forest $\cF_i^*$ of size $k^*_i$ and then the tree $T_n$ is inserted by a maximum acyclic agreement forest~$\cF_n^*$ of size $k^*_n$.
\item Regarding $\Pi$, first the tree $T_n$ is inserted by a maximum acyclic agreement forest $\cF_n'$ of size $k'_n$ and then the tree $T_i$ is inserted by a maximum acyclic agreement forest~$\cF_i$ of size $k'_i$.
\end{itemize}

Now, in order to establish a contradiction, let us assume that $k'_n=|\cF_n'|<|\cF_n|$. Notice that through $\cF_n$ the edge set $E^*_n$ is reinserted, which implies that $\cF_n$ has to contain precisely $|E^*_n|+1=k^*_n+x$ components, where $x$, as already mentioned above, denotes the difference between $|\cF_i^*|$ and the size of a maximum acyclic agreement forest $\hat\cF^*_i$ for $T_i$ and $T_i'$, i.e., $x=|\cF^*_i|-|\hat\cF^*_i|$. Regarding $\cF_n'$, this means that we could insert $T_i$ and $T_n$ to $\cN_{n-2}$ by inserting precisely $r_1=k'_n-1+k^*_i-x-1$ reticulations edges. Next, by considering the number of reticulation edges that are added in $\Pi^*$ for $T_i$ and $T_n$, which are $r_2=k_i^*-1+k^*_n-1$, we can establish the following inequation: $$r_1=k'_n-1+k^*_i-x-1<k_n^*+x-1+k^*_i-x-1=k_i^*-1+k^*_n-1=r_2.$$ In summary, this means that, if $|\cF_n'|<|\cF_n|$ holds, we could construct a network $N'$ with $r_1=r(N')<r(N^*)=r_2$ by inserting both trees $T_i$ and $T_n$ into $N_{n-2}$ in respect to $\hat\cF^*_i$ and $\cF_n'$, which implies that $N^*$ would not be a relevant network for $\cT$; a contradiction to the choice of $N^*$.

Lastly, based on $N_{n-1}$, again due to both previous Lemmas~\ref{lem-two}~and~\ref{lem-one}, the network $N^*$ can be calculated by inserting the components of an specific acyclic agreement forest for $T_i$ and some embedded tree of $N_{n-1}$. Notice that this acyclic agreement forest has to be of minimum size, since, otherwise, by simply taking only maximum acyclic agreement forests into account we could directly construct networks providing a smaller reticulation number than $N^*$. Again, this would directly imply that $N^*$ could not be a relevant network for $\cT$; a contradiction to the choice of $N^*$, which finally establishes the induction step.\\

Based on the induction above, we can make the following observation. If for a specific ordering of the input trees there exists a tree $T_i$ that has to be added by a \textit{non}-maximum acyclic agreement forest in order to enable an insertion of another input tree $T_j$ $(i<j)$, which is necessary for the computation of a relevant network $N$, then, in this case, we can compute $N$ by applying the algorithm \textsc{allHNetworks} to an ordering where $T_i$ is located after $T_j$. Thus, for each relevant network $N$ that could only be computed by our algorithm by applying \textit{non}-maximum acyclic agreement forests, there exists a certain ordering of the input trees such that our algorithm is able to compute $N$ by only taking maximum acyclic agreement forests into account. Finally, as a direct consequence, since our algorithm takes all possible orderings of input trees into account, our algorithm obviously guarantees the computation of all relevant networks without considering \textit{non}-maximum acyclic agreement forests. Thus, the correctness of Lemma~\ref{lem-three} is established.
\end{proof}

\begin{figure}[h!]
\centering
\begin{tabular}{cc}
\includegraphics[width = 7.2cm]{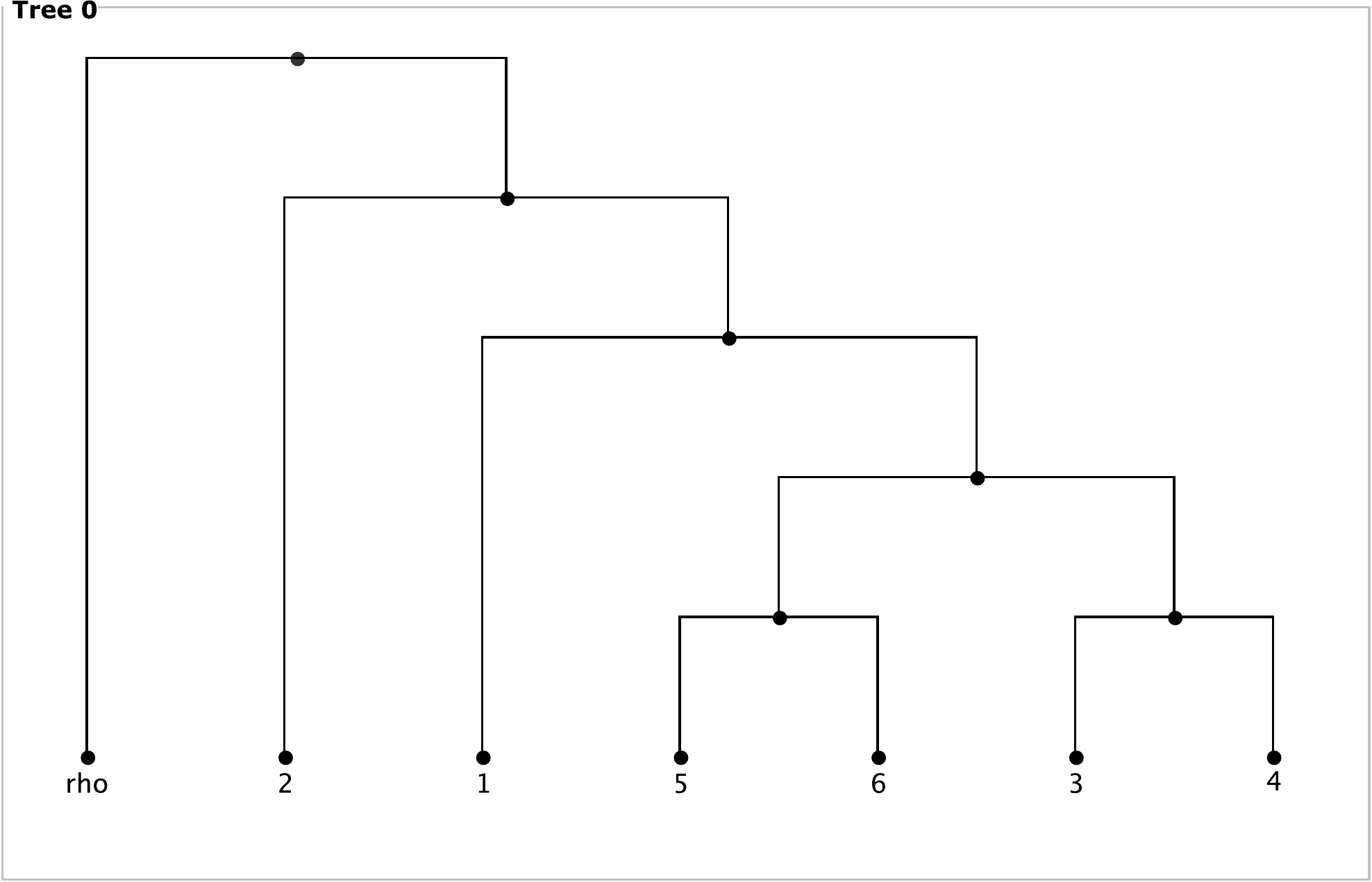}
&
\includegraphics[width = 7.2cm]{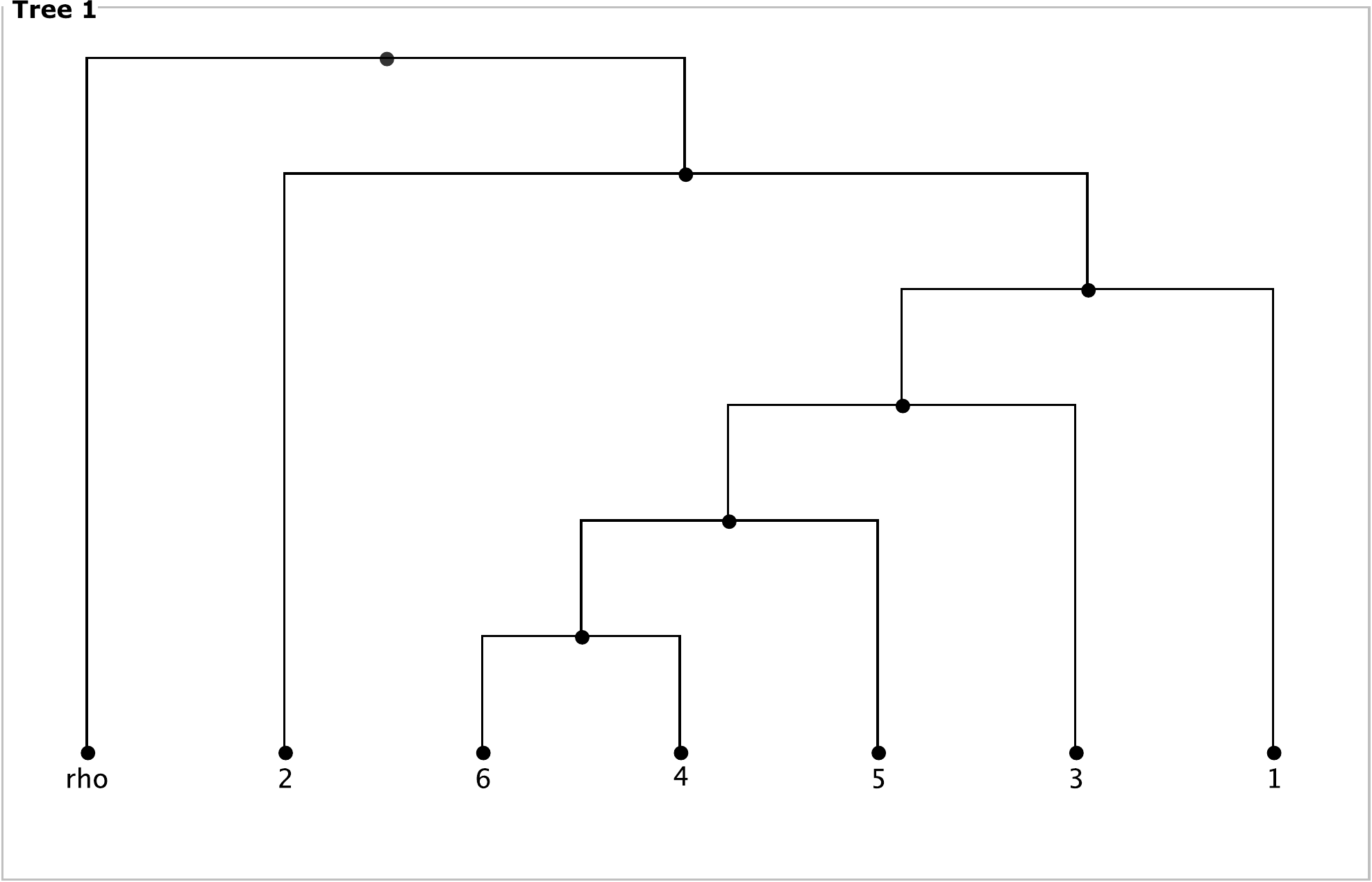}
\\
\includegraphics[width = 7.2cm]{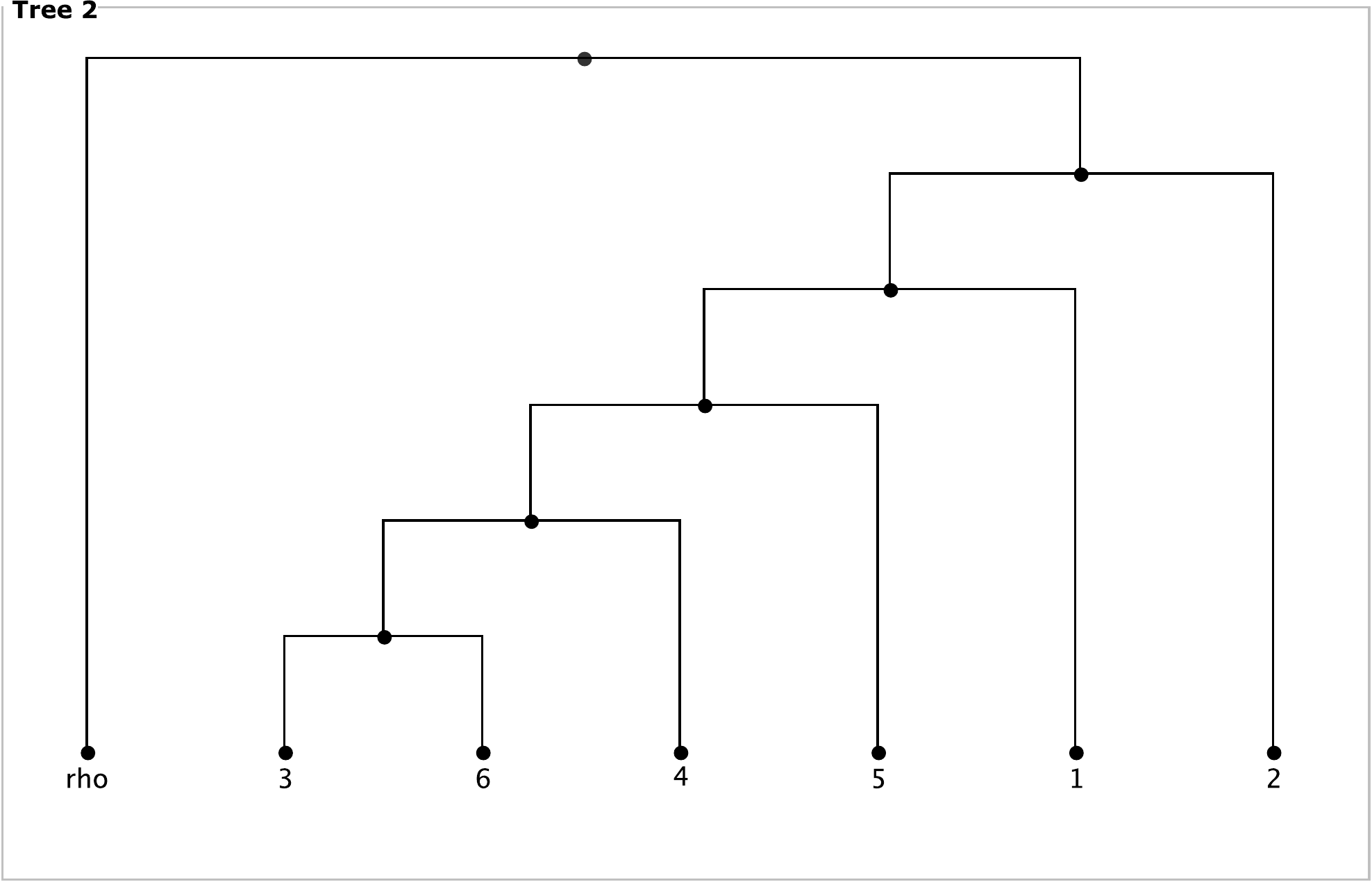}
&
\includegraphics[width = 7.2cm]{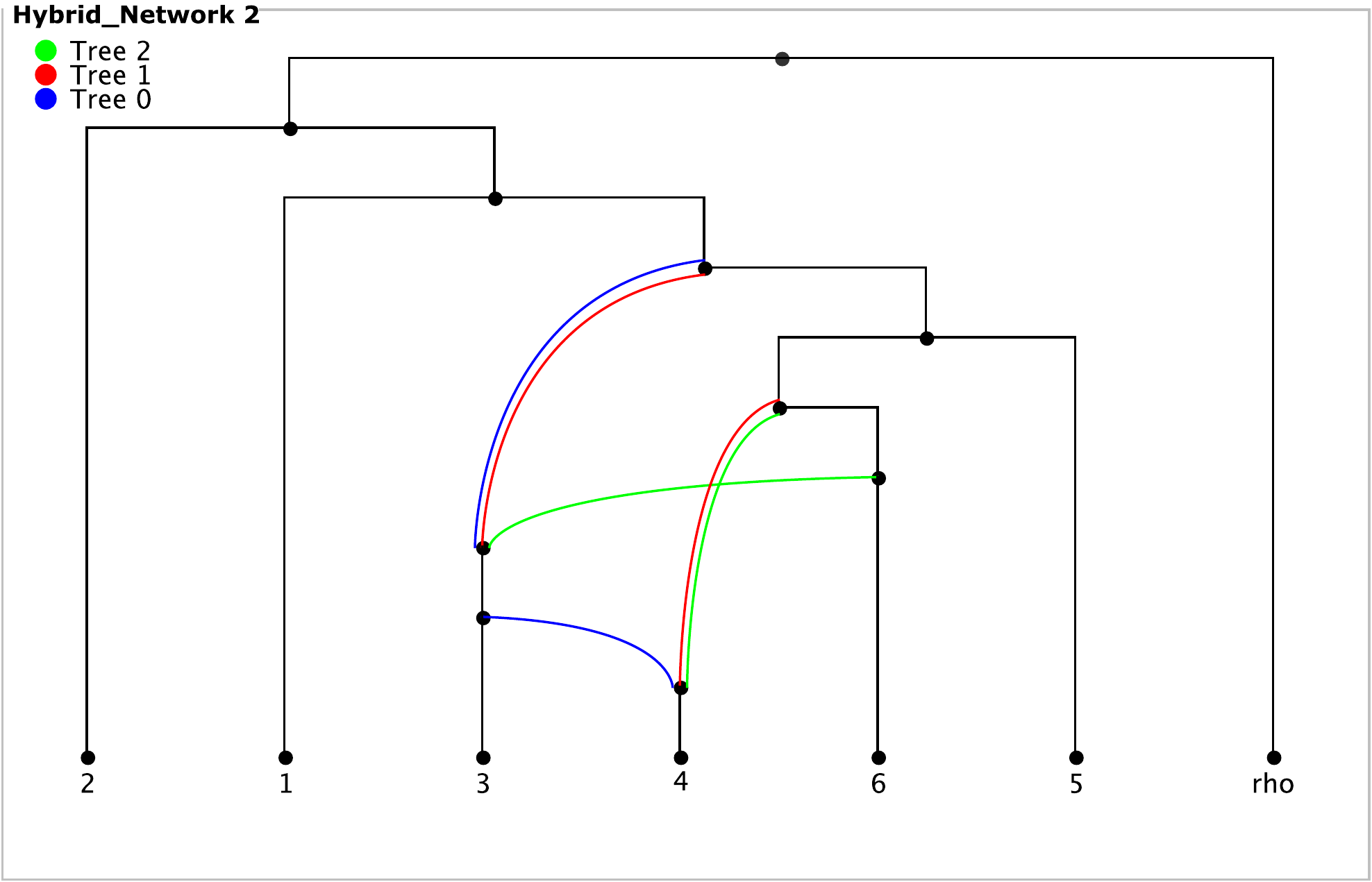}
\\
\includegraphics[width = 7.2cm]{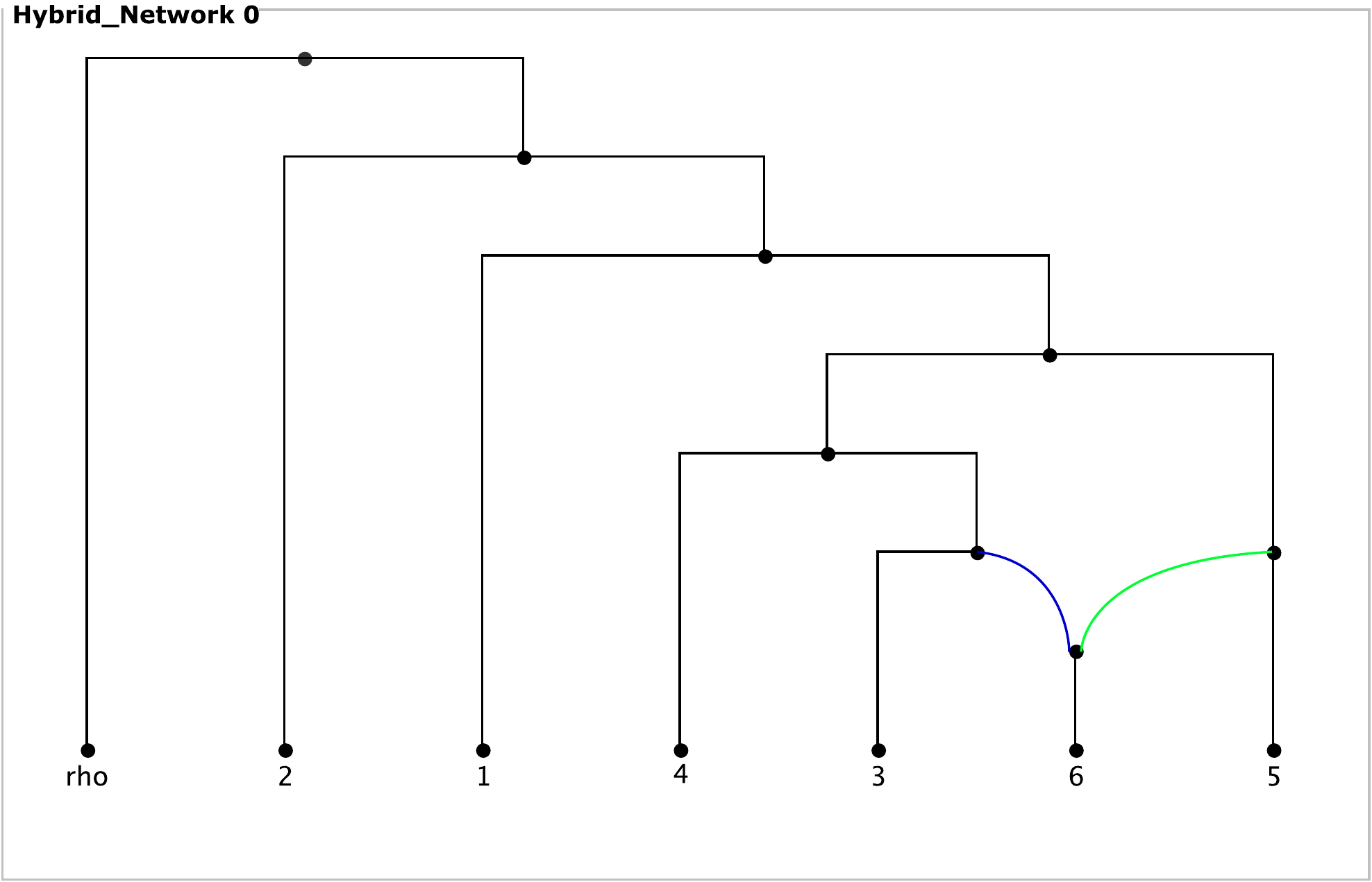}
&
\includegraphics[width = 7.2cm]{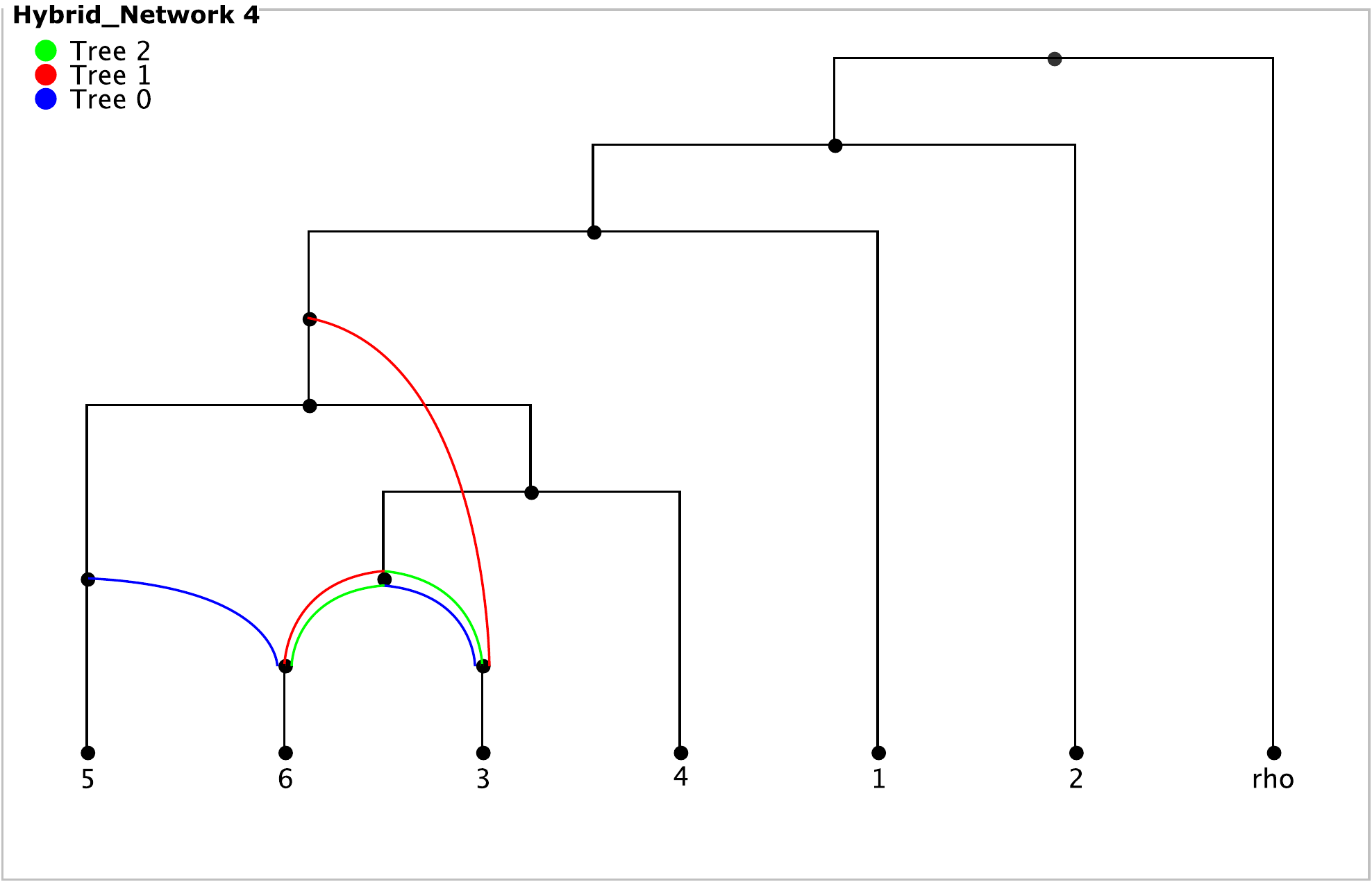}
\\
\end{tabular}
\caption[An example showing why the algorithm \textsc{allHNetworks} has to consider different orderings of the input trees]{An example showing why the algorithm \textsc{allHNetworks} has to consider different orderings of the input trees. By running the algorithm \textsc{allHNetworks} for the ordering $\Pi_1=(\text{\textit{Tree~0},\textit{Tree~1},\textit{Tree~2}})$ only those networks with hybridization number two, as the one denoted by \textit{Hybridization\_Network~2}, are computed providing a hybridization node whose subtree consists of taxon~4. This is the case, since the only maximum acyclic agreement forest for \textit{Tree~0} and \textit{Tree 1} is of size two containing the component consisting of the single taxon 4. To compute the network denoted as \textit{Hybridization\_Network~4} at bottom right, you have to apply the algorithm to the ordering $\Pi_2=($\textit{Tree~0},\textit{Tree~2},\textit{Tree~1}$)$, since now, in a first step, by adding \textit{Tree~2} to \textit{Tree~0} the network at bottom left, denoted by \textit{Hybridization\_Network~0}, is computed. Based on this network you can select an embedded tree $T'$ by choosing its blue in-edge. As a direct consequence, the only maximum acyclic agreement forest for $T'$ and \textit{Tree~2} is of size two containing the component consisting of the single taxon 3 and, thus, by adding this component to the network \textit{Hybridization\_Network~0}, finally, the network \textit{Hybridization\_Network~4} is computed. Regarding the first mentioned ordering $\Pi_1$, this network could only be computed by our algorithm by considering the \textit{non}-maximum acyclic agreement forest for \textit{Tree~0} and \textit{Tree~1} of size three containing the two components consisting of the single taxa 3 and 6.} 
\label{fig-order}
\end{figure}

Now, based on the fact that the algorithm \textsc{allMAAFs} returns all maximum acyclic agreement forests for two binary phylogenetic $\cX$-trees \cite[Theorem 2]{Scornavacca2012}, by combining Lemma \ref{lem-two}--\ref{lem-three} the correctness of Theorem~\ref{th-one} is established. 

More precisely, this is the case, because due to Lemma~\ref{lem-two} we can derive a network displaying a further input tree $T_i$ from an acyclic agreement forest $\cF$ for $T_i$ and an embedded tree of a so far computed network. Moreover, due to Lemma~\ref{lem-one}, by taking all orderings of the input trees into account, for this purpose it suffices to consider only acyclic agreement forests of minimum size. Furthermore, by considering all possible ways of how such a maximum acyclic agreement forest $\cF$ can be inserted (cf.~Lemma~\ref{lem-one}), the algorithm \textsc{allHNetworks} calculates each network displaying $\cF$. Now, since $T_i$ is added to all so far computed networks by taking all maximum acyclic agreement forests for all embedded trees into account, all networks embedding $T_i$ are calculated. Consequently, by adding all input trees sequentially for all orderings in this way all relevant networks for all input trees are calculated.
\end{proof}

\clearpage
\section{Runtime of \textsc{allHNetworks}}
\label{sec-rt}

In order to analyze the theoretical worst-case runtime of the presented algorithm \textsc{allHNetworks}, we have to discuss the complexity of three major steps including the computation of embedded trees, the computation of all maximum acyclic agreement forests of size $k$, and the computation of all possible reticulation edges that can be added for a given maximum acyclic agreement forest. Given an ordering of the input trees, each of those major steps has to be applied sequentially to each input tree in order to insert this tree into a set of so far computed networks. At the beginning, when adding the second input tree, this set of networks only consists of the first tree of the ordering. However, as shown in the upcoming part, this set grows exponentially in the number of input trees. 

\begin{theorem}
The theoretical worst-case runtime of the algorithm \textsc{allHNetworks} for computing all relevant networks for a set $\cT$ of rooted binary phylogenetic $\cX$-trees with minimum hybridization number $k$ is $$O\left(n!\left(2^k\tbinom{|E|}{k}k!\tbinom{|V|}{2}k\right)^{n-1}\left(|V|+3^{|\cX|}\right)\right),$$ where $E$ denotes the edge set and $V$ denotes the node set of a binary tree in $\cT$. 
\label{lem-run}
\end{theorem}

\begin{proof}

To show the correctness of Theorem \ref{lem-run}, we divide the stated runtime estimation into four parts A--D and discuss each of those parts separately:

$$O\left(\underbrace{n!}_{A}\left(\underbrace{2^k}_{B}\underbrace{\tbinom{|E|}{k}k!}_{C}
\underbrace{\tbinom{|V|}{2}k}_{D}\right)^{n-1}\left(\underbrace{|V|}_{B}+\underbrace{3^{|\cX|}}_{C}\right)\right)$$\\

\textbf{Part A.} Since different orderings of the input trees can lead to different relevant networks, the insertion of the trees has to be performed for all $n!$ possible orderings.\\

\textbf{Part B.} The number of embedded trees of a network is at most $2^r$~where~$r$ denotes its number of hybridization nodes. This upper bound, however, is achieved only if each hybridization node has in-degree $2$. Otherwise, if a hybridization node has more than two in-edges, the number of embedded trees is smaller as only one of those edges can be part of an embedded tree. Moreover, extracting a tree from a given network is a process of rather low complexity, which can be solved by iterating a constant number of times over all nodes of the network. Thus, the complexity of extracting one certain embedded tree is linear in the number of nodes.\\ 

\textbf{Part C.} The number of all maximum acyclic agreement forests of size $k$ for two input trees $T_1$ and $T_2$ can be estimated by $O(\tbinom{|E(T_1)|}{k})$. In practice, however, this number is clearly smaller since, in general, less than $k$ hybridization events, say $r$, are necessary for the insertion of one of the input trees. Moreover, only a few number of all $\tbinom{|E|}{r}$ possible sets of components fulfills the definition of an acyclic agreement forest. 

Given an agreement forest of size $k$, there exist at most $k!$ acyclic orderings. Note that, similar to the number of all maximum acyclic agreement forests, there exist, in general, clearly less orderings. This number, however, can be large if there are a lot of components consisting of isolated nodes. The runtime for the computation of those maximum acyclic agreement forests is stated in the work of Scornavacca~\emph{et al.}~\cite[Theorem 3]{Scornavacca2012} by $O(3^{|\cX|})$, where $\cX$ denotes the taxa set of each input tree.\\

\textbf{Part D.} As mentioned during the presentation of the algorithm \textsc{allHNetworks}, a component of a maximum acyclic agreement forest can potentially be added in several ways to a so far computed network $N$. This number is, obviously, bounded by $\tbinom{|V|}{2}$ where $V$ denotes the set of nodes corresponding to $N$. In practice, however, this number is clearly smaller, since only a small fraction of all possible node pairs enable a valid embedding of an input tree. Lastly, given a source and a target node, a new reticulation edge can be simply added by performing a constant number of basic tree operations.
\end{proof}

\section{Speeding Up the Algorithm \textsc{allHNetworks}}
\label{sec-speed}

To handle the huge computational effort, which is indicated in Theorem~\ref{lem-run}, it is very important to implement the algorithm in an efficient way. This can be done by parallelizing its execution on distributed systems, by initially applying certain reductions to the input trees, and by reducing the computation of isomorphic networks.

\subsection{Parallelization}
\label{sec-para}

In order to improve the practical runtime of our algorithm, each exhaustive search looking for relevant networks with minimum hybridization number $k$ can be parallelized as follows. As described in Section~\ref{sec-alg}, the insertion of a tree $T_i$ to a so far computed network results in several new networks which are then processed by inserting the next input tree $T_{i+1}$ of the chosen ordering (cf.~Fig.~\ref{fig-exSearch}). 

\begin{figure}[t]
\centering
\includegraphics[width = 12cm]{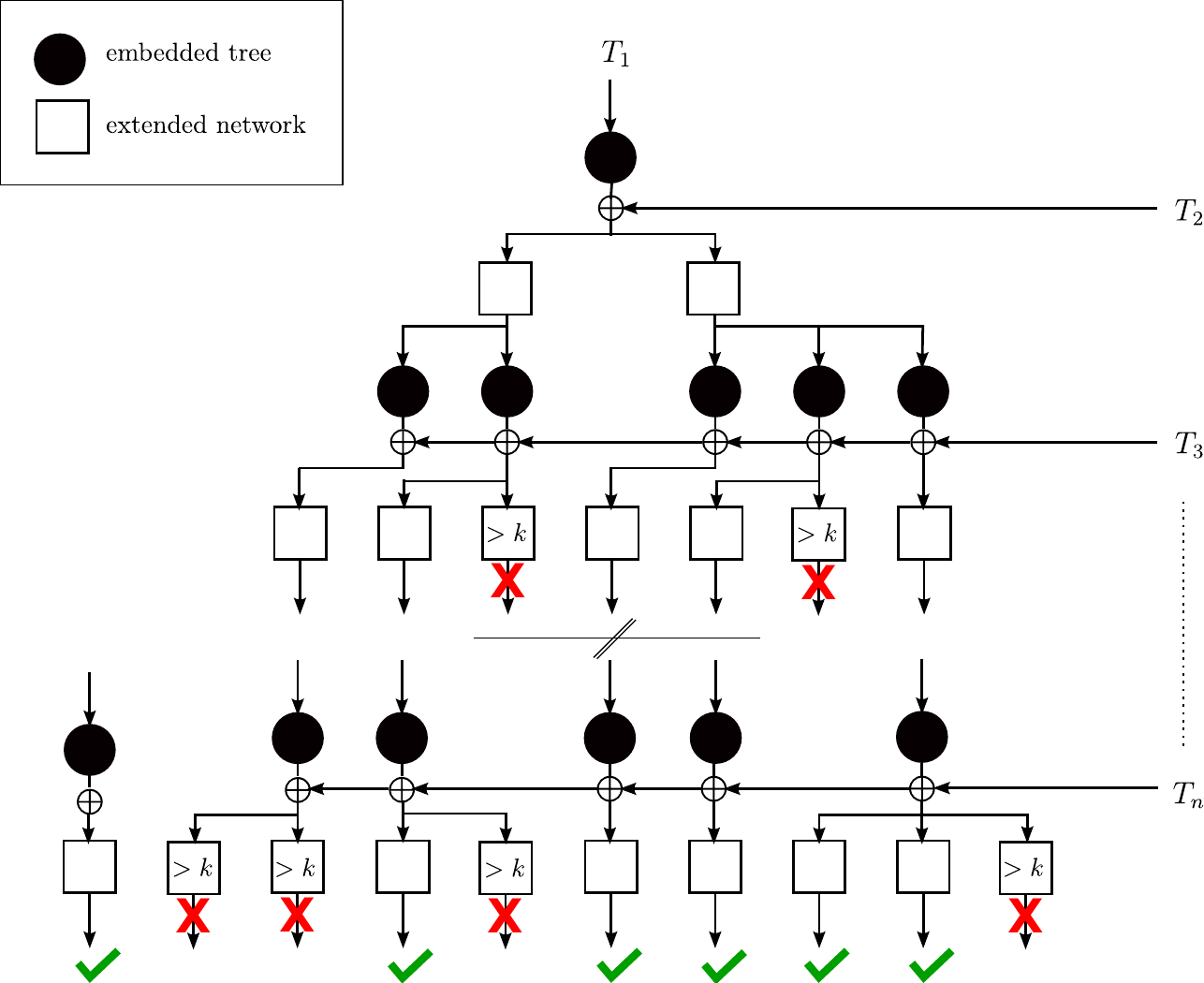}
\caption[An illustration of computational paths generated by the algorithm \textsc{allHNetworks}]{An illustration of how the insertion of the input trees is conducted by the algorithm \textsc{allHNetworks} in respect of the parameter $k$ bounding the maximal reticulation number of resulting networks. Beginning with the first input tree $T_1$, repeatedly, first, an embedded tree $T'$ of a so far computed network $N$ is selected, and, second, the current input tree $T_i$ is inserted into $N$ by sequentially adding the components of a maximum acyclic agreement forest for $T'$ and $T_i$. As soon as the reticulation number of a so far computed network exceeds $k$ one can be sure that this network cannot lead to a network whose reticulation number is smaller than or equal to $k$ and, thus, the corresponding computational path can be early aborted.} 
\label{fig-exSearch}
\end{figure}

Since the processing of networks runs independently from each other, these steps can be parallelized in a simple manner. Notice, however, that, based on the reticulation number of so far computed networks, each of those steps is more or less likely to result in relevant networks. Thus, one can set up a priority queue to process the most promising networks first, which depends, on the one hand, on the number of so far embedded input trees and, on the other hand, on its current reticulation number. One should keep in mind, however, that such a priority queue can only speed up the computation of the hybridization number, since, only in this case, the search can be aborted immediately as soon as the first relevant network has been calculated. Otherwise, if one is interested in all relevant networks, each network has to be processed anyway until either it can be early aborted (which is the case if the reticulation number exceeds $k$) or it results in relevant networks.

\subsection{Reductions rules}
\label{sec-redRules}

In order to reduce the size of the input trees, before entering the exhaustive part of the algorithm, one can apply the particular reduction rules that are, on the one hand, the subtree reduction, following the work of Bordewich and Semple~\cite{Bordewich-2005}, and, on the other hand, the cluster reduction, following the work of Baroni~\emph{et al.}~\cite{Baroni-2006} and Linz~\cite{Linz-2008}. \\

\textbf{Subtree reduction.} Let $\cT$ be a set of rooted binary phylogenetic $\cX$-trees, then the subtree reduction transforms all of those trees into a set $\cT'$ of rooted binary phylogenetic $\cX$-trees by replacing each maximal pendant subtree $T'$ of size $\ge 2$ occurring in all trees of $\cT$. More precisely, let $v$ be the root of such a maximal pendant subtree $T'$. Then, in each tree of $\cT$, first all nodes that can be reached from $v$ are deleted and afterwards $v$ is labeled by a new taxon $a\not\in\cX$. Notice that, in order to undo the subtree reduction at a given time, one has to keep track which of these new taxa belongs to which common subtree.\\

\textbf{Cluster reduction.} Let $\cT$ be a set of rooted binary phylogenetic $\cX$-trees and let $A\subset \cX$ be a \emph{cluster} with $A\ge 2$ such that for each tree $T_i$ in $\cT$ there exists a specific node $v_i$ with $\cL(v_i)=A$. Then, the cluster reduction separates $\cT$ into two tree sets $\cT|_A$ and $\cT_a$, where $\cT|_A$ contains each tree $T_i|_A$ and $\cT_a$ contains each tree $T_i$ where $T(v_i)$ is replaced by a new taxon $a$. More precisely, the tree set $\cT_a$ is obtained from $\cT$ by first deleting from each tree $T_i$ all nodes that can be reached from $v_i$ and then by labeling $v_i$ by a new taxon $a\not\in\cX$. Notice that, in order to reattach those clusters back together at a given time, one has to keep track which of these new taxa belongs to which common cluster.\\

Hence, the cluster reduction cuts down the set of input trees into all minimum common clusters whose relevant networks can then be computed independently by running the presented algorithm for each of those clusters separately. Consequently, the cluster reduction usually provokes a significant speedup, because often a problem of high computational complexity can be separated into several subproblems providing low computational complexities, which can be solved efficiently on its own. Notice that, in Section~\ref{sec-cr}, we give a proof showing that the cluster reduction is save for multiple rooted binary phylogenetic $\cX$-trees $\cT$, which means that $h(\cT)$ corresponds to the sum of the minimum hybridization numbers each calculated for a different common cluster.

However, when applying the cluster reduction to a set $\cT$ of rooted binary phylogenetic $\cX$-trees, in order to obtain a set consisting of all relevant hybridization networks $\cN$ displaying $\cT$, due to the following observation one still has conduct further combinatorial steps. Let $A\subset \cX$ be a common cluster of $\cT$ and let $\cT_a$ be the set of trees obtained from $\cT$ by replacing each cluster $A$ through a leaf labeled by taxon $a\not\in\cX$. Then, in a further step, one still has to reattach the networks computed for $\cT|_A$ and $\cT_a$, shortly denoted by $\cN_A$ and $\cN_a$, respectively, as follows. 

First, replace each taxon $a$ of a network in $\cN_A$ by each network in $\cN_a$ resulting in a set of networks $\cN_{A,a}$. However, due to the following observation this set $\cN_{A,a}$ might be just a subset of all relevant networks $\cN$. Since $\cN_A$ and $\cN_a$ are calculated separately, the source node of a reticulation edge is caught in $\cN_A$ and $\cN_a$, respectively. This means, in particular, that, regarding the set of reattached networks $\cN_{A,a}$, each network whose source node of a specific edge $e$ could be also located outside of its subgraph referring to $\cN_A$ or $\cN_a$, is missing (cf.~Fig.~\ref{fig-caughtEdges}). For example, regarding Figure~\ref{fig-shiftUp}, the network at the bottom left would not be calculated, which is due to the fact that the blue in-edge referring to node $6$ could not ``leave'' the common cluster $\{2,3,4,9\}$. However, keep in mind, that, as proven in Section~\ref{sec-cr}, this fact does not have an impact on the calculation of the minimum hybridization number $h(\cT)$ and, as demonstrated in the following, we can still generate the set $\cN\setminus\cN_{A,a}$ of missing networks by applying further linking patterns.

\begin{figure}[t]
\centering
\includegraphics[width = 10cm]{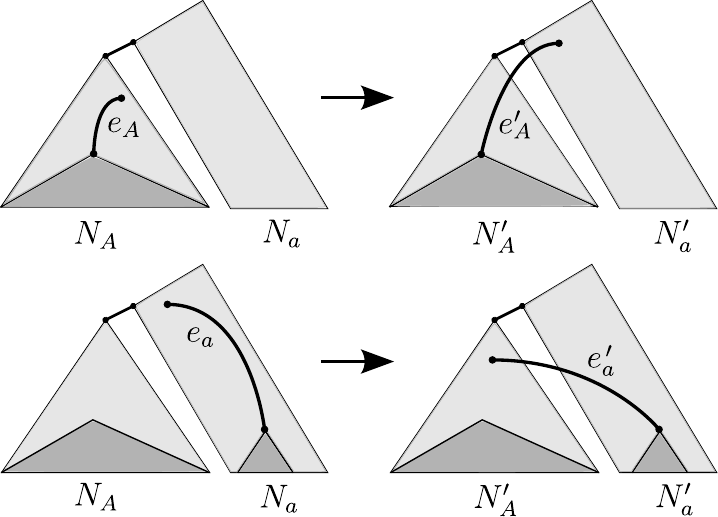}
\caption[An illustration of an caught edge within a cluster network]{An illustration of an edge in $\cN_{A,a}$ that is caught in a subgraph corresponding to $\cN_A$ and $\cN_a$, respectively. This means in particular that both relevant networks on the right hand side are missing in $\cN_{A,a}$.} 
\label{fig-caughtEdges}
\end{figure}

Let $N_a$ and $N_A$ be a network of $\cN_a$ and $\cN_A$, respectively. Then, this second step generates each missing network by performing \emph{legal shifting steps} reattaching reticulation edges to subgraphs beyond the border of two joined networks being part of $N_A$ and $N_a$, respectively. More precisely, we call a \emph{legal shifting step of a reticulation edge $e$} that is part of a subgraph corresponding to $N_A$ (resp. $N_a$), if it is possible to reattach $e=(x,y)$ to a node $s\ne x$ located in the subgraph corresponding $N_a$ (resp. $N_A$), so that still all input trees are displayed in the resulting network $N_{A,a}'$. Note that by saying reattaching we mean that first $e$ is deleted, then a new edge $(s,y)$ is inserted, and finally all nodes of both in- and out-degree one are suppressed. Now, in order to guarantee the computation of all relevant networks, for each network in $\cN_{A,a}$, one simply has to take all combinations of legal shifting steps into account.

For a better illustration of this concept, we will describe some linking-patterns that can be used to apply legal shifting steps in respect to a node $v$ and a subset $\cT'$ of all input trees $\cT$. In general, those pattern can be separated into three different types. A linking-pattern of Type A can be used to shift reticulation edges downwards in the given network, which means that the new source will be a successor of the original source node being part of a network separately calculated for a particular cluster of $\cT$. Similarly, a linking-pattern of Type B can be used to shift reticulation edges upwards in the given network, which means that the new source will be a predecessor of the original source node being part of a network separately calculated for a particular cluster of $\cT$. Once an edge has been shifted in terms of a pattern of Type A or Type B (or Type C), one can apply an additional linking-pattern of Type C (as defined below).

Just for convenience, in the following we will assume that the set of input trees $\cT$ only consists of two trees, which means that each reticulation edge is only necessary for the embedding of one of both trees and not for more than one tree (which obviously could be the case if $\cT$ contains more than two trees). Therefore, let $N$ be a relevant network displaying two rooted binary phylogenetic $\cX$-trees $T_1$ and $T_2$ and let $E_i$ be an edge set referring to $T_i\in\{T_1,T_2\}$.\\

\textbf{Linking-pattern of Type A.} Let $e=(x,y)$ be a reticulation edge of $E_i$ and let $\cP=(s_0,s_1,\dots,s_k)$ be a path in $N$ in which $s_0=x$, $y\not\in\cP$, and $v=s_i$ with $0<i\le k$. Moreover, let the out-degree of each node $s_i$, with $0<i\le k$, in $N|_{E_1,\cX}$ be $1$. Then, we can conduct a legal shifting step by first pruning $e$ and then by reattaching it to any node of $\cP$ (except $s_0$) (cf.~Fig.~\ref{fig-shiftDown}).\\

\textbf{Linking-pattern of Type B.} Let $e=(v,y)$ be a reticulation edge of $E_i$ and let $\cP=(s_0,s_1,\dots,s_k)$ be a path in $N$ in which $s_k=v$. Moreover, let the out-degree of each node $s_i$, with $0\le i<k$, in $N|_{E_1,\cX}$ be $1$. Then, we can conduct a legal shifting step by first pruning $e$ and then by reattaching it to any node of $\cP$ (except $s_k$) (cf.~Fig.~\ref{fig-shiftUp}).\\

\textbf{Linking-pattern of Type C.} Let $x$ be the source node of a reticulation edge $e_s$ that has already been shifted by applying a pattern of Type A or Type B. Moreover, let $e_t=(x,y)$ be an out-going tree edge of $x$ not necessary for displaying $T_i$. Then, we can conduct a legal shifting step by first pruning $e$ and then by reattaching it to $y$ (cf.~Fig.~\ref{fig-shiftDown}).\\

Now, let $N$ be a network of the set $\cN_{A,a}$ as defined above. Moreover, let $v$ be the root of the subgraph corresponding to $N_a$. Then, by applying those three linking-patterns to each network in $\cN_{A,a}$ and repeatedly to all resulting networks, one can produce the missing set of relevant networks. 

Note that, when applying a linking-pattern of Type B, the initial node $v$ might gets suppressed, if its in- and out-degree is $1$. In such a case, $v$ has to be redefined by the target of its out-going edge. Moreover, once an edge has been shifted downwards, one has to take care not shifting it back again upwards (and vice versa). This means, in particular, that edges that have been shifted in terms of a linking-pattern of Type A or B must not be shifted again by applying of one those two patterns. 

Lastly, by applying those linking patterns, the resulting networks not necessarily have to match the definition of a relevant network as given in Section~\ref{sec-pre}. Thus, one additionally has to apply the following two modifications.\\

\textbf{Modification of Type A.} By the linking patterns from above one automatically generates multifurcating nodes. Consequently, in order to turn those nonbinary networks into binary networks, one still has to resolve these nodes in all possible ways.\\

\textbf{Modification of Type B.} Moreover, by applying a linking-pattern of Type B, one can attach an edge $e$ to a hybridization node, which consequently means that a network is generated containing hybridization nodes of out-degree larger than one. As a consequence, one either has to reject those networks or, if such a hybridization node provides an in-edge $e_h$ that can be used for displaying the same set of trees as for $e$, one can first split $e_h$ and then attach $e$ to the new inserted node (cf.~Fig.~\ref{fig-shiftUp}).

\begin{figure}
\centering
\begin{tabular}{cc}
\includegraphics[height = 6cm, width = 8cm]{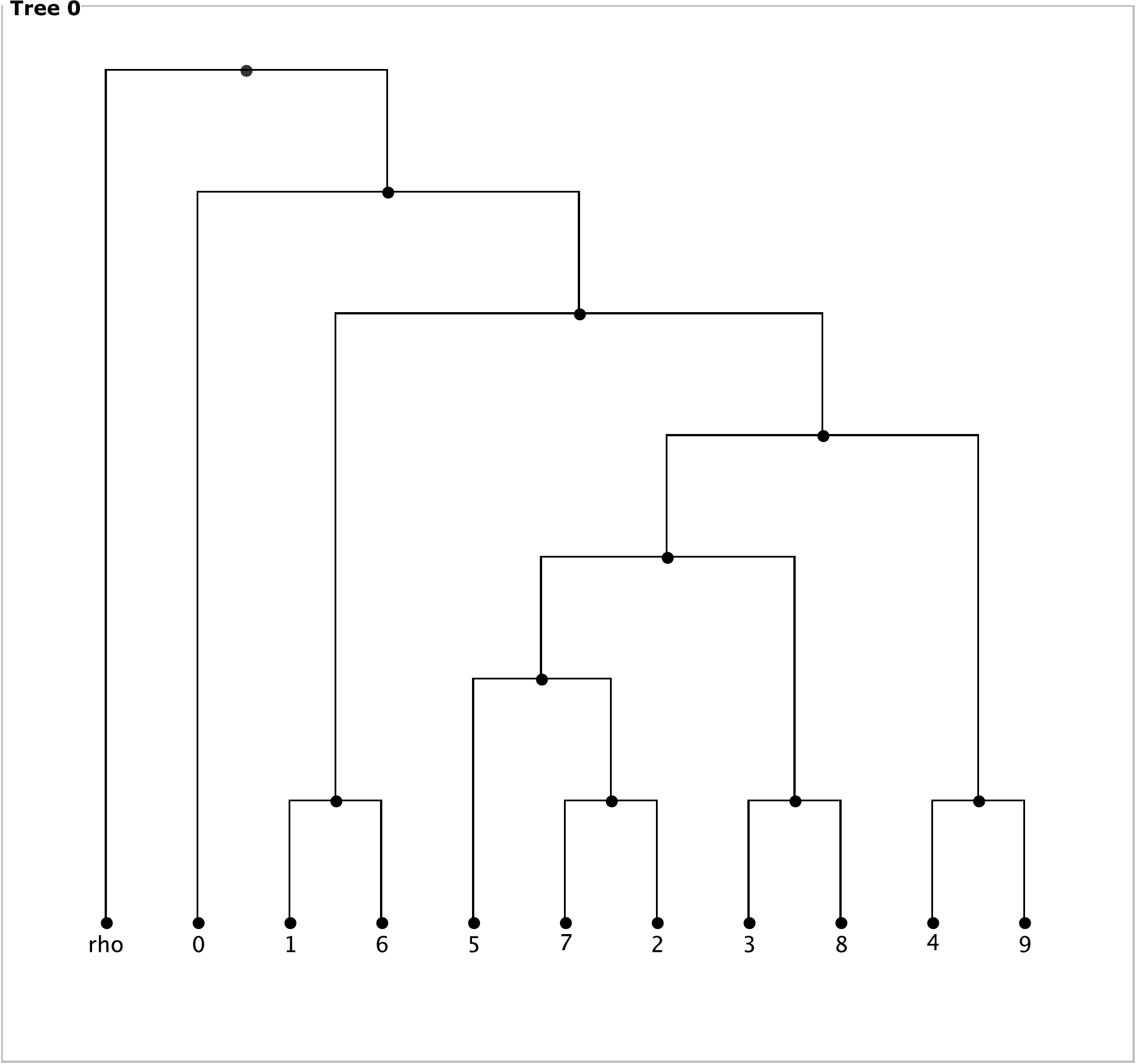}
&
\includegraphics[height = 6cm, width = 8cm]{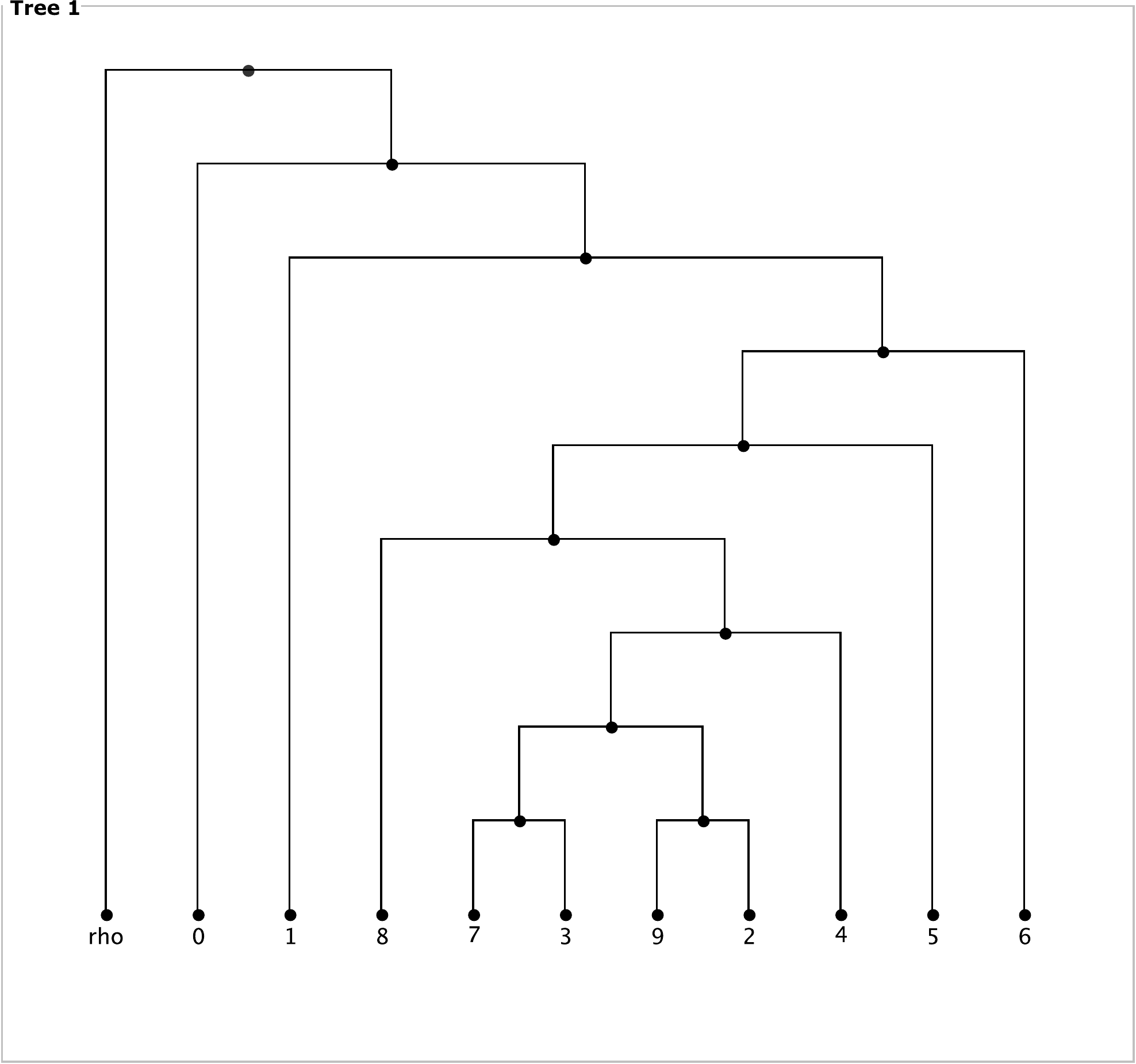}
\\
\includegraphics[height = 6cm, width = 8cm]{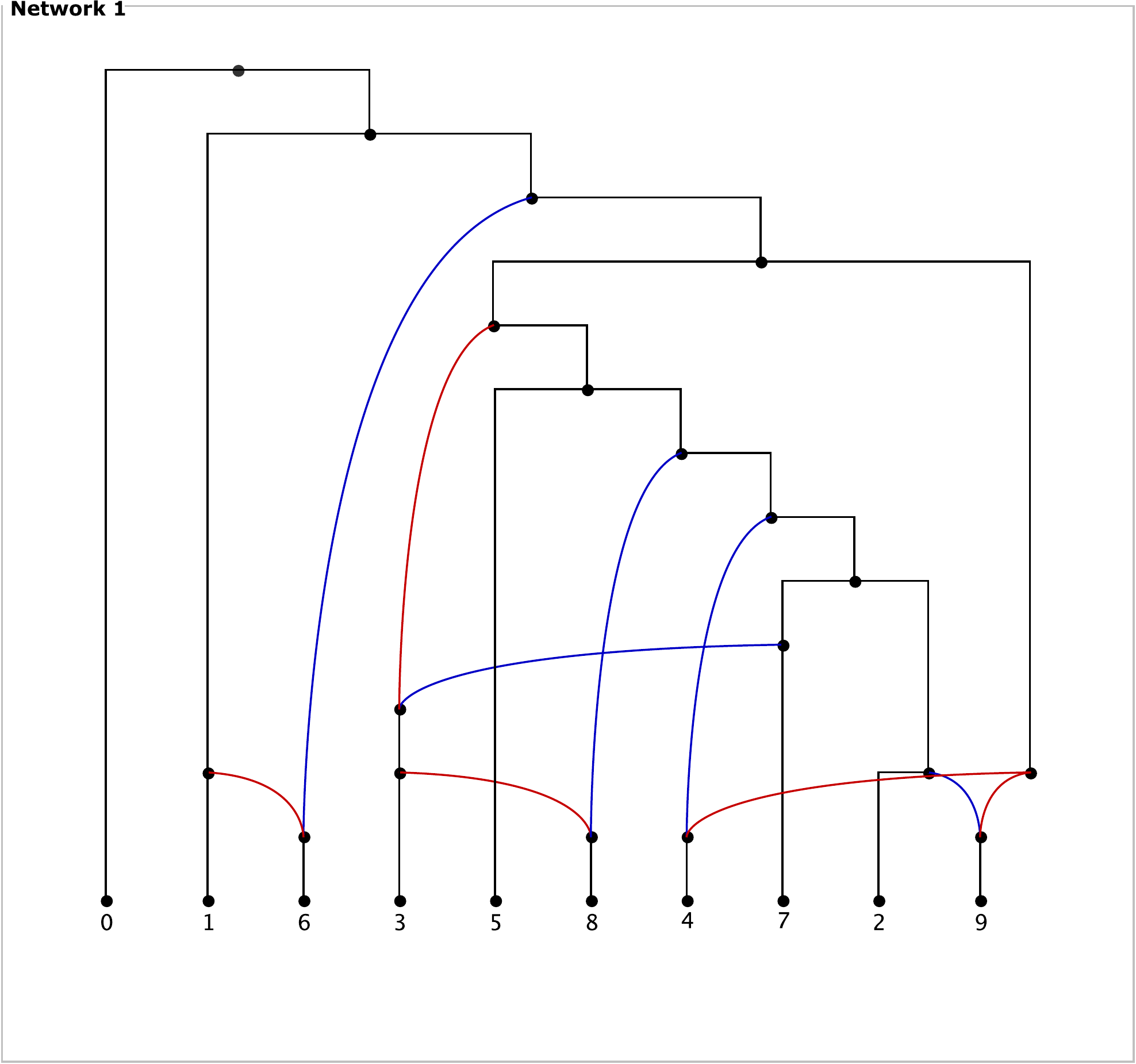}
&
\includegraphics[height = 6cm, width = 8cm]{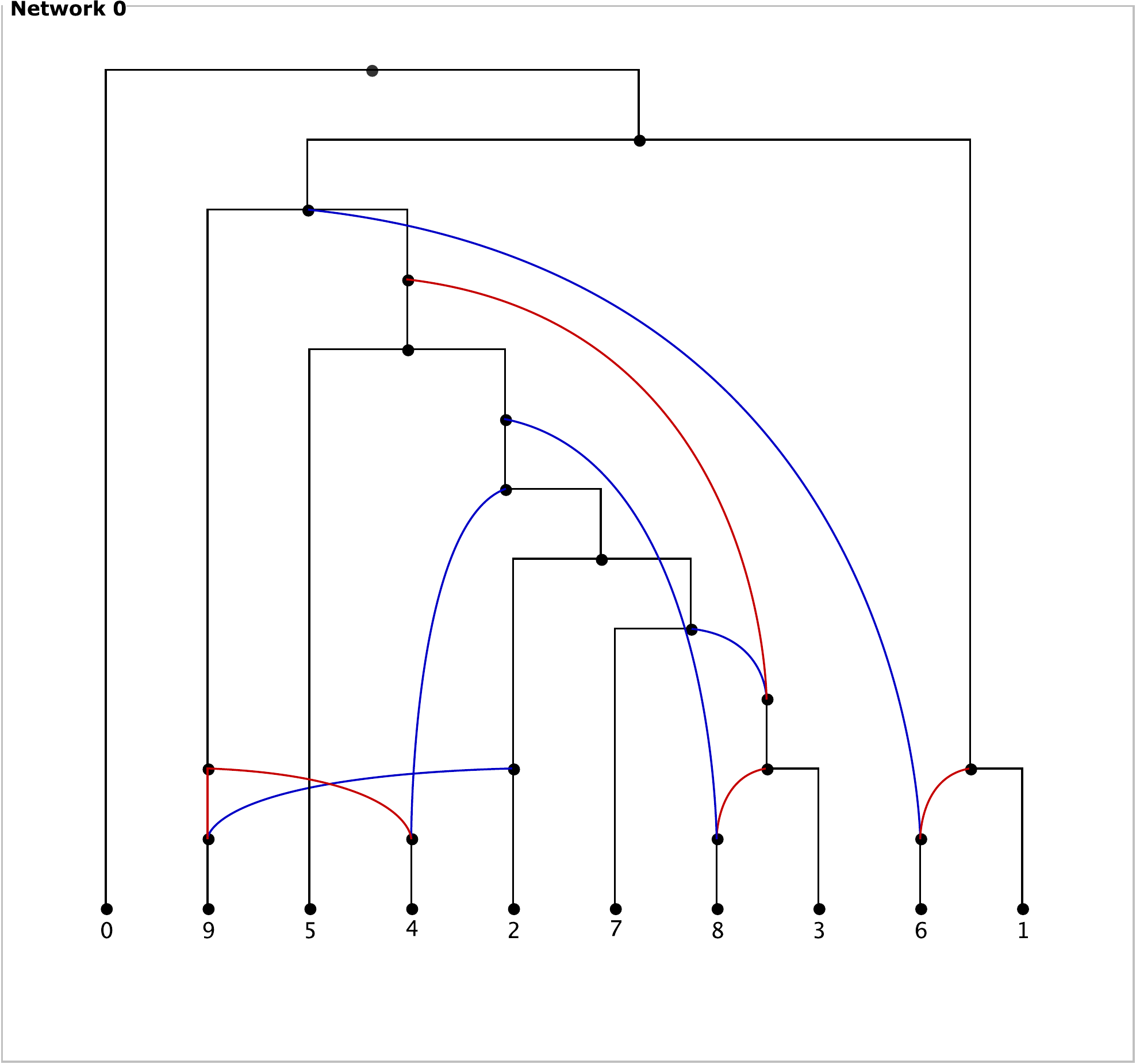}
\\
\includegraphics[height = 6cm, width = 8cm]{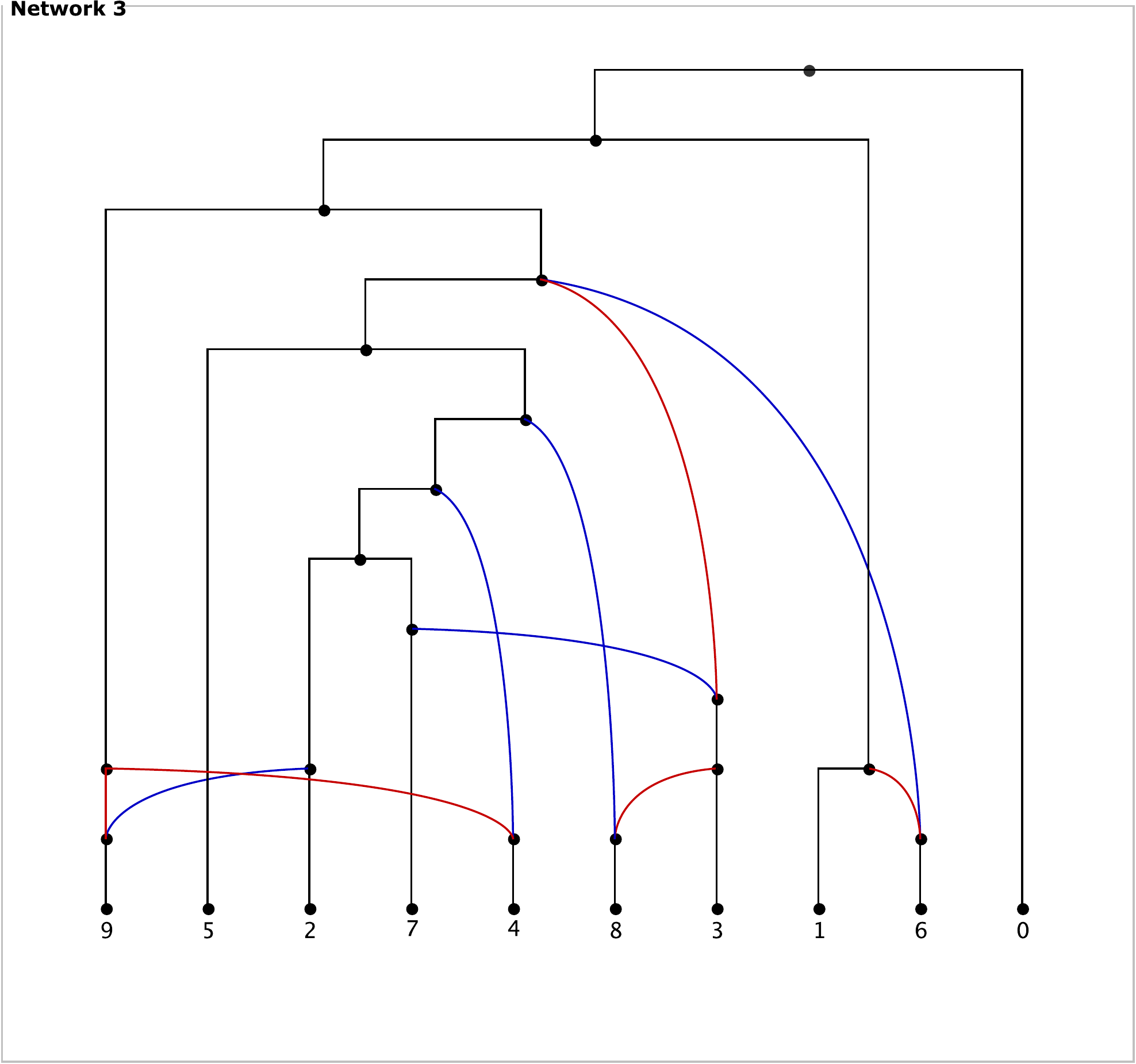}
&
\includegraphics[height = 6cm, width = 8cm]{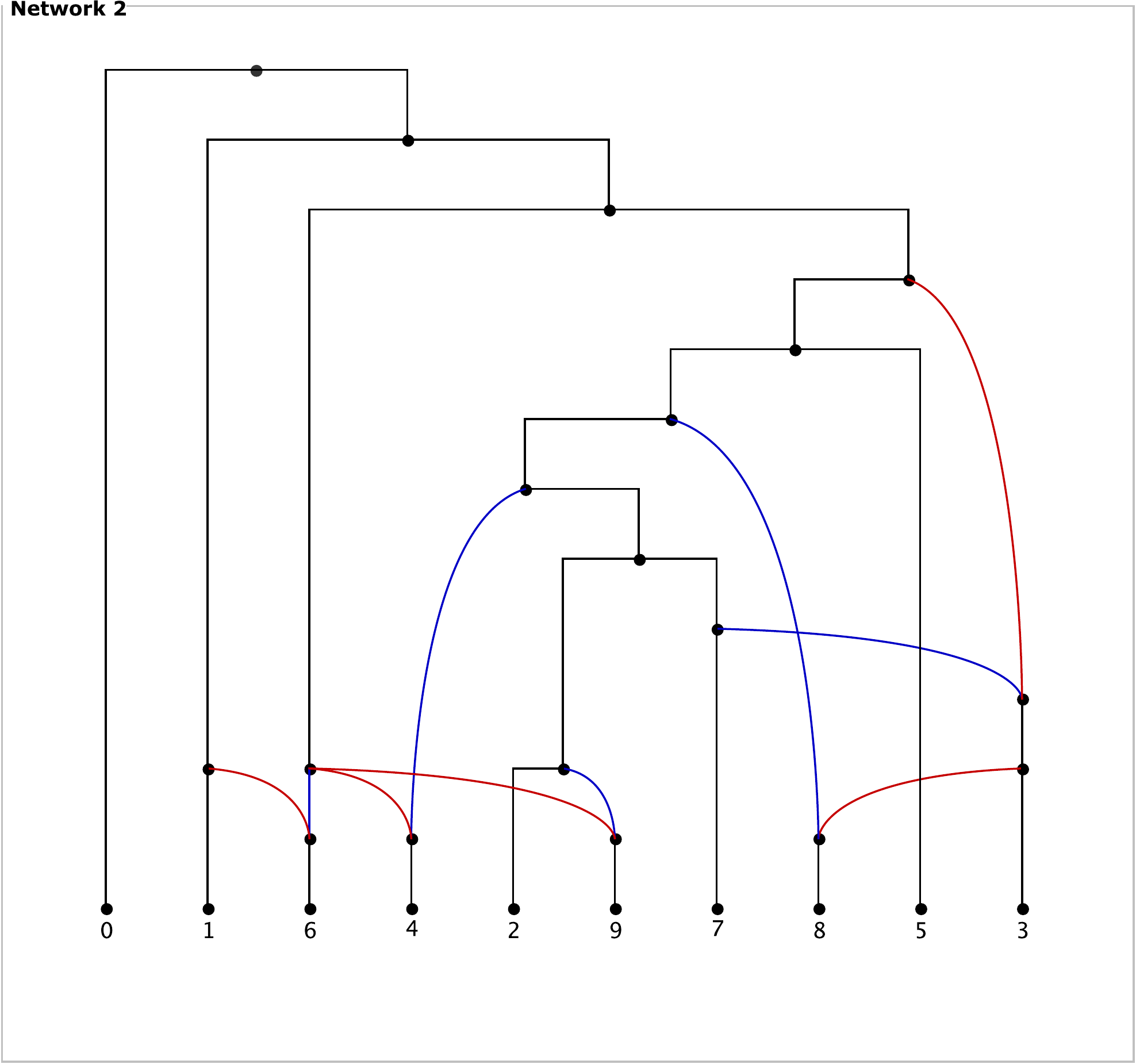}
\\
\end{tabular}
\caption[An illustration of the application of linking patterns - Part 2] {(\textbf{Top}) Two rooted binary phylogenetic $\cX$-trees sharing the common cluster $\{2,3,4,5,7,8,9\}$. (\textbf{Rest}) Minimum hybridization networks displaying both trees from the top where red edges refer to the left and blue edges to the right tree. Both networks, Networks 0~and Networks~2, can be obtained from Network~1 by applying a linking-pattern of Type A, whereas Network~3 can be obtained from Network~0 by applying a linking-pattern of Type C.} 
\label{fig-shiftDown}
\end{figure}

\begin{figure}
\centering
\begin{tabular}{cc}
\includegraphics[width = 8cm]{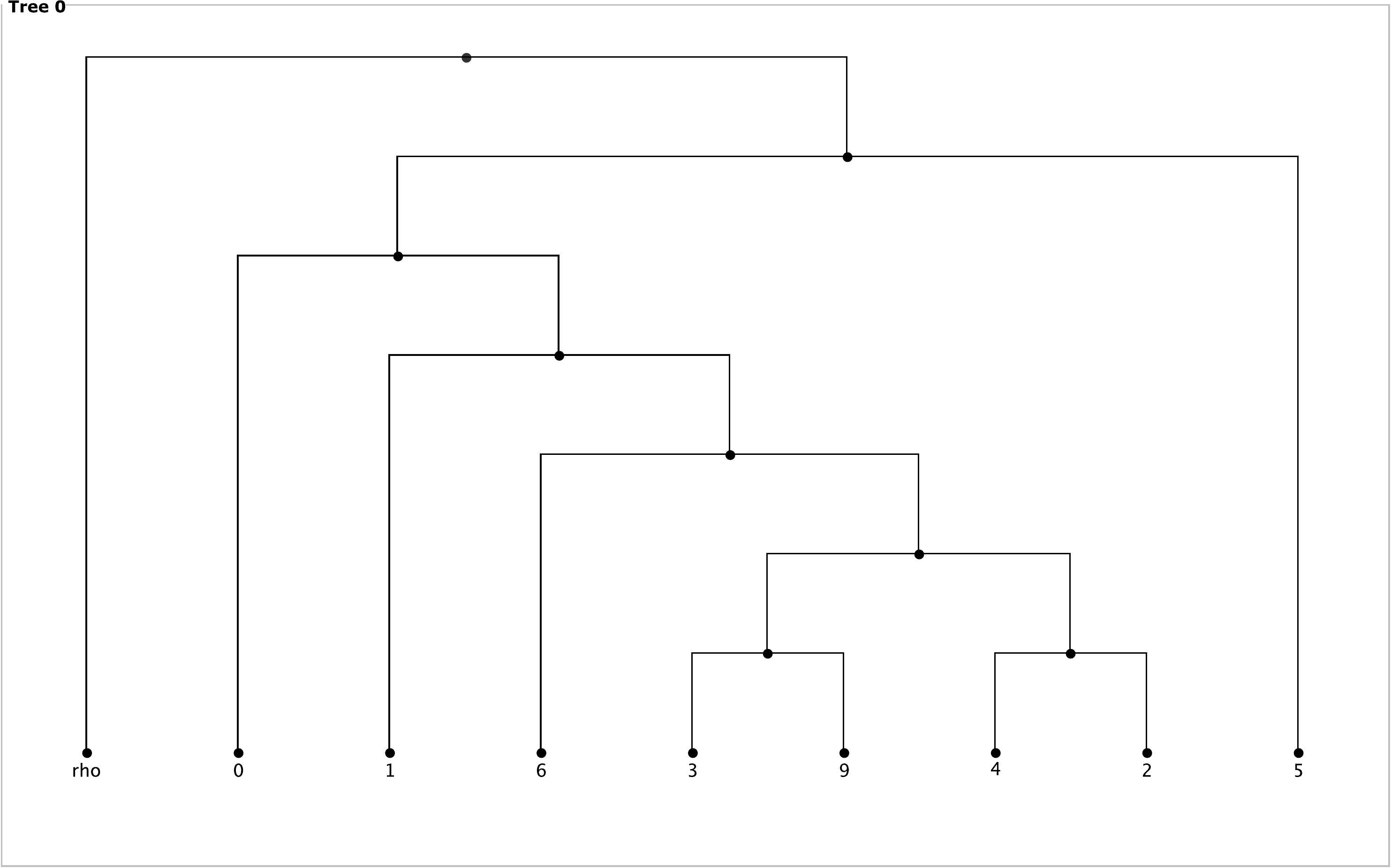}
&
\includegraphics[width = 8cm]{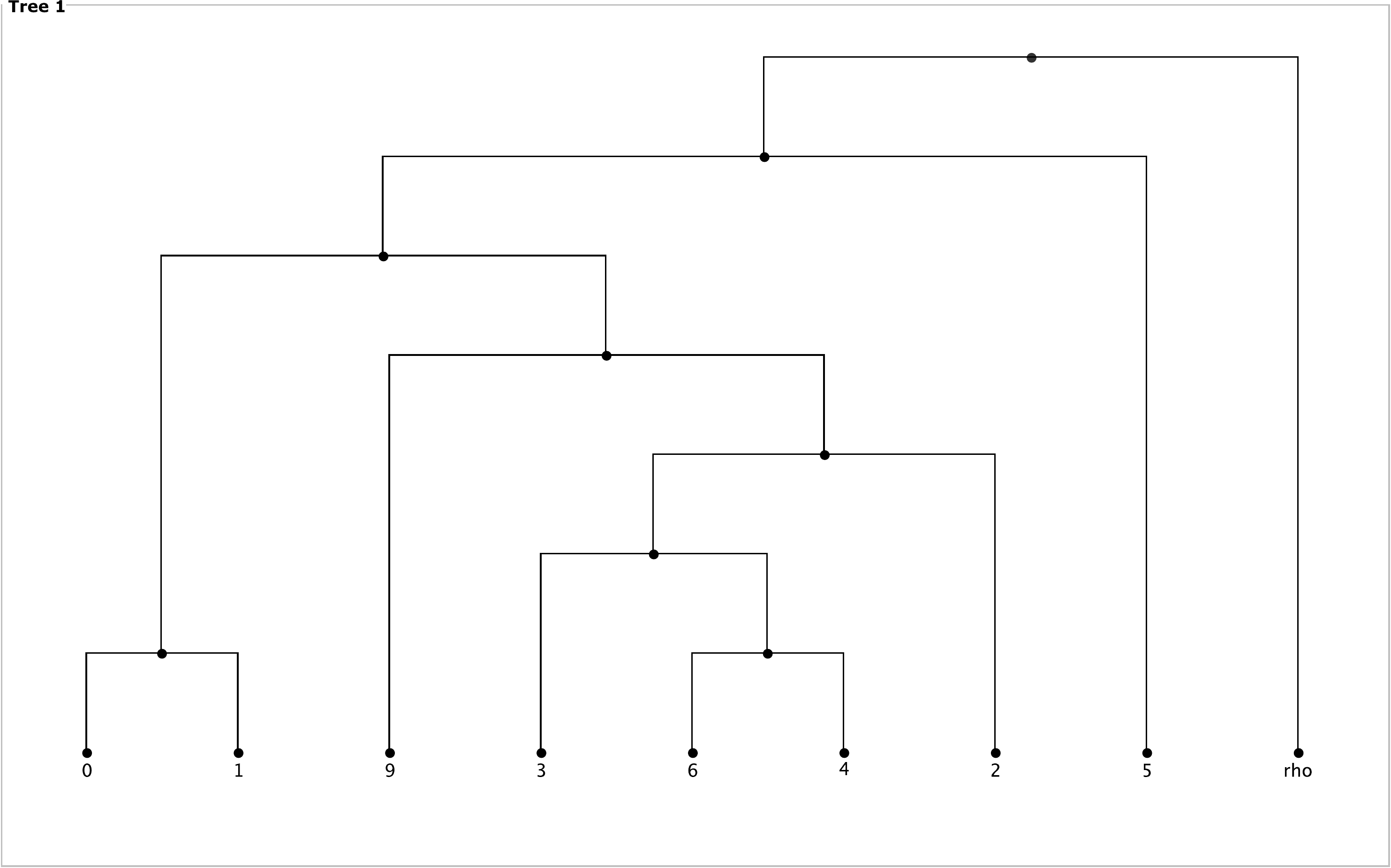}
\\
\includegraphics[width = 8cm]{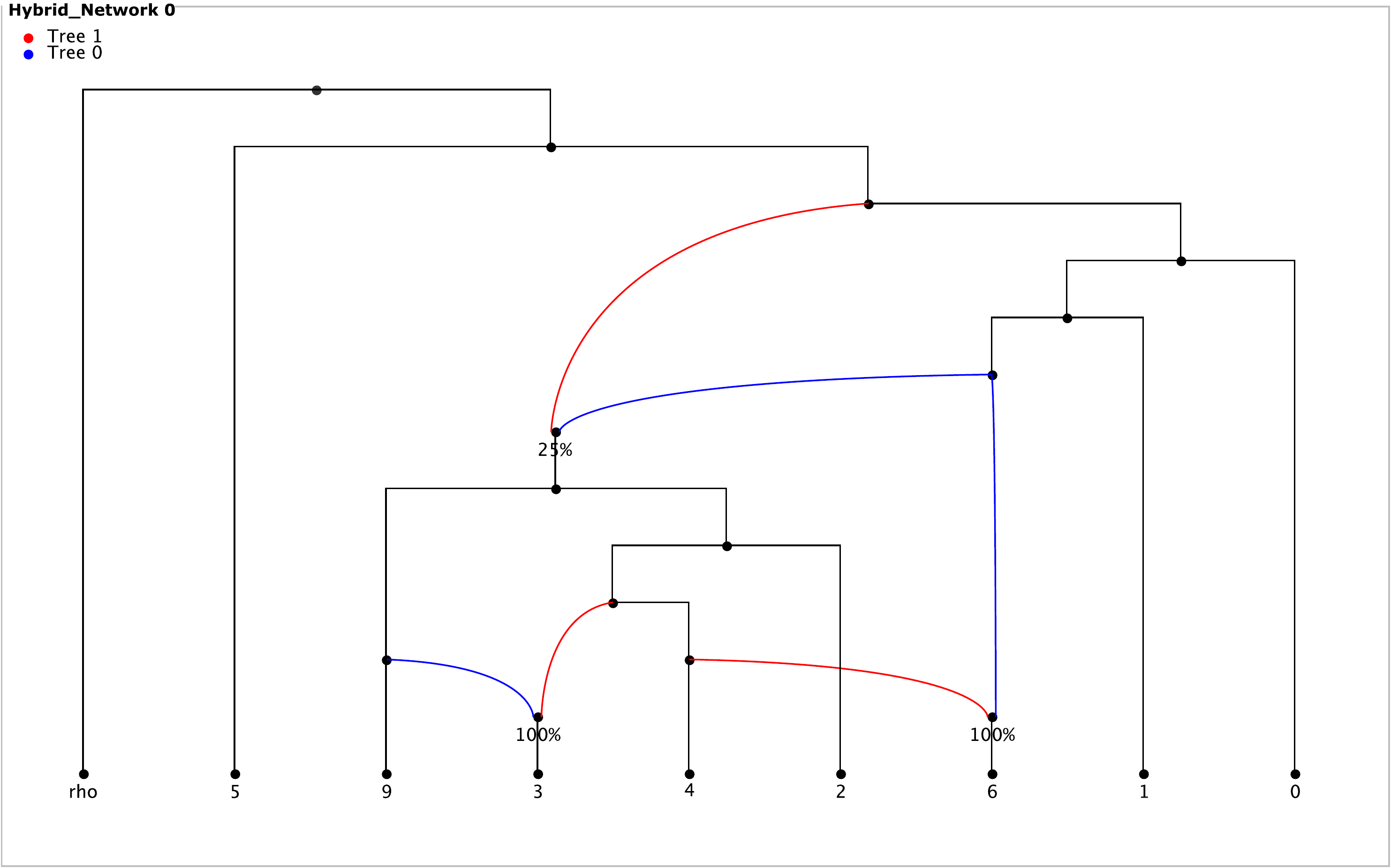}
&
\includegraphics[width = 8cm]{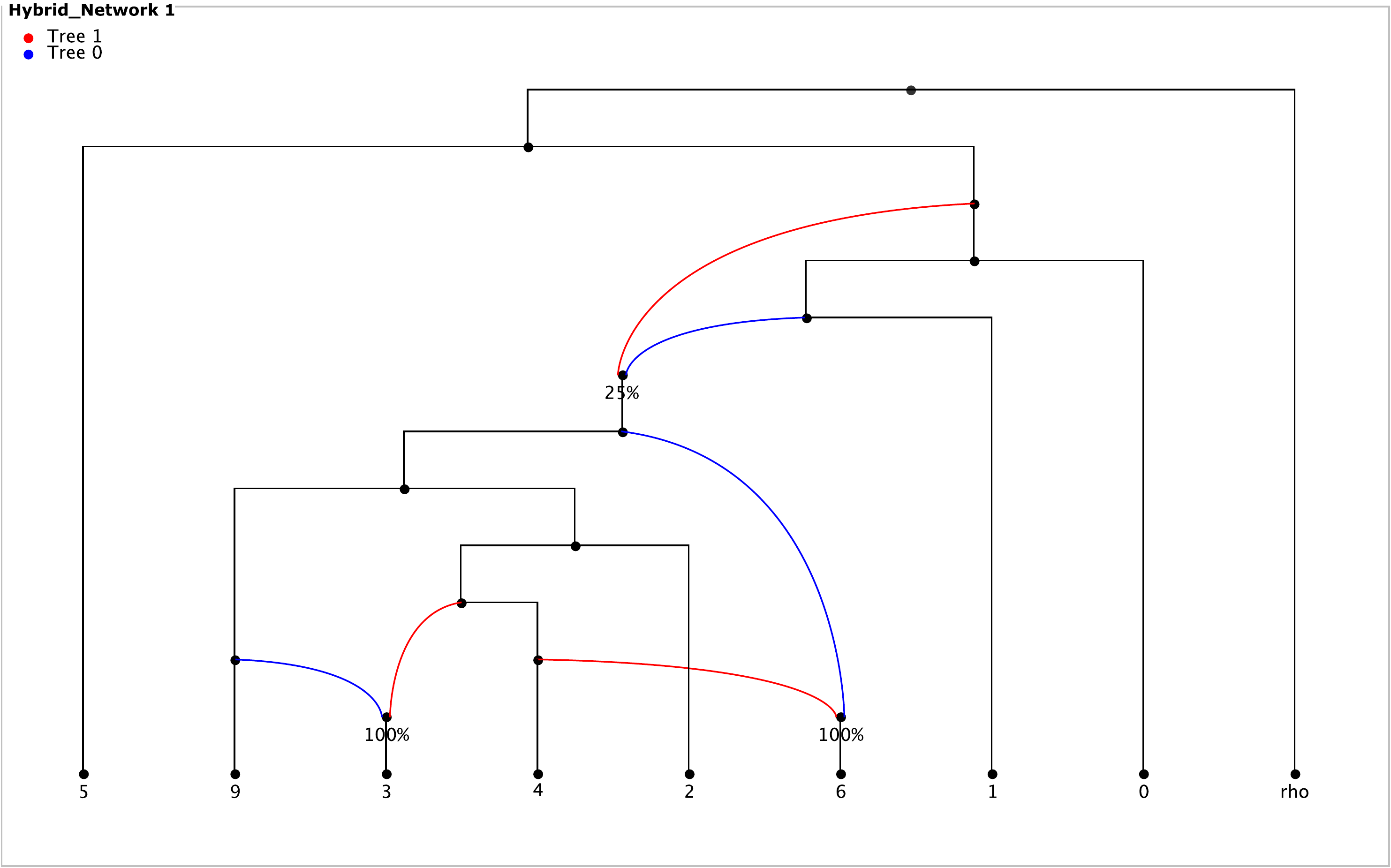}
\\
\end{tabular}
\caption[An illustration of the application of linking patterns - Part 1] {(\textbf{Top}) Two rooted binary phylogenetic $\cX$-trees sharing the common cluster $\{2,3,4,6,9\}$. (\textbf{Bottom}) Two networks displaying both trees from the top  where blue edges refer to the left and red edges to the right tree. The left network can be obtained from the right one by first applying a linking-pattern of Type B, attaching the blue in-edge referring to node $6$ to the hybridization node labeled by $25\%$, and then by applying a modification of Type B.} 
\label{fig-shiftUp}
\end{figure}

\clearpage
\section{Cluster reduction on multiple trees}
\label{sec-cr}

In the following, we will give a formal proof showing that the cluster reduction is safe for a set $\cT$ of \emph{multiple} rooted binary phylogenetic $\cX$-trees as noted in Theorem~\ref{21-th-cr}.

\begin{theorem}
Given a set of rooted binary phylogenetic $\cX$-trees $\cT$ all containing a common cluster $A\subset \cX$, then, $h(\cT)=h(\cT|_A)+h(\cT_a)$.
\label{21-th-cr}
\end{theorem}

\subsection{Related work}

In general, there are two important works dealing with the cluster reduction of rooted binary phylogenetic $\cX$-trees.\\

\textbf{Baroni, 2006.} The well-known work of Baroni \emph{et al.}~\cite{Baroni-2006} contains a proof showing that the hybridization number of \emph{two} rooted binary phylogenetic $\cX$-trees can be computed by simply summing up the hybridization numbers of its common clusters. More precisely, given two rooted binary phylogenetic $\cX$-trees $T$ and $T'$ containing a common cluster $A\subset \cX$, then $h(T,T')=h(T|_A,T'|_A)+h(T_a,T'_a)$, where $T_a$ and $T'_a$ refers to the respective input tree in which the common cluster has been replaced by a new taxon $a$.\\

\textbf{Linz, 2008.} A more general proof, showing that a similar fact also holds for more than two rooted binary phylogenetic $\cX$-trees, is given in the PhD thesis of Linz~\cite[Theorem 2.5]{Linz-2008}. This proof, however, in contrast to our definition of the hybridization number $h$ (cf.~Eq.~\ref{04-eq-hNumBin}), is based on a different definition, denoted by $h'$, only considering the total number of hybridization nodes of a network. More precisely, given a set of rooted binary phylogenetic $\cX$-trees $\cT$, then $h'(\cT)=h'(\cT|_A)+h'(\cT_a)$ with $$h'(\cT)=\min\{h'(N):\text{N is a hybridization network displaying }\cT\},$$ where $h'(N)$ just counts the number of hybridization nodes in $N$. This means, in particular, that, in contrast to the reticulation number $r(N)$ as defined here~(cf.~Eq.~\ref{04-eq-retNumBin}), $h'(N)$ does not take the number of edges that are directed into a hybridization node into account. Consequently, the two values $r(N)$ and $h'(N)$ differ, if the network $N$ provides a hybridization node with in-degree larger than two.

\subsection{Further definitions}

In the following, we will first give some further definitions that are crucial for establishing Theorem~\ref{21-th-cr}.\\

\textbf{Hybridization networks.} Given a hybridization network $N$ on $\cX$ and a subset $E'$ of reticulation edges in $N$, then, by writing $N-E'$ we denote the network that is obtained from $N$ by first deleting $E'$ and then by suppressing each node of both in- and out-degree $1$.\\

\textbf{Restricted pendant subtrees.} Let $E'$ be an edge set referring to a rooted binary phylogenetic $\cX'$-tree $T'$ that is displayed in a hybridization network $N$. Then for a path $P$ connecting two nodes both contained in $N|_{E',\cX'}$, we denote by $\cR_N(P,E',\cX')$ the set of \emph{non-empty restricted pendant subtrees} of each node lying on $P$. More precisely, each subtree $R_i$ in $\cR_N(P,E',\cX')$ refers to a non-empty subgraph of $N|_{E',\cX'}$ with root $v\not\in P$, which is connected through an edge to a node $w\in P$, such that $R_i$ equals $\overline{N|_{E',\cX'}(v)}$ (cf.~Fig.~\ref{21-fig-embTree}(b)).

\begin{figure}[t]
\centering
\includegraphics[scale=0.9]{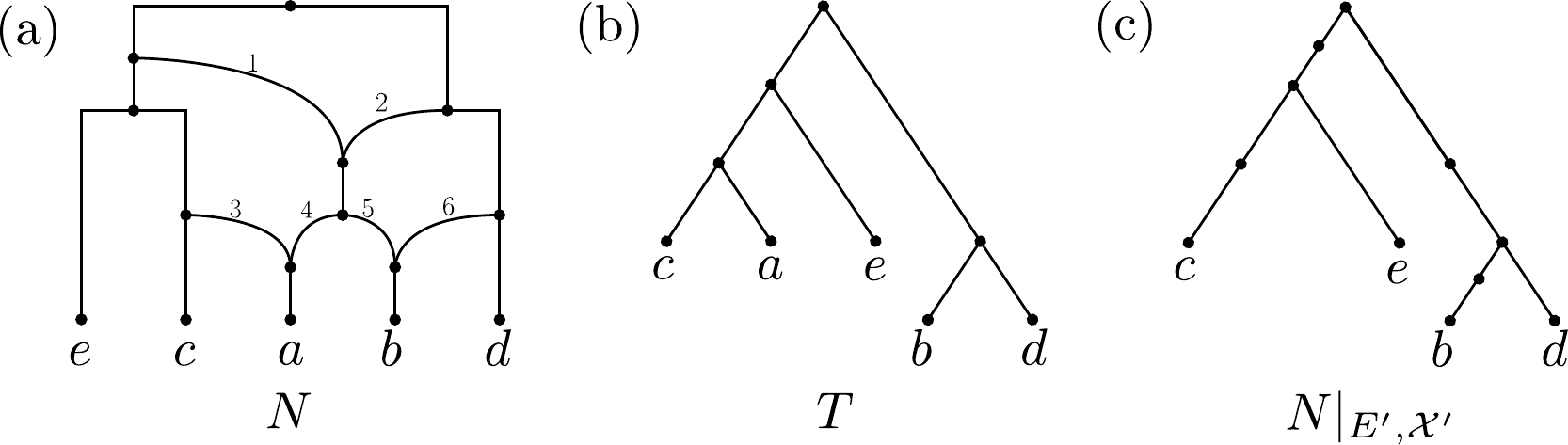}
\caption[Definition regarding restricted pendant subtrees]{\textbf{(a)} A hybridization network $N$ with taxa set $\cX=\{a,b,c,d,e\}$ whose reticulation edges are consecutively numbered. \textbf{(b)}~The restricted network $N|_{E',\cX'}$ with $E'=\{3,6,1\}$ and $\cX'=\{b,c,d,e\}$ still containing nodes of both in- and out-degree~$1$. Let $P$ be the path connecting both nodes $v$ and $w$ in $N|_{E',\cX'}$, then $\cR_N(P,E',\cX')$ consists of the four non-empty restricted pendant subtrees $(c)$, $(e)$, $(b)$, and~$(d)$.} 
\label{21-fig-embTree}
\end{figure}

\subsection{Proof of Theorem~\ref{21-th-cr}}
\label{sec-pcr}

\begin{proof}
As defined in Theorem \ref{21-th-cr}, we have that $A\subset\cX$. Now, in a first step, we show that 
\begin{equation}\label{21-eq-cr1}
h(\cT)\le h(\cT|_A)+h(\cT_a)
\end{equation}
by contradiction. Let $N$ be a hybridization network displaying $\cT$ with minimum hybridization number $h(\cT)$, and let $N_A$ and $N_a$ be a hybridization network displaying $\cT|_A$ with minimum hybridization number $h(\cT|_A)$ and $\cT_a$ with minimum hybridization number $h(\cT_a)$, respectively. Moreover, let $N_{A,a}$ be the network that is obtained from $N_a$ by replacing taxon $a$ through $N_A$. This is done, in particular, by first attaching each in-going edge of the leaf $v_a$ labeled by taxon $a$ to the root of $N_A$ and then by removing label $a$ from $v_a$. Now, if $h(\cT)>h(\cT|_A)+h(\cT_a)$ holds, then simultaneously $r(N)>r(N_{A,a})$ must hold which is a contradiction to the choice of $N$. $\lightning$ 

Next, we will show that 
\begin{equation}\label{21-eq-cr2}
h(\cT)\ge h(\cT|_A)+h(\cT_a)
\end{equation}
by discussing several cases.

In a first step, however, we have to establish a new lemma that is crucial for proving this inequation. Given a hybridization network $N$ containing a reticulation edge $e$, we say that \emph{$e$ can be compensated} if $N-\{e\}$ still displays $\cT$. This is the case if and only if $N$ displays the scenario as described in Lemma~\ref{21-lem-comp} (cf.~Fig.~\ref{21-fig-comp1}). 


\begin{lemma}
Given a hybridization network $N$ displaying a set $\cT$ of rooted binary phylogenetic $\cX$-trees. Then, a reticulation edge $e$ in $N$ can be compensated if and only if for each tree $T_i$ in $\cT$, whose referring edge set $E_i$ contains $e$, there exists another edge set $E_i'\ne E_i$ such that the following condition is satisfied. There exist two node-disjoint paths $P$ and $P'$ both connecting two nodes $u$ and $w$ with $e\in P$ and $e\not\in P'$ such that $\cR_N(P,E_i,\cX)=\cR_N(P',E_i',\cX)$.
\label{21-lem-comp}
\end{lemma}

\begin{proof}
'$\Longleftarrow$': For each tree $T_i$, whose referring edge set $E_i$ contains $e$, let $\hat E_i=E_i\setminus\{e\} \cup E_H(P')$, where $E_H(P')$ denotes the set of reticulation edges in $P'$. Then, since $\cR_N(P,E_i,\cX)=\cR_N(P',E_i',\cX)$, $\hat E_i$ refers to $T_i$ and, thus, $N-\{e\}$ still displays $T_i$.

'$\Longrightarrow$': If $e$ can be compensated, this implies that $N-\{e\}$ still displays $\cT$. This means, in particular, that for each edge set $E_i$ containing $e$ and referring to a tree $T_i$ in $\cT$, there has to exist a further edge set $E_i'\ne E_i$ not containing $e$ but still referring to $T_i$. Now, based on the two restricted networks $N|_{E_i,\cX}$ and $N|_{E_i',\cX}$, we can define two particular paths $P$ and $P'$. 

Let $e'$ be an edge of $N$ satisfying the following two conditions. First the source node $u$ of $e'$ is part of $N|_{E_i,\cX}$ as well as of $N|_{E_i',\cX}$ and its target node is only part of $N|_{E_i',\cX}$ but not of $N|_{E_i,\cX}$. Second, there is no other edge in $N$ fulfilling this property and is closer to the root. Similarly, let $e''$ be an edge of $N$ satisfying the following two conditions. First, the target node $w$ of $e''$ is part of $N|_{E_i,\cX}$ as well as of $N|_{E_i',\cX}$ and its source node is only part of $N|_{E_i',\cX}$ but not of $N|_{E_i,\cX}$. Second, there is no other edge in $N$ fulfilling this property and is closer to the root. 

Then, there are two specific paths in $N$ running from $u$ to $w$; one being part of $N|_{E_i,\cX}$ (and, thus, containing~$e$), denoted by $P=(u,a_1,\dots,a_k,w)$, and the other one being part of $N|_{E_i',\cX}$ (and, thus, not containing~$e$), denoted by $P'=(u,b_1,\dots,b_{k'},w)$ (cf.~Fig.~\ref{21-fig-comp1}). Moreover, as both edges $e'$ and $e''$ are chosen such there exist no other edges fulfilling the respective properties and are closer to the root, $a_i\ne b_j$ for each node $a_i\in\{a_1,\dots,a_k\}$ and each node $b_j\in(b_1,\dots,b_{k'})$. Additionally, since $E_i$ and $E_i'$ both refer to $T_i$, $\cR_N(P,E_i,\cX)=\cR_N(P',E_i',\cX)$, which finally establishes Theorem~\ref{21-th-cr}. 
\end{proof}


\begin{figure}[t]
\centering
\includegraphics[scale=1.15]{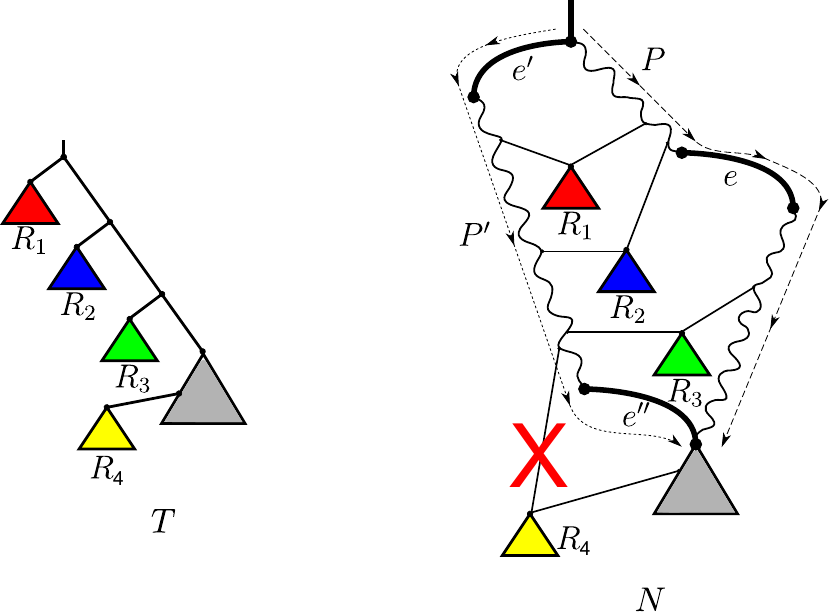}
\caption[Illustration of a scenario compensating an edge]{Illustration of a scenario compensating the reticulation edge $e$ regarding the embedding of $T$ in $N$ (see proof of Lemma~\ref{21-lem-comp}). Note that $e$ could not be compensated if $R_4$ would be a pendant subtree of $P'$.} 
\label{21-fig-comp1}
\end{figure}

Now, let each network $N$, $N_A$, $N_a$, and $N_{A,a}$ be as defined above. Moreover, let $E_A$ and $E_a$ be a specific subset of reticulation edges of those contained in the subnetwork corresponding to $N_A$ and $N_a$, respectively, satisfying the following condition. For each edge $e$ in $E_A$ (resp. $E_a$), this edge can be reattached to a specific edge of the subnetwork corresponding to $N_a$ (resp. $N_A$), so that the resulting network $N_{A,a}'$ still displays $\cT$ (cf.~Fig.~\ref{21-fig-shift}). Additionally, let $N_a'$ and $N_A'$ be the two subgraphs in $N_{A,a}'$ consisting of each element in $N_a$ and $N_A$, respectively. Now, if $h(\cT)<h(\cT|_A)+h(\cT_a)$ holds, based on $E_A$ and $E_a$, we have to consider the following four cases (cf.~Fig.~\ref{21-fig-shift}).

\begin{itemize}
\item[(i)] Let $E_A=\emptyset$ and $E_a=\emptyset$. There exists a set of reticulation edges $E_{A,a}'\neq \emptyset$ in $N_A$ or $N_a$ that can be compensated.
\item[(ii)] Let $E_A\neq\emptyset$ and $E_a=\emptyset$. There exists a set of reticulation edges $E_{A,a}'\neq \emptyset$ in $N_{A,a}'$ that can be compensated. 
\item[(iii)] Let $E_A=\emptyset$ and $E_a\neq\emptyset$. There exists a set of reticulation edges $E_{A,a}'\neq \emptyset$ in $N_{A,a}'$ that can be compensated.
\item[(iv)] Let $E_A\neq\emptyset$ and $E_a\neq\emptyset$. There exists a set of reticulation edges $E_{A,a}'\neq \emptyset$ in $N_{A,a}'$ that can be compensated.
\end{itemize} 

\begin{figure}[t]
\centering
\includegraphics[scale=1.15]{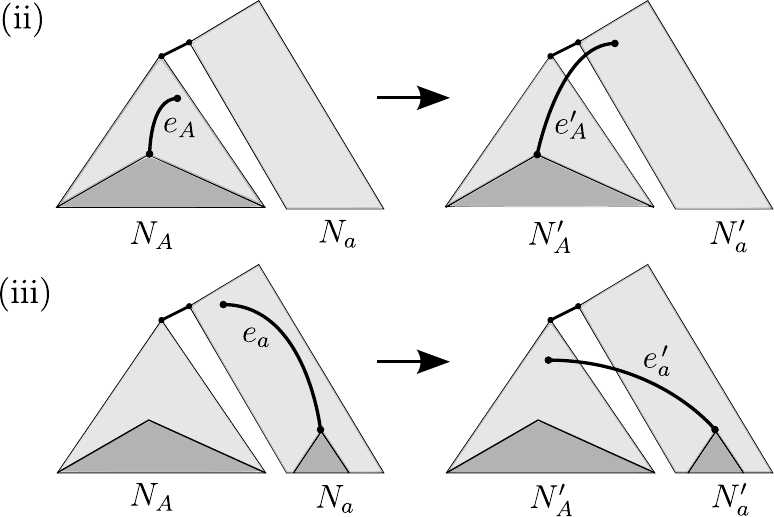}
\caption[Illustration of Case (ii) and (iii) regarding Lemma~\ref{21-lem-comp}]{Illustration of the scenario referring to Case (ii) and Case (iii).} 
\label{21-fig-shift}
\end{figure}

In the following, we will show that each scenario, which is described by one of the four cases, cannot occur due to certain circumstances.\\

\textbf{Case (i).} In this case, either $r(N_A)\ne h(\cT|_A)$ or $r(N_a)\ne h(\cT_a)$, which is a contradiction to the choice of $N_A$ or $N_a$, respectively. $\lightning$\\

\textbf{Case (ii).} Let $N_{A,a}'$ be the network that is obtained from $N_{A,a}$ by reattaching the source nodes of each edge in $E_A$ to the subnetwork corresponding to $N_a$ such that $N_{A,a}'$ still displays $\cT$. Now, first notice that from those shifted edges there does not arise a new path whose start- and end-node both lie in $N_a'$. As a consequence, due to Lemma~\ref{21-lem-comp}, each edge of $N_a'$ that could be compensated, could be also compensated in the original network $N_{A,a}$, which is a contradiction to the choice of both networks $N_a$ and $N_A$. $\lightning$ 

The same argument holds for the subnetwork corresponding to $N_A$ and, thus, in this case the network $N_{A,a}'$ cannot contain any reticulation edges that can be compensated.\\ 

\textbf{Case (iii).} The argumentation regarding this case equals the one of Case (ii).\\

\textbf{Case (iv).} Again, let $E_A'$ and $E_a'$ be the set of edges in $N_{A,a}'$ whose source nodes have been reattached to the subnetwork corresponding $N_a$ and $N_A$, respectively. Now, there additionally exist three out of four sub-cases that have to be considered here (cf.~Fig.~\ref{21-fig-shift-2}).

\begin{itemize}
\item[(iv.i)] Neither a source node of an edge in $E_a'$ is contained in a subnetwork rooted at a target node of an edge in $E_A'$ nor a source node of an edge in $E_A'$ is contained in a subnetwork rooted at a target node of an edge in $E_a'$.
\item[(iv.ii)] There exists a source node of an edge $e_a'$ in $E_a'$ that is contained in a subnetwork rooted at the target node of an edge $e_A'$ in $E_A'$.
\item[(iv.iii)] There exists a source node of an edge $e_A'$ in $E_A'$ that is contained in a subnetwork rooted at the target node of an edge $e_a'$ in $E_a'$.
\item[(iv.iv)] There exists a source node of an edge $e_A'$ in $E_A'$ that is contained in a subnetwork rooted at the target node of an edge $e_a'$ in $E_a'$ and, simultaneously, there exists a source node of an edge $e_A'$ in $E_A'$ that is contained in a subnetwork rooted at the target node of an edge $e_a'$ in $E_a'$. This directly implies that the graph contains a directed cycle and, thus, does not apply to the definition of hybridization networks. Consequently, this case has not to be considered here. 
\end{itemize}

\begin{figure}[t]
\centering
\includegraphics[scale=1.15]{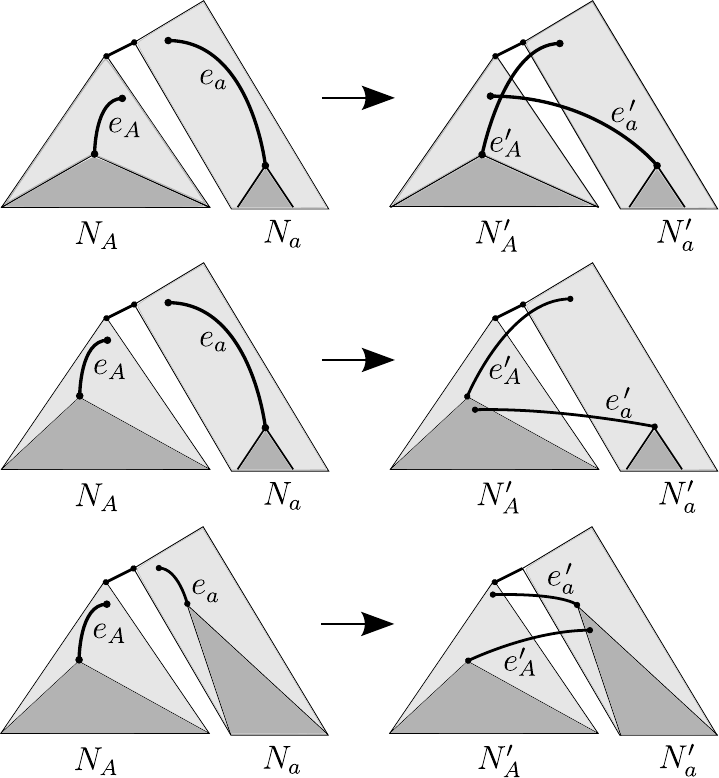}
\caption[Illustration of Case (iv.i), (iv.ii), and (iv.iii) regarding Lemma~\ref{21-lem-comp}]{Illustration of the scenario referring to Case (iv.i), Case (iv.ii), and Case (iv.iii).} 
\label{21-fig-shift-2}
\end{figure}

\textbf{Case (iv.i).} Again, similar to Case (ii), there does not arise a new path whose start- and end-node both lie in the subnetwork corresponding to $N_a$ and $N_A$, respectively. Thus, each edge that is contained in this part of the network and could be compensated, could be also compensated in the original network $N_{A,a}$ which is a contradiction to the choice of both networks $N_a$ and $N_A$. $\lightning$\\

\textbf{Case (iv.ii).} In this certain case there exists a path leading from a target node of $e_A'$ in $E_A'$ back to $N_a'$ (cf.~Fig.~\ref{21-fig-comp2}). Thus, potentially, there could exist a reticulation edge $e$ in $N_a'$ such that $N_{A,a}'-\{e\}$ still displays $\cT$. More precisely, this would be the case if $e_A'$ and $e_a'$ could compensate a deletion of $e$. 

Now, let $E_i$ be an edge set referring to an input tree $T_i$ and let $P$ be the path of $N_{A,a}'|_{E_i,\cX}$ leading from the source node of $e_A'$ to the target node $e_a'$. Moreover, without loss of generality, we assume that there does not exist a further edge set referring to another input tree $T_j$ with $j\ne i$ containing $e$. Now, if there would exist an edge set $E_i'$ with $e\not\in E_i'$ referring to $T_i$, this would automatically imply that $e$ could be compensated. 

If $e$ is not part of such a path $P$, $e$ cannot be compensated by the two shifted edges $e_A'$ and $e_a'$. Otherwise, let $\cR_{N_{A,a}'}(P,E_i,\cX)=(R_0,\dots,R_k)$ be the ordered set of non-empty pendant subtrees of each node lying on $P$ in which the first restricted subtree $R_0$ corresponds to $\overline{N_A'|_{\{E_i,\cX\}}}$. Now, only if there exists a path $P'$ leading from the target node of $e_A'$ to the target node of $e_a'$ such that $\cR_{N_{A,a}'}(P,E_i,\cX)$ equals $\cR_{N_{A,a}'}(P',E_i,\cX)$, $e$ could be compensated by using $e_A'$ and $e_a'$. However, as $A$ is a cluster of $T_i$, in this case $R_1$ to $R_{k-1}$ may not exist, meaning that $\cR_{N_{A,a}'}(P',E_i,\cX)$ could only consist of the two elements $R_0$ and $R_k$. Thus, if $e$ could be compensated, this would directly imply that $e$ could be also compensated in $N_a$ by $e_a$, which is a contradiction to the choice of $N_a$. $\lightning$ \\

\begin{figure}[tb]
\centering
\includegraphics[width = 8cm]{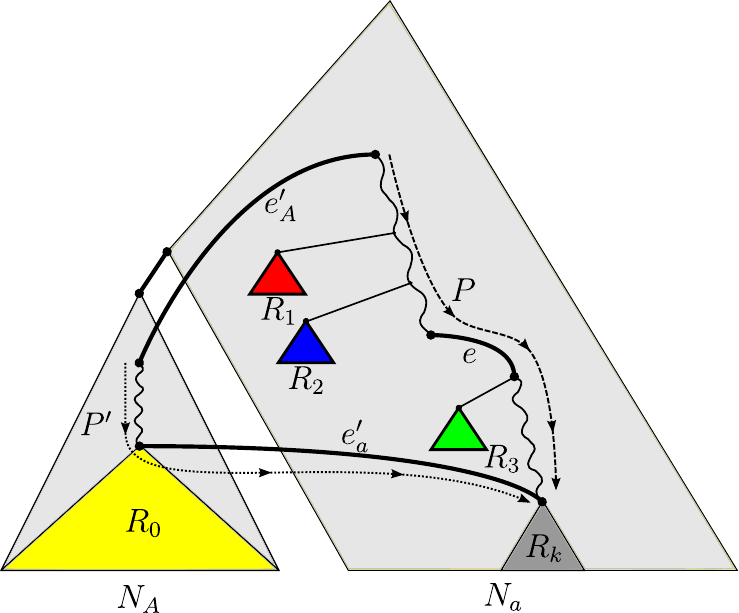}
\caption[Illustration of Case (iv.i) regarding Lemma~\ref{21-lem-comp}]{An illustration of the scenario concerning Case (iv.ii). } 
\label{21-fig-comp2}
\end{figure}

\textbf{Case (iv.iii).} The argumentation regarding this case equals the one of Case (iv.ii).\\

Finally, combining both Inequations \ref{21-eq-cr1} and \ref{21-eq-cr2} completes the proof of Theorem~\ref{21-th-cr}.
\end{proof}

\clearpage
\section{Discussion}
\label{sec-dis}
To analyze hybridization events, it is of high interest to compute all hybridization networks, since the more frequently an event occurs in all those networks the more likely it may be part of the true underlying evolutionary scenario. In this work, we first presented the algorithm \textsc{allHNetworks} calculating all relevant networks for an input consisting of multiple rooted binary phylogenetic $\cX$-trees and then established its correctness by a detailed formal proof. Notice that a major finding of this work is that for computing such networks it suffices to use the concept of maximum acyclic agreement forests.

The stated theoretical worst-case runtime of the algorithm \textsc{allHNetworks} reveals, however, that the number of relevant networks growths in a strong exponential manner in terms of the number and the size of the input trees which obviously complicates its application to real biological problems. As a consequence, it is very important to implement the algorithm in an efficient way (e.g., parallelizing particular substeps), which is addressed in another paper of Albrecht \cite{Albrecht2015}. Moreover, in an algorithmic point of view, in order to improve the practical runtime, one can apply a cluster reduction to the input trees. We have demonstrated, however, that when separating those input trees into several clusters, in order to obtain all relevant networks, one still has to spend some work in attaching back different networks separately computed for each of those clusters. However, if one is only interested in the hybridization number, this post-processing step is not necessary as proven in Section~\ref{sec-cr}.


\bibliographystyle{natbib} 


\end{document}